\documentclass[11pt]{article}

\usepackage{authblk}
\usepackage{graphicx}
\usepackage{enumerate}
\usepackage{mathrsfs,epsfig,epstopdf, amsfonts,amsthm}
\usepackage{bbm}
\usepackage{color}
\usepackage[numbers,sort&compress]{natbib}% Citation support using natbib.sty
\usepackage[ruled]{algorithm2e}
\usepackage{multirow}
\usepackage{amssymb} % add a triangle above an equation
\usepackage{hyperref} % add url
\usepackage{amsmath}  % add brace above/below the equation \underbrace \overbrace
\usepackage{adjustbox}      % adjust the width of three part table
\usepackage{threeparttable}

\newcommand{\algorithmfootnote}[2][\footnotesize]{%
  \let\old@algocf@finish\@algocf@finish% Store algorithm finish macro
  \def\@algocf@finish{\old@algocf@finish% Update finish macro to insert "footnote"
    \leavevmode\rlap{\begin{minipage}{\linewidth}
    #1#2
    \end{minipage}}%
}}%

\theoremstyle{plain}
\newtheorem{theorem}{Theorem}[section]

\theoremstyle{definition}

\theoremstyle{remark}

\newtheorem{prop}[theorem]{Proposition}
\newcommand{\myalpha}{\boldsymbol{\alpha}}
\newcommand{\mybeta}{\boldsymbol{\beta}}
\newcommand{\mymu}{\boldsymbol{\mu}}
\newcommand{\mytheta}{\boldsymbol{\theta}}

\newcommand{\myTheta}{\boldsymbol{\Theta}}

\newcommand{\myeta}{\boldsymbol{\eta}}
\newcommand{\mygamma}{\boldsymbol{\gamma}}

\newcommand{\myB}{\mathbf{B}}

\newcommand{\myW}{\mathbf{W}}
\newcommand{\myX}{\mathbf{X}}
\newcommand{\myY}{\mathbf{Y}}
\newcommand{\myZ}{\mathbf{Z}}

\newcommand{\myh}{\mathbf{h}}
\newcommand{\myn}{\mathbf{n}}
\newcommand{\myx}{\mathbf{x}}
\newcommand{\myy}{\mathbf{y}}
\newcommand{\myz}{\mathbf{z}}
\newcommand{\myr}{\mathbf{r}}

\SetKwBlock{OuterIteration}{Outer Iteration:}{end}
\SetKwBlock{InnerIteration}{Inner Iteration:}{end}

\providecommand{\keywords}[1]
{
  \small	
  \textbf{\textit{Keywords---}} #1
}

\graphicspath{{figure/}}

\begin{document}

\title{Hot-spots Detection in Count Data by Poisson Assisted Smooth Sparse Tensor Decomposition}

\author[1]{Yujie Zhao}
\author[2]{Xiaoming Huo}
\author[2]{Yajun Mei}
\affil[1]{Biostatistics and Research Decision Sciences Department, Merck \& Co., Inc, PA, USA}
\affil[2]{School of Industrial and Systems Engineering, Georgia Institute of Technology, Atlanta, GA, USA}

\maketitle

\begin{abstract}
Count data occur widely in many bio-surveillance and healthcare applications, e.g., the numbers of new patients of different types of infectious diseases from different cities/counties/states repeatedly over time, say, daily/weekly/monthly. 
For this type of count data, one important task is the quick detection and localization of hot-spots in terms of unusual infectious rates so that we can respond appropriately.
In this paper, we develop a method called \textit{Poisson assisted Smooth Sparse Tensor Decomposition (PoSSTenD)}, which not only detect when hot-spots occur but also localize where hot-spots occur.
The main idea of our proposed PoSSTenD method is articulated as follows.
First, we represent the observed count data as a three-dimensional tensor including 
  (1) a spatial dimension for location patterns, e.g., different cities/countries/states;
  (2) a temporal domain for time patterns, e.g., daily/weekly/monthly;
  (3) a categorical dimension for different types of data sources, e.g., different types of diseases.
Second, we fit this tensor into a Poisson regression model, and then we further decompose the infectious rate into two components: smooth global trend and local hot-spots.
Third, we detect when hot-spots occur by building a cumulative sum (CUSUM) control chart and localize where hot-spots occur by their LASSO-type sparse estimation.
The usefulness of our proposed methodology is validated through numerical simulation studies and a real-world dataset, which records the annual number of 10 different infectious diseases from $1993$ to $2018$ for $49$ mainland states in the United States.
\end{abstract}


%---------------------------------------------------%
%                 Introduction                      %
%---------------------------------------------------%
\section{Introduction}
  
Count data are frequently seen in many real-world applications in fields such as biosurveillance \cite[see][]{chen2021natural}, epidemiology and sociology \cite[see][]{zhao2019rapid}.
  Usually, the count data are collected from multiple data sources and from many spatial locations repeatedly over time, say, daily, monthly, or annually.
  For instance, the daily number of people infected by some types of diseases in 50 states in the United States. 
  We call this type of data as \textit{multivariate spatio-temporal count data}, which has three domains: 
  (1) the spatial domain recording the locations, 
  (2) the temporal domain recording the time, and 
  (3) the categorical domain recording the different types of data sources.
When such count data are non-stationary over time, there are two standard research areas: one is model fitting and prediction, and the other is the detection of global/system-wise changes, seeing the classical statistical process control (SPC) or sequential change-point detection literature. 
  
  In this paper, we investigate a new research area of multivariate spatio-temporal count data by focusing on the detection and localization of \textit{hot-spots}.
  Here, the hot-spots are defined as the anomalies/outliers in term of infectious rates occurring in the temporal, spatial, or  categorical domains. There are a couple of subtasks: the \textit{hot-spots detection} aims to detect when hot-spots occur (in the temporal domain) and the \textit{hot-spots localization} aims to localize where hot-spots occur (in the spatial and categorical domain).
  
  It is informative to point out that there are two kinds of anomaly detection for multivariate data: one is global-level change (e.g., the \textit{first-order} changes), and the other is local-level hot-spots (e.g., \textit{second-order} changes). 
  The former is the focus of the classical SPC or sequential change-point detection methods, whereas the latter is the main interest of this paper. 
  Here we assume that local-level hot-spots have 
  (1) spatial sparsity, i.e., the local changes are sparse in the spatial and categorical domains; 
  (2) temporal consistency, i.e., the local changes last for a while once they occur. 
  
  In the remainder of this section, we present the literature review in Section \ref{sec: Literature Review} and articulate our proposed method and our contribution in Section \ref{sec: method overview and contribution}.

%---------------------------------------------------%
%             Literature Review                     %
%---------------------------------------------------%
\subsection{Literature Review}
\label{sec: Literature Review}

  In this subsection, we review the existing literature on hot-spots detection based on three types of count data: univariate, multivariate, and multivariate spatio-temporal. 
  
  For univariate count data, i.e., when we observe the data of the form $\{y_t\}_{t = 1, \ldots, T}$ with $y_t \in \mathbbm N$ (a set of non-negative integers) denoting the observed count data at time $t$, there are mainly four methods to detect hot-spots:
  (1) time series based method 
  \cite[see][]{killick2014changepoint};
  %\cite[see][]{erdman2007bcp, killick2014changepoint, ross2015parametric, zeileis2015package};
  (2) control charts based method by using the likelihood functions, such as the cumulative sum (CUSUM) control chart 
  \cite[see][]{han2010comparison} 
  %\cite[see][]{perry2006estimating, han2010comparison, mei2011early} 
  and exponentially weighted moving average (EWMA) control chart 
  \cite[see][]{han2010comparison};
  %\cite[see][]{han2010comparison, shen2013monitoring};
  (3) scan statistics based method \cite[see][]{kulldorff2009scan}, including both the frequentist scan statistics
  \cite[see][]{kulldorff1995spatial} %\cite[see][]{kulldorff1995spatial, kulldorff2001prospective, almeida2018multiple} 
  and Bayesian scan statistics \cite[see][]{neill2005bayesian};
  (4) Poisson process-based method that is often combined with Bayesian methods \cite[see][]{ihler2006adaptive}.
  
  For multivariate count data, i.e., when we observe the data of the form $\{y_{it}\}_{i = 1,\ldots, n, t = 1, \ldots, T}$ with $y_{it} \in  \mathbbm N$, there are mainly three methods widely-used hot-spots detection:
  (1) control chart based method, including CUSUM or EWMA control charts  \cite[see][]{liu2019scalable};
  (2) Poisson process based method \cite[see][]{turcotte2016poisson};
  (3) hidden Markov model (HMM) based method \cite[see][]{conesa2015bayesian}.
  
  For multivariate spatio-temporal count data, i.e., when the observed data are of the form $y_{ijt} \in \mathbbm N$ that denotes the observed count data over the $i$-th spatial location for the $j$-th type of disease at time $t$ with $i = 1,\ldots, n_1, j = 1, \ldots, n_2, t = 1, \ldots, T$, research on hot-spots detection is rather limited, although there are three types of closely related methods developed for other purposes or other contexts:
  (1) scan statistics for anomalous clusters,
  (2) tensor decomposition for Gaussian data, and
  (3) change point detection for global changes. 
  Since our paper deals with multivariate spatio-temporal count data, below we will provide a more detailed review of these three methods. 
  
  For the scan statistics based method, it was first developed in the 1960s \cite[see][]{naus1963clustering} and later extended by \cite{kulldorff1997spatial} to detect anomalous clusters in spatio-temporal data.
  The main idea of scan statistics is to detect the abnormal clusters by utilizing a maximal log-likelihood ratio.
  It is worth noting that the scan statistics-based method is a parametric method that assumes the data distributions are known.   
  For instance, authors of \cite{tango2011space} assumes the negative binomial distribution, 
  \cite{kulldorff2001prospective}
  %\cite{neill2005detection, kulldorff1997spatial, kulldorff2001prospective, neill2006bayesian} 
  investigate the Poisson distribution.
  A limitation of the scan statistics is that it assumes that the background is independent and identically distributed (i.i.d.) or follows a rather simple probability distribution, which might not be suitable to handle non-stationary spatio-temporal data.
  
  For the tensor decomposition-based detection, one representative is \cite{SSD} for Gaussian data, and below is the main idea.
  First, it assumes that the raw data follows the normal distribution and then decomposes it into three components: global trend mean, local hot-spots, and residuals.
  Then, it uses the Tucker decomposition
  \cite[see][]{hitchcock1927expression}
  %\cite[see][]{hitchcock1927expression, tucker1966some} 
  to decompose the first two components.
  Finally, it detects the hot-spots by monitoring the model residuals.
  Similar ideas can also be found in elsewhere, for example, \cite{AnomalyInImage}.
  %\cite{AnomalyInImage,AnomalyInVideo}.
  The limitation of this method is that it does not consider the effect of population size and the normal distribution is clearly inappropriate for count data, especially those small values.
  Another representative is \cite{eren2020multi}, where the authors first assume the count data follows the Poisson distribution and then decompose the Poisson mean by canonical polyadic decomposition (CPD)
  \cite[see][]{hitchcock1927expression}, %\cite[see][]{hitchcock1927expression, carroll1970analysis, harshman1970foundations}, 
  and finally localize the hot-spots by small p-values.
  The limitation of this approach is that it does not remove the global trend mean from the raw data, which leads to its incapacity to detect the local hot-spots.
  Besides, it can only realize the hot-spots localization, and the analysis of the detection delay is not reported.
  
  For the change-point detection framework, there are two related categories, the \textit{Least Absolute Shrinkage and Selection} (Lasso) based methods and \textit{dimension reduction based methods}. 
  For the Lasso-based change-point detection method, one of the existing research representative is \cite{zou2009multivariate}.
  Note that Lasso has been demonstrated to be an effective method for variable selection to address sparsity issues for high-dimensional data in the past decades since its developments in \cite{tibshirani1996regression}, and thus it is natural to apply it to detect sparse changes in high-dimensional data 
  \cite[see][]{wei2018u, zou2009multivariate, zou2008LASSO, zhao2021identification}.
  %\cite[see][]{zou2009multivariate, zou2008LASSO, LassoBased1, vsaltyte2011spatial, zhao2021identification}.
  While the sparse change of Lasso is similar to the hot-spots,  unfortunately, as our extensive simulation studies will demonstrate, the Lasso-based control chart is unable to separate the local hot-spots from the non-stationary global trend mean in the spatio-temporal data.     
  For the dimension reduction-based method, the representative is \cite{PCA}.
  The main idea of this type of method is that it first uses principal component analysis (PCA) or other dimension reduction methods to extract the features from the high-dimensional data, and then it combines the reduced dimension information with change-point models to detect the hot-spots.
  For other dimension-reduction-methods, please see \cite{ paynabar2013PCA, PCA, chen2018pseudo}
  %\cite{liu1995PCA, paynabar2013PCA, PCA, tensorPCA1, tensorPCA2, bakshi1998PCA, yan2014image, chen2018pseudo}
  for more details.
  The drawbacks of the dimension reduction-based method are the restriction of the change-point detection problem and the failure to consider the spatial sparsity and temporal consistency of hot-spots.

%---------------------------------------------------%
%              Our Contribution                     %
%---------------------------------------------------%

\subsection{Overview of Our Proposed Method and its Contributions}
\label{sec: method overview and contribution}
The essential idea of our proposed method is to combine the tensor decomposition-based method and the Lasso-based method over the multivariate spatio-temporal data.
  Mathematically speaking, the  multivariate spatio-temporal data can be a \textit{tensor} of order three, which is a multi-dimensional array, i.e., $\mathcal Y \in \mathbb R^{n_1 \times n_2 \times n_3}$.
  Here $n_1$ is the number of the spatial locations,  $n_2$ is the number of different diseases or variates,  and $n_3$ is the number of time points.
  In our proposed method, we assume the tensor $\mathcal Y$ comes from the Poisson distribution and then consider the additive model that decomposes its Poisson rate into two components:
  (1) smooth but non-stationary global trend mean,
  and (2) sparse local hot-spots.
  Next, we fit the raw data with a penalty function, i.e., a Lasso type penalty to guarantee the spatial sparsity of hot-spots.
  This allows us to not only detect when the hot-spots happen over the temporal domain (i.e.,  hot-spots detection problem) but also localize where and which types/attributes of the hot-spots occur if the change happens (i.e., hot-spots localization problem).
  We term our proposed decomposition method as Poisson-assisted Smooth Sparse Tensor Decomposition (PoSSTenD).
  
  It is useful to highlight the novelty of our proposed method as compared to those existing methods on multivariate spatio-temporal count data.
  First, our proposed PoSSTenD method takes into account of the effect of population size to focus on the anomalous infectious diseases rates or Poisson rates, which is important in many real-world applications. 
  Second, our proposed PoSSTenD method can detect hot-spots when the global trend of the spatio-temporal data is dynamic (i.e., non-stationary or non-i.i.d).
  That is, our method is robust no matter whether the global trend is decreasing, stationary, or increasing, in the sense that it can detect hot-spots with positive or negative local mean shifts on top of the global trend of raw data. 
  In comparison, existing SPC or change-point detection methods often assume that the background is i.i.d. and focus on detecting the anomalies under the static and i.i.d. background.
  Finally, while our paper focuses only on the tensor of order three arising from our motivating application in 10 annual infectious diseases in the U.S., our proposed hot-spots method can easily be extended to the tensor of order $d$ ($d \geq 3$), as we can simply add corresponding dimensions and bases in the tensor analysis.
  The capability of extending to high-dimensional tensor data is one of the main advantages of our proposed PoSSTenD method.
  
  We would like to clarify that our contribution lies in the hot-spots detection and localization, instead of model fitting/prediction of the multivariate spatio-temporal count data.
  For the fitness of the count data, there are lots of well-established methods.
  One widely used one is the generalized linear regression model \cite[see][Chapter 3]{hastie2015statistical}.
  Another popular method is the time series models such as Autoregressive Moving Average (ARMA).%, for instance, \cite{mckenzie1988some, mckenzie1985some}.
  In this paper, we will not compare our proposed PoSSTenD method to these methods, because our objective is to detect and localize hot-spots, rather than model fitting.

  The remainder of this paper is as follows.
  In Section \ref{sec: data description}, we introduce and visualize a motivating dataset.
  In Section \ref{sec: model}, we present our proposed PoSSTenD method,  discuss how to estimate model parameters from data, and describe how to use our proposed PoSSTenD method to detect and localize  hot-spots.
  Our proposed method is compared with several benchmark methods in Section \ref{sec: simulation}, demonstrating its usefulness through extensive simulations. 
  The application of our proposed method in the dataset (described in Section \ref{sec: data description}) is reported in Section \ref{sec: data description}.

%---------------------------------------------------%
%                     Data                          %
%---------------------------------------------------%
\section{Motivating Example \& Tensor Background }
\label{sec: data description}

  In this section, we provide a detailed description of our motivating dataset, and some tensor background to help readers better understand our method in Section \ref{sec: model}.

%---------------------------------------------------%
%             Motivating Data                       %
%---------------------------------------------------%
\subsection{Motivating Example}
Our motivated data are from CDC on the number of annual new cases infected by 10 different diseases of 49 mainland US states from 1993 to 2018.
   The 10 infectious diseases are
   (1) mumps,
   (2) legionellosis,
   (3) tuberculosis,
   (4) syphilis,
   (5) shigellosis,
   (6) pertussis,
   (7) hepatitis A,
   (8) hepatitis B,
   (9) rabies,
   (10) malaria.
 The dataset is publicly available in \url{https://www.cdc.gov/mmwr/mmwr_nd/index.html} (for data from 1993 to 2015) and \url{https://wonder.cdc.gov/nndss/nndss_annual_tables_menu.asp} (for data from 2016-2018).  
 For illustration, we present the a selected subset of this dataset in Table \ref{table: motivating dataset}, where the states are listed in alphabetical order.
 %, and the  New York City and New York (excluding New York City) together as the number of people get infected in New York City.
   \begin{table}[htbp]
     \centering
     \caption{The head of our motivating dataset, which records number of people get infected by ten different types of diseases for all states in U.S. from 1993 to 2018
     \label{table: motivating dataset}}
     \begin{adjustbox}{max width=0.98\textwidth}
     \centering
     \begin{threeparttable}
     \begin{tabular}{cccccccccccc}
       \hline
       state & year & Mumps & Legionellosis & Tuberculosis & Syphilis & Shigellosis & Pertussis & Hepatitis A & Hepatitis B & Rabies & Malaria \\
       \hline
       Alabama & 1993 & 22 & 2  & 487 & 2333 & 375 & 65 & 58  & 107 & 116 & 7\\
       Arizona & 1993 & 19 & 17 & 231 & 557  & 693 & 70 & 1493& 96  & 60  & 1\\
       $\vdots$& $\vdots$ & $\vdots$ & $\vdots$ & $\vdots$ & $\vdots$  & $\vdots$ & $\vdots$ & $\vdots$ & $\vdots$  & $\vdots$  & $\vdots$\\
       Wyoming & 1993 & 5&    7  &  7   &   9 &  22   &  2   &  17   &34&25 &1\\
       Alabama & 1994 & 12 & 13 & 433 &  1933    &  418   &  35  &  93   &   91  &  128   &9\\
       $\vdots$& $\vdots$ & $\vdots$ & $\vdots$ & $\vdots$ & $\vdots$  & $\vdots$ & $\vdots$ & $\vdots$ & $\vdots$  & $\vdots$  & $\vdots$\\
       Wyoming & 2018 & 1&  2    &  1   & 42   &  3   &  62   &  5   &2&39&2\\
       \hline
     \end{tabular}
    \end{threeparttable}
    \end{adjustbox}
   \end{table}
   
%   \begin{tablenotes}
%      \footnotesize
% We use the linear interpolation method for missing values.
%(ii) We add up New York City and New York (excluding New York City) together as the number of people get infected in New York City.
   %  \end{tablenotes}

   Mathematically, we can organize the above motivating dataset into a tensor of order three, i.e., $\mathcal Y \in \mathbb R^{n_1 \times n_2 \times n_3}$, where $n_1 = 49$ is the number of states, $n_2 = 10$ is the number of diseases, and $n_3 = 26$  is the number of observed years.
   For this tensor of order three $\mathcal Y$, its $(i,j,t)$-th entry, i.e., $\mathcal Y_{i,j,t}$ is the number of people get infected by the $j$-th type of disease in the $i$-th state $i$ in year $t$.

   To better capture the character of the data $\mathcal Y$, we provide some data visualization plots in Figure \ref{fig: data description of three dimensions}.
   In Figure \ref{fig: data description of three dimensions}(a), we show the bar plot for the total number of affected subjects for these ten different diseases.
   Here the $j$-th bar represents the $j$-th type of disease and the height of $j$-th bar is the value of
   $
     \sum_{i=1}^{n_1} \sum_{t=1}^{n_3} \mathcal Y_{i,j,t}.
   $
   In Figure \ref{fig: data description of three dimensions}(b), we plot the cumulative number of infected subjects of 49 states in a map, where the color of the $i$-th state is decided by the total value
   $
     \sum_{j=1}^{n_2} \sum_{t=1}^{n_3} \mathcal Y_{i,j,t}
   $
   and a larger value leads to a deeper color.
   From this plot, we find that generally speaking, California, Texas, and New York observe a larger number of infected subjects  as compared with other states, possibly due to a large population size.
   In Figure \ref{fig: data description of three dimensions}(c), we display a time series plot, where the x-axis is the year and the y-axis is the logarithm of the value of
   $
     \sum_{i=1}^{n_1} \sum_{j=1}^{n_2} \mathcal Y_{i,j,t}
   $
   for $t = 1,2,\ldots, n_3$.
   We can see that, there is a trough from 2000 to 2010, but  the number of infected people has an increasing trend otherwise.

   \begin{figure}[htbp]
      \centering
      \begin{tabular}{ccc}
      \includegraphics[width=0.38\textwidth]{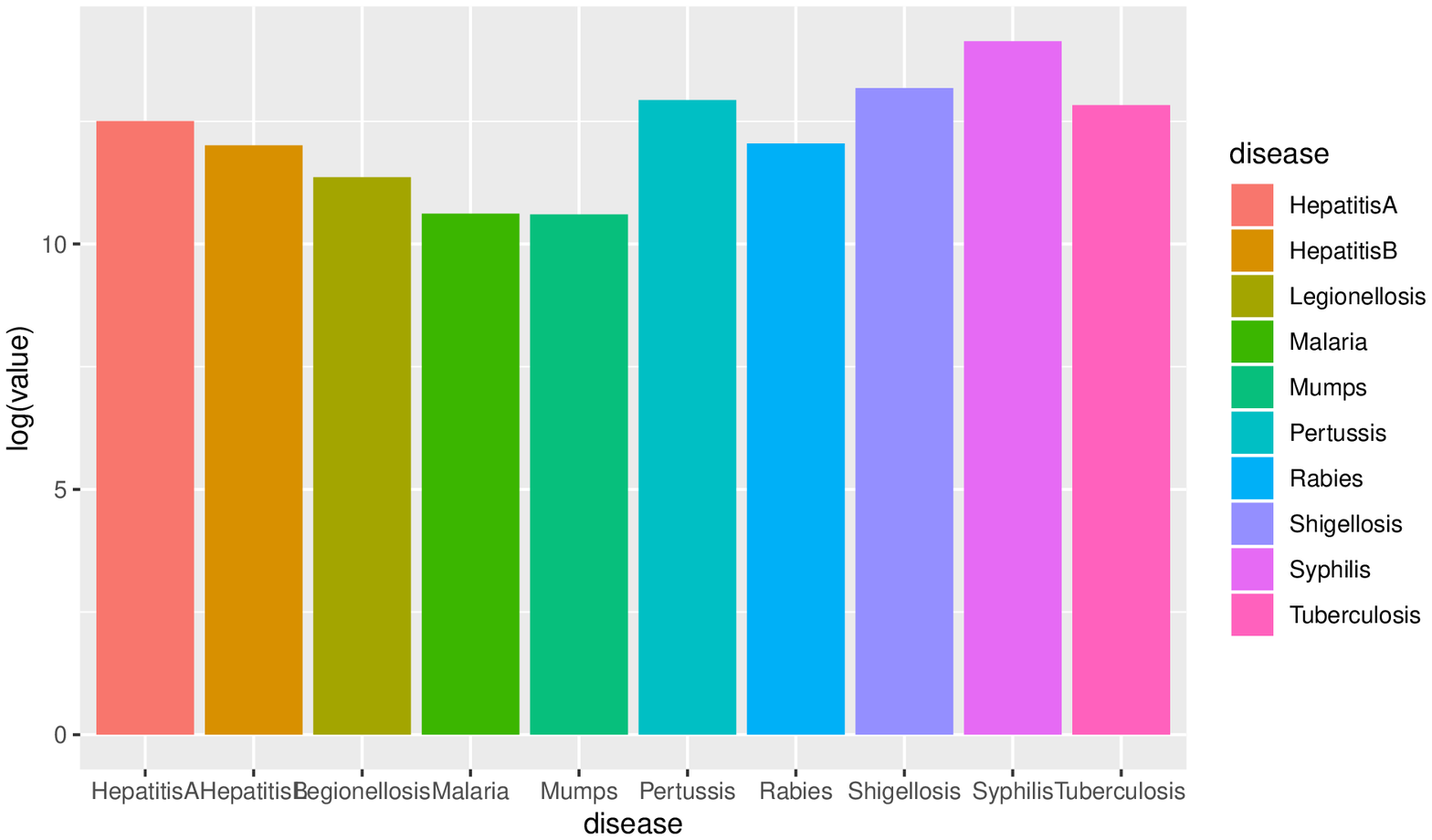}&
      \includegraphics[width=0.3\textwidth]{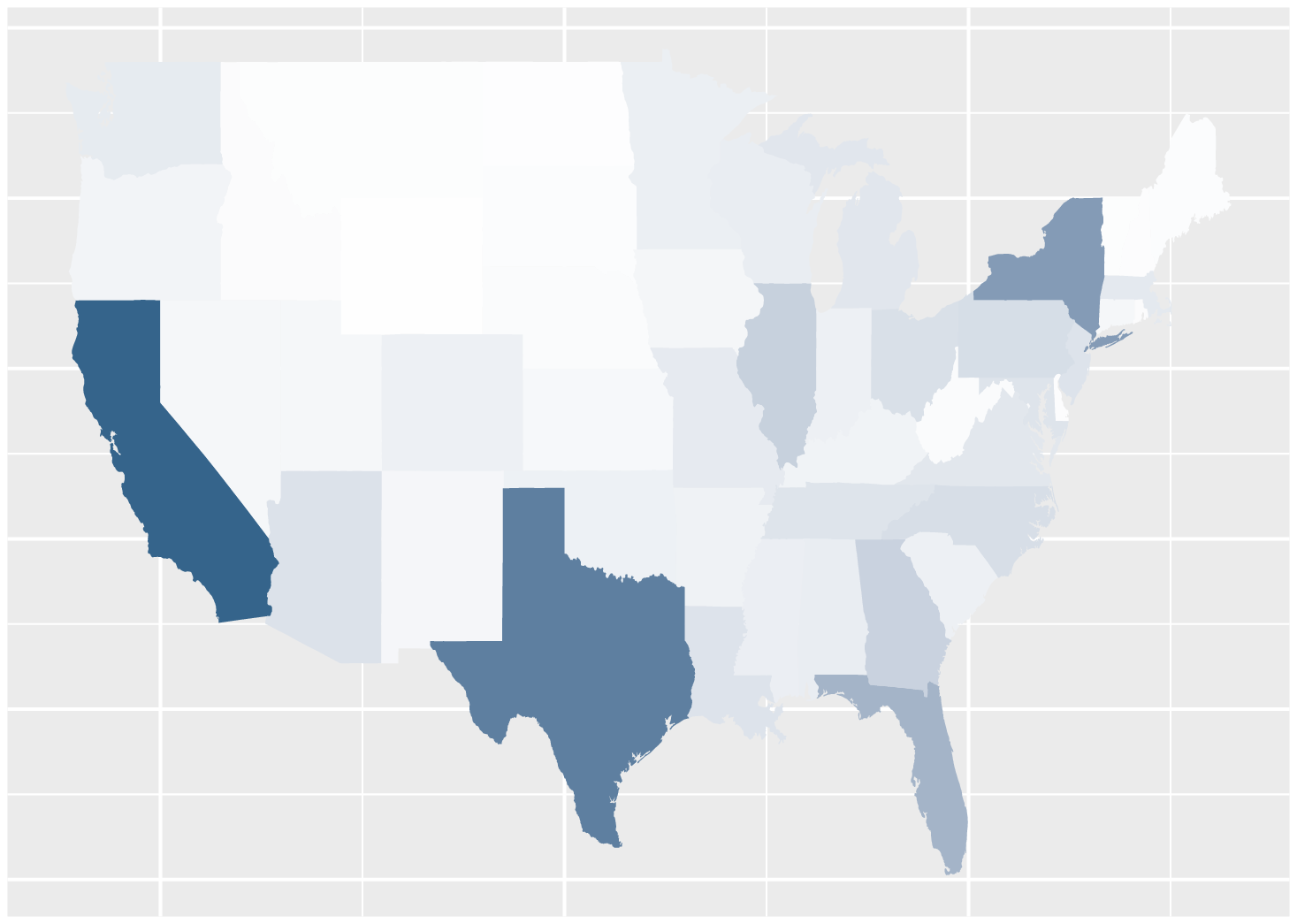}&
      \includegraphics[width=0.3\textwidth]{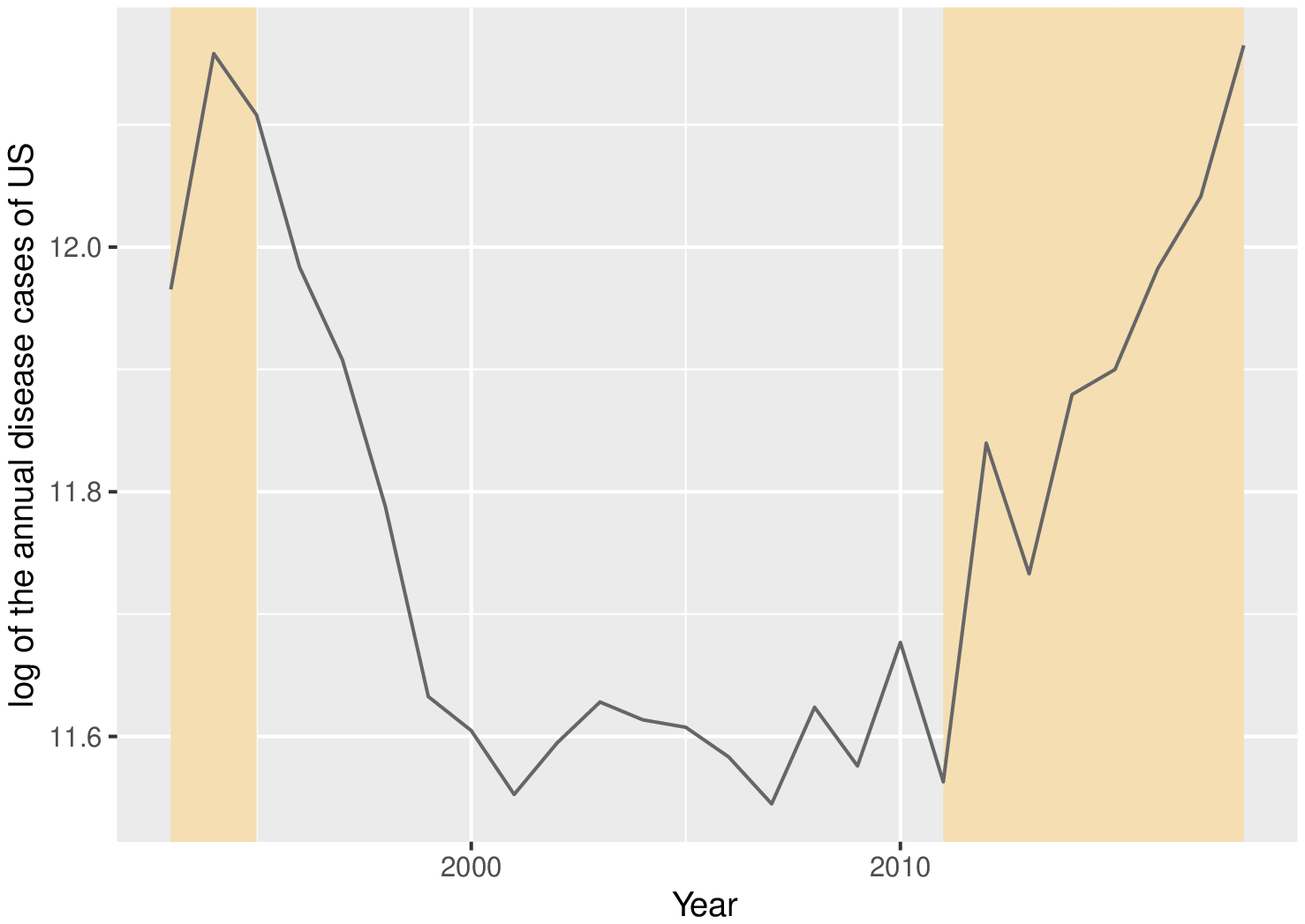}\\
      (a) bar plot of 10 diseases &
      (b) map of 49 states &
      (b) time series of 26 years
   \end{tabular}
   \caption{Data visualization of our motivating dataset
   \label{fig: data description of three dimensions} }
   \end{figure}

%---------------------------------------------------%
%             Tensor Background                     %
%---------------------------------------------------%

\subsection{Tensor Background}
\label{sec: tensor algebra}
  
  In this subsection, we present necessary tensor notations and tensor algebra so as to help readers understand our proposed  methodology in Section \ref{sec: model}.

  First, for the notations throughout the paper,
  scalars are denoted by lowercase letters (e.g., $\theta, y$), and vectors are denoted by lowercase boldface letters ($\mytheta, \myy$), whose $i$-th entry is denoted by the subscript ($\mytheta_i, \myy_i$).
  Matrices are denoted by uppercase boldface letter ($\boldsymbol{\Theta}, \myY$), whose $(i,j)$-th entry is denoted by the subscript ($\myTheta_{i,j}, \myY_{ij}$).
  Tensors are denoted by curlicue letter ($\vartheta, \mathcal Y$).
  For example, an $K$-dimension tensor is represented by
  $
    \vartheta \in \mathbb{R}^{I_{1} \times \ldots \times I_{K}},
  $
  where $I_{k}$ represent the mode-$K$ dimension of $\vartheta$ for $k = 1, \ldots, K$.
  Its $(i_1,i_2, \ldots, i_K)$-th entry is denoted by the subscript ($\vartheta_{i_1,i_2, \ldots, i_K}, \mathcal Y_{i_1,i_2, \ldots, i_K}$).

  Next, we introduce the definition of \textit{slice} in tensor.
  The slice is a two-dimensional section of a tensor by fixing all but two indices.
  Let us take a 3-dimension tensor
  $
    \mathcal Y \in \mathbb R^{n_1 \times n_2 \times n_3}
  $
  as an example.
  Its horizontal, lateral, and frontal slides are denoted by
  $
    \mathcal Y_{i::} (\forall i = 1,\ldots, n_1),
    \mathcal Y_{:j:} (\forall j = 1,\ldots, n_2),
  $
  and
  $
    \mathcal Y_{::t} \; (\forall t = 1,\ldots, n_3),
  $
  respectively.
  A visualization of the aforementioned three types of slides are available in Figure 3 of \cite{zhao2021rapid}.
 
  Finally, we introduce the \textit{Tucker decomposition}.
  It can be regarded as an extension of singular value decomposition (SVD) in the matrix.
  Mathematically, it decomposes a tensor $\mathcal Y \in \mathbb R^{n_1 \times n_2 \times \ldots \times n_d}$ into a product of a ``core tensor'' $\mathcal C$ and several matrices $\myB_1, \ldots, \myB_d$ which correspond to different core scalings along each mode, i.e.,
  $$
    \mathcal Y 
    = 
    \mathcal C 
    \times_1 \myB_1 
    \times_2 \myB_2 
    \ldots 
    \times_d \myB_d,
  $$
  where $\times_n$ is the mode-n product for any $n = 1, \ldots, d$ \cite[see][Section 2]{zhao2021rapid}. And the core tensor $\mathcal C$ is of dimension $\mathcal C \in \mathbb R^{p_1 \times p_2 \times \ldots \times p_d}$.
  The matrix $\myB_i$ is of dimension $\myB_i \in \mathbb R^{n_i \times p_i}$ for any $i = 1, \ldots, d$.
  In most of the applications, we assume the matrices $\myB_1, \ldots, \myB_d$ are known \cite[see][]{yan2018TECH, AnomalyInImage, yan2014image}.
  The high-level criterion to select $\myB_1, \ldots, \myB_d$ are listed as follows.
  If the data $\mathcal Y$ has a smooth pattern on mode-n, then one can use the either spline basis \cite[see][]{yan2018TECH, gahrooei2021multiple} or Gaussian kernel \cite[see][]{yan2018TECH} for $\myB_n$.
  If the data $\mathcal Y$ has no prior information on mode-n, then one can use the identity matrix for $\myB_n$ \cite[see][]{zhao2021rapid}. 
  And we call $\min\{p_1, p_2, \ldots, p_d\}$ as the rank of the Tucker decomposition \cite[see][]{TensorAlgebra}.

  %The question of how to choose the rank of the Tucker decomposition has been addressed by \cite{kiers2003fast}, who have a straightforward procedure for choosing the appropriate rank of a Tucker model based on an HOSVD calculation.
%---------------------------------------------------%
%                     Model                         %
%---------------------------------------------------%
\section{Our Proposed Method }
\label{sec: model}
    
In this section, we present our proposed PoSSTenD method, including the optimization algorithm for the model parameter estimation, and statistical procedures  for hot-spots detection and localization.

    The structure of this section is organized as follows.
    Subsection \ref{sec: model details} develops our proposed method.
    Subsection \ref{sec: Estimation and Optimization Algorithm} describes the parameters estimation and the algorithm to solve the estimation problems.
    Subsection \ref{sec: hot-spots detection and localization} discusses the procedure to detect and localize the hot-spots.

\subsection{Model}
\label{sec: model details}
    In this subsection, we present the mathematical model of our proposed PoSSTenD method.
    For better illustration, we take the motivating dataset described in Section \ref{sec: data description} as an example, which is a tensor of order three, i.e., $\mathcal Y \in \mathbb R^{n_1 \times n_2 \times n_3}$.
    The $(i,j,t)$-th entry of $\mathcal Y$ is the number of people who are infected by the $j$-th type of disease in the $i$-th state in year $t$.
    Since the number of infected patients is count data, we assume that it follows the Poisson distribution:
    \begin{equation}
    \label{equ: poisson model}
      \mathcal{Y}_{i,j,t} \sim \text{Poisson}( \mathcal R_{i,j,t} \mathcal{N}_{i,j,t})
    \end{equation}
    where $\mathcal R_{i,j,t}$ is the infectious rate of the $j$-th type of disease in the $i$-th state in year $t$, and $\mathcal N_{i,j,t}$ denotes the corresponding population size, for $i =1\ldots n_1, j = 1\ldots n_2, t = 1\ldots n_3$. In our simulation studies and case study, we assume that the population size $\mathcal N_{i,j,t}$ is the same over different disease type $j$, although they can be different under other scenarios, e.g., when certain diseases will affect certain sub-populations depending on gender, age, race, etc.  

    In our paper, we are interested in detecting and localizing the hot-spots in infectious rate $\mathcal R_{i,j,t}$, and thus we further decomposed the logarithm of the infectious rates into two components: the smooth global trend mean $\mathcal{U}_{i,j,t}$ and the local hot-spots $\mathcal{H}_{i,j,t}$.
    Mathematically, it is an additive model:
    \begin{eqnarray}
    \label{equ: decompose infectious rate into 3 components}
        \log( \mathcal R_{i,j,t} )
        =
        \mathcal{U}_{i,j,t} + \mathcal{H}_{i,j,t},
    \end{eqnarray}
    where
    $
      \mathcal U_{i,j,t}$ and $\mathcal H_{i,j,t}
    $
    are the $(i,j,t)$-th entry of tensors
    $
      \mathcal U$ and $\mathcal H \in \mathbb R^{n_1 \times n_2 \times n_3},
    $
    respectively.

    For the first component -- the global trend mean $\mathcal U \in \mathbb R^{n_1 \times n_2 \times n_3}$ --  we apply Tucker decomposition for dimension reduction of tensors, i.e.,
    \begin{equation}
    \label{equ: decompose global mean}
      \mathcal U
      =
      \vartheta_m \times_1 \myB_{m,1} \times_2 \myB_{m,2} \times_3 \myB_{m,3},
    \end{equation}
    where
    $
      \myB_{m,1} \in \mathbb R^{n_1 \times p_1},
      \myB_{m,2} \in \mathbb R^{n_2 \times p_2},
      \myB_{m,3} \in \mathbb R^{n_3 \times p_3}
    $
    are the pre-assigned basis describing the within-state correlation, within-disease-type correlation, and within-year correlation in $\mathcal U$, respectively.
    For the selection of $\myB_{m,1}, \myB_{m,2}, \myB_{m,3}$, we postpone its discussion at the end of this subsection.
    The operators $\times_1, \times_2, \times_3$ are the mode-$n$ product reviewed in Section \ref{sec: tensor algebra} with $n = 1,2,3$.
    These three mode-$n$ products are used to model the between-dimension correlations in $\mathcal U$.
    The tensor $\vartheta_m \in \mathbb R^{p_1 \times p_2 \times p_3}$ is an unknown tensor parameter to be estimated.

    For the second component -- local hot-spots $\mathcal H \in \mathbb R^{n_1 \times n_2 \times n_3}$ -- we apply the similar Tucker decomposition as in $\mathcal U$:
    \begin{equation}
    \label{equ: decompose local hotspots}
      \mathcal H
      =
      \vartheta_h \times_1 \myB_{h,1} \times_2 \myB_{h,2} \times_3 \myB_{h,3},
    \end{equation}
    where the $\vartheta_h \in \mathbb R^{q_1 \times q_2 \times q_3}$ is the unknown parameter to be estimated and basis
    $
      \myB_{h,1} \in \mathbb R^{n_1 \times q_1},
      \myB_{h,2} \in \mathbb R^{n_2 \times q_2},
      \myB_{h,3} \in \mathbb R^{n_3 \times q_3}
    $
    are pre-assigned matrices describing the within-state correlation, within-disease-type correlation, and within-year correlation in $\mathcal H$, respectively. The suitable choices of these bases will be discussed in more details in Section \ref{sec: simulation}.
    The operators $\times_1, \times_2, \times_3$ are the mode-$n$ product used to model the between-dimension correlations in $\mathcal H$.

    By combining \eqref{equ: poisson model}, \eqref{equ: decompose infectious rate into 3 components}, \eqref{equ: decompose global mean}, \eqref{equ: decompose local hotspots}, we summarize our method as
    \begin{eqnarray*}
    \label{equ: model2}
          \left\{
          \begin{array}{l}
              \mathcal{Y}_{i,j,t}
              \sim
              \mbox{Poisson}( \mathcal{R}_{i,j,t} \mathcal{N}_{i,j,t} ) \\
              \log( \mathcal{R}_{i,j,t} )
              =
              (\vartheta_m \times_1 \myB_{m,1} \times_2 \myB_{m,2} \times_3 \myB_{m,3} )_{i,j,t}+
              (\vartheta_h \times_1 \myB_{h,1} \times_2 \myB_{h,2} \times_3 \myB_{h,3})_{i,j,t}
          \end{array}
          \right.,
    \end{eqnarray*}
    where
    $
      (\vartheta_m \times_1 \myB_{m,1} \times_2 \myB_{m,2} \times_3 \myB_{m,3} )_{i,j,t},
      (\vartheta_h \times_1 \myB_{h,1} \times_2 \myB_{h,2} \times_3 \myB_{h,3})_{i,j,t}
    $
    denote the $(i,j,t)$-th entry in tensor
    $
      \vartheta_m \times_1 \myB_{m,1} \times_2 \myB_{m,2} \times_3 \myB_{m,3} ,
      \vartheta_h \times_1 \myB_{h,1} \times_2 \myB_{h,2} \times_3 \myB_{h,3}
      \in
      \mathbb R^{n_1 \times n_2 \times n_3},
    $
    respectively.

    This tensor representation allows us to develop computationally efficient methods for estimation and prediction.
    Under tensor algebra, the above model can be rewritten as
    \begin{equation}
    \label{equ: model -- vector form}
      \left\{
      \begin{array}{l}
        \myy
        \dot \sim
        \mbox{Poisson}( \myr \odot \myn ) \\
        \log(\myr)
        =
        \left( \myB_{m,1} \otimes \myB_{m,2} \otimes \myB_{m,3} \right)
        \mytheta_m +
        \left( \myB_{h,1} \otimes \myB_{h,2}  \otimes \myB_{h,3} \right)
        \mytheta_h
      \end{array}
      \right..
    \end{equation}
    Here vectors
    $
      \myy, \myr, \myn \in \mathbb R^{n_1 n_2 n_3},
      \mytheta_m \in \mathbb R^{p_1 p_2 p_3},
      \mytheta_h \in \mathbb R^{q_1 q_2 q_3}
    $
    are the vectorized form of tensor
    $
      \mathcal Y, \mathcal R, \mathcal N \in \mathbb R^{n_1 \times n_2 \times n_3},
      \vartheta_m \in \mathbb R^{p_1 \times p_2 \times p_3},
      \vartheta_h \in \mathbb R^{q_1 \times q_2 \times q_3},
    $
    respectively.
    The operator $\odot$ is the Hadamard product, i.e., the output of $\myr \odot \myn$ is a vector of length $n_1 n_2 n_3$, whose $i$-th entry is $r_i \times n_i$ with $r_i, n_i$ as the $i$-th elements of vector $\myr, \myn$, respectively.
    The operator $\otimes$ is the Kronecker product whose mathematical definition can be found in \cite{zhao2021rapid, zhao2021new}.
    Besides, the operator $\dot \sim$ means that the $i$-th entry of $\myy$ follows the Poisson distribution with parameter $r_i n_i$ for any $i = 1, \ldots, n_1 n_2 n_3$.

    In the reminder of this subsection, we articulate the selection of basis for global trend mean $\myB_{m,1}, \myB_{m,2}, \myB_{m,3}$, and basis for local hot-spots $\myB_{h,1}, \myB_{h,2}, \myB_{h,3}$, particularly in our case study.
    
    For the selection of $\myB_{m,1}, \myB_{m,2}, \myB_{m,3}$, given $\mathcal Y$'s smooth pattern, we use the cubic B-spline basis with 8, 7, 7 knots in the x-direction (spatial domain), y-direction (disease-categorical domain) and z-direction (temporal domain).
    This selection strategy is also adopted in Section 6 of \cite{yan2018TECH}, and the above three sets of knots are reasonable in our case study.
    In general, a smaller number of knots may result in the loss of detection accuracy, and a larger number of knots will lead to the selection of normal regions with large noises.
    As pointed out by \cite{ruppert2002selecting}, as long as the number of knots is sufficiently large to capture the variation of the background, a further increase in knots will have little effect due to the regulation of smoothness.
    After verifying the above three sets of knots by generalized cross validation (GCV), as suggested by \cite{ruppert2002selecting}, we find they are sufficient large to capture the variation of the global trend mean.
    Other choice of $\myB_{m,1}, \myB_{m,2}, \myB_{m,3}$ is welcomed, for example, the Gaussian kernel basis \cite[see][]{zhao2021rapid, yan2018TECH}.
    In this paper, we will use the B-spline basis for illustration of our proposed PoSSTenD method. 
    Given the above selection of $\myB_{m,1}, \myB_{m,2}, \myB_{m,3}$, the rank of the Tucker decomposition of  $\mathcal U$ is $\min\{p_1, p_2, p_3\} = 4$. 

    For the selection of $\myB_{h,1}, \myB_{h,2}, \myB_{h,3}$, we set them as the identity matrices, due to the following two reasons. 
    First, there is no prior information about the hot-spots distribution, which makes identity matrices a reasonable choice.
    Second, we aim at detecting sparse hot-spots.
    If one in interested in detecting clustered hot-spots, then a spline basis is a better choice \cite[see][Section 3.4]{AnomalyInImage}. 
    Given the above selection of $\myB_{h,1}, \myB_{h,2}, \myB_{h,3}$, the rank of the Tucker decomposition of $\mathcal H$ is $\min\{n_1, n_2, n_3\} = 10$ in the our case.
    This rank is not reduced much, because we do not want to lose much information considering the hot-spots detection accuracy. 
    
%---------------------------------------------------%
%                Estimation                         %
%---------------------------------------------------%
\subsection{Estimation and Optimization}
\label{sec: Estimation and Optimization Algorithm}

    This subsection discusses the parameter estimation of our proposed PoSSTenD method.
    We will first formulate the parameter estimation as an optimization problem in Subsection \ref{sec: estimation}, and then develop an efficient algorithm to solve this optimization problem in Subsection \ref{sec: Optimization Algorithm}.

\subsubsection{Estimation of the Model Parameters}
\label{sec: estimation}

    To estimate the two unknown vector $\mytheta_m \in \mathbb R^{p_1 p_2 p_3}, \mytheta_h \in \mathbb R^{q_1 q_2 q_3}$, it is natural to first consider  the maximal likelihood estimating (MLE) method.
    We begin with the likelihood function
    \begin{equation*}
        L(\mytheta_m, \mytheta_h)
        =
        \prod_{i=1}^{n_1 n_2 n_3}
        \left(n_i r_i \right)^{y_i} e^{ -r_i n_i } / y_i!,
    \end{equation*}
    where $n_i, r_i, y_i$ are the $i$-th elements of vector $\myn, \myr, \myy \in \mathbb R^{n_1 n_2 n_3}$, respectively.
    Accordingly, the log-likelihood function of $(\mytheta_m, \mytheta_h)$, i.e., $\ell(\mytheta_m, \mytheta_h) = \log L(\mytheta_m, \mytheta_h)$ can derived as
    \begin{eqnarray*}
        \ell(\mytheta_m, \mytheta_h)
        & = &
        \sum_{i=1}^{n_1 n_2 n_3}
        y_i \log\left(n_i r_i \right)  - n_i r_i - \log(y_i!)
        \propto
        \sum_{i=1}^{n_1 n_2 n_3}  y_i \log \left( r_i \right)  - n_i r_i ,
    \end{eqnarray*}
    where $\propto$ is done by removing the constant unrelated to $(\mytheta_m, \mytheta_h)$.
    If we denote the
    $
      \ell^*(\mytheta_m, \mytheta_h)
      :=
      \sum_{i=1}^{n_1 n_2 n_3}  y_i \log \left( r_i \right)  - n_i r_i
    $
    then by plugging \eqref{equ: model -- vector form} in, we have
    $$
      \ell^*(\mytheta_m, \mytheta_h)
      =
      \sum_{i=1}^{n_1 n_2 n_3}
      y_i
      \left( \myB_{m} \mytheta_m +  \myB_{h} \mytheta_h \right)_i
      -
      n_i e^{\left( \myB_{m} \mytheta_m +  \myB_{h} \mytheta_h \right)_i }.
    $$
    By denoting the $i$-th row of matrix $\myB_m, \myB_h$ as $\myx_i^\top, \myz_i^\top$, respectively, we can rewrite  the above equation as
    $$
      \ell^*(\mytheta_m, \mytheta_h)
      =
      \sum_{i=1}^{n}
      y_i
      \left( \myx_i^\top \mytheta_m +  \myz_i^\top \mytheta_h \right)
      -
      n_i e^{ \left( \myx_i^\top \mytheta_m +  z_i^\top \mytheta_h \right) },
    $$
    where $n = n_1 n_2 n_3$.
    Given the above log-likelihood function, on the one hand, we would like to minimize it with respect to $\mytheta_m, \mytheta_h$.
    On the other hand, we would like to detect the hot-spots with sparsity.
    Considering the trade-off between the maximization of a log-likelihood function $\ell^*(\mytheta_m, \mytheta_h)$ and the detection of sparse hot-spots $\mytheta_h$, we propose to  estimate the model parameter $(\mytheta_m, \mytheta_h)$ under the penalized likelihood ration framework, which yields the following Lasso-type optimization problem:
    \begin{eqnarray}
    \label{equ: objective function1}
      F(\mytheta_m, \mytheta_h)
      & = &
      \sum_{i=1}^{n_1 n_2 n_3}
      \left[
        - y_i
        \left( \myx_i^\top \mytheta_m +  \myz_i^\top \mytheta_h \right)
        +
        n_i e^{\left( \myx_i^\top \mytheta_m +  \myz_i^\top \mytheta_h \right)  }
      \right]
      +
      \lambda \left\| \mytheta_h \right\|_1.
    \end{eqnarray}
    Here the first term is the negative log-likelihood function and the second term is a $\ell_1$ regularization term, which gives sparsity of the hot-spots estimators.
    The parameter $\lambda>0$ trades off the minimization of the negative log-likelihood function and the sparsity of the estimators, and is often chosen by cross-validation of the training data.

  %For the above optimization problem, there is a major limitation, that is, it assigns the same weight to the individual sample $(y_i, x_i, z_i)$ regardless of population size $n_i$’s, although the $y_i$’s with larger population sizes $n_i$’s surely provide more information.
%  To overcome this advantage, we following \cite{MeiTsuiPossion} and change the above optimization problem into:
%  \begin{eqnarray}
%  \label{equ: objective function2}
%    \widetilde F(\mytheta_m, \mytheta_h)
%    & = &
%    \sum_{i=1}^{n}
%    \frac{1}{n_i}
%    \left[
%    - y_i \left( \myx_i^\top \mytheta_m +  \myz_i^\top \mytheta_h \right) + n_i e^{\left( \myx_i^\top \mytheta_m +  \myz_i^\top \mytheta_h \right)  }
%    \right]
%    +
%    \lambda \left\| \mytheta_h \right\|_1.
%  \end{eqnarray}

%---------------------------------------------------%
%                algorithm                          %
%---------------------------------------------------%
\subsubsection{Optimization Algorithm}
\label{sec: Optimization Algorithm}

    In this subsection, we discuss the optimization algorithm to minimize the objective function $F(\mytheta_m, \mytheta_h)$ in \eqref{equ: objective function1}.
    The main tools we use to minimize $F(\mytheta_m, \mytheta_h)$ is the Iteratively Reweighted Least Square (IRLS) method \cite[see][Chapter 4.4]{friedman2001elements} and the Fast Iterative Shrinkage Threshold algorithm (FISTA) \cite[see][]{FISTA}.
    
    The proposed  optimization algorithm is an iterative algorithm with two loops: the outer-loop for updating $\mytheta_m$ and the inner-loop for updating $\mytheta_h$.
    The optimization procedure is visualized in in Figure \ref{fig: min F tilde} and described as follows. 
    We begin with the initial point $(\mytheta_m^{(0)}, \mytheta_h^{(0)})$.
    Then, in the first outer-loop, we minimize $F(\mytheta_m, \mytheta_h^{(0)})$ with respect to $\mytheta_m$ by the IRLS algorithm.
    In this way, we update $\mytheta_m$ from $\mytheta_m^{(0)}$ to $\mytheta_m^{(1)}$.
    Next, we minimize $F(\mytheta_m^{(1)}, \mytheta_h)$ with respect to $\mytheta_h$ by the FISTA algorithm, which allows us to update $\mytheta_h$ from $\mytheta_h^{(0)}$ to $\mytheta_h^{(1)}$.
    The iteration of the FISTA algorithm is called the inner-loop under the first outer-iteration.
    The above outer-loop and inter-loop will be repeated until the convergence.
    The detailed pseudo-code is summarized in Algorithm \ref{alg: pseudo code}.

    \begin{figure}[htbp]
      \centering
      \includegraphics[width=0.9\textwidth]{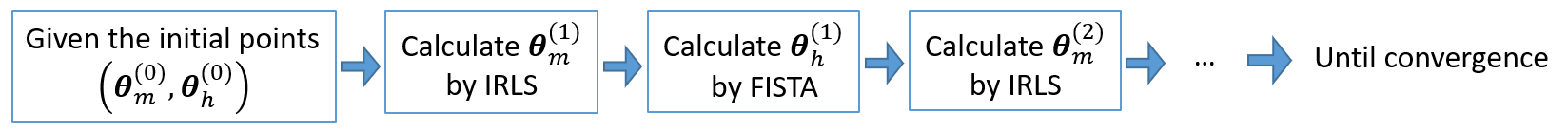}
      \caption{The pipeline of our algorithm to minimize $F(\mytheta_m, \mytheta_h)$ with respect to $\mytheta_m, \mytheta_h$.
      \label{fig: min F tilde}}
    \end{figure}

    \begin{algorithm}[htbp]
	\caption{Pseudo-code of the optimization algorithm in our proposed PoSSTenD method}
    \label{alg: pseudo code}
        \LinesNumbered
        \KwIn{
        $
          \myy, \myn \in \mathbb R^{n_1 n_2 n_3},
          \myB_m \in \mathbb R^{n_1 n_2 n_3 \times p_1 p_2 p_3},
          \myB_h \in \mathbb R^{n_1 n_2 n_3 \times q_1 q_2 q_3},
          \lambda > 0.
        $}
        \KwOut{The estimator of $\mytheta_m, \mytheta_h$, noted as $\widehat\mytheta_m, \widehat\mytheta_h$.}
        %------ initialization --------------- %
        \bfseries{Initialization}
        $\mytheta_m^{(0)}, \mytheta_h^{(0)}$. \\
        %------ outer iteration --------------- %
        \OuterIteration{
          \textnormal{ update $\mytheta_m$ from $\mytheta_m^{(k-1)}$ to $\mytheta_m^{(k)}$ by the IRLS algorithm}\\
          \InnerIteration{
          \textnormal{ update $\mytheta_h$ from $\mytheta_h^{(k-1)}$ to $\mytheta_h^{(k)}$ by the FISTA algorithm }\\
          }
        }
    \end{algorithm}

    In the reminder of this section, we discuss the detailed implementation of the IRLS algorithm and FISTA algorithm, as shown in the aforementioned Algorithm \ref{alg: pseudo code}

    First,  the implementation of the IRLS algorithm to estimate the parameter of the global trend mean $\mytheta_m$ is summarized in Proposition \ref{prop: alg IRLS}.
    
    \begin{prop}
    \label{prop: alg IRLS}
        In Algorithm \ref{alg: pseudo code}, given $(\mytheta_m^{(k-1)}, \mytheta_h^{(k-1)})$, one can update 
        $\mytheta_m^{(k)}$ by IRLS algorithm, i.e.,
        \begin{eqnarray*}
            \mytheta_m^{(k)}
            =
            \left( \myX^\top \myW \myX \right)^{-1}
            \myX^\top \myW
            \underbrace{
            \left[
            \myX \mytheta_m^{(k-1)} + \myW^{-1} \left( \myy - \mygamma^{(k-1)} \right)
            \right]
            }_{\myeta^{(k-1)}}.
        \end{eqnarray*}
        Here the matrix $\myX^\top$ is of dimension $p_1 p_2 p_3 \times n_1 n_2 n_3$.
        And its $i$-th column, denoted as $\myx_i$, is the transpose of $i$-th row of matrix $\myB_m$.
        The matrix
        $
          \myW \in \mathbb R^{n_1 n_2 n_3 \times n_1 n_2 n_3}
        $
        is a diagonal matrix.
        And its $(i,i)$-th entry is
        $
          n_i e^{\left( \myx_i^\top \mytheta_m^{(k-1)} +  \myz_i^\top \mytheta_h^{(k-1)} \right)}.
        $
        The vector $\myy$ is of length $n_1 n_2 n_3$, whose $i$-th entry is $y_i$. 
        The vector $\mygamma^{(k-1)}$ is of length $n_1 n_2 n_3$.
        And its $i$-th entry is 
        $
          n_i e^{\left( \myx_i^\top \mytheta_m^{(k-1)} 
          +  
          \myz_i^\top \mytheta_h^{(k-1)} \right)},
        $
        where $\myz_i^\top$ is the $i$-th row of the matrix $\myB_h$.
    \end{prop}

    \begin{proof}
      The proof of the above proposition is shown in Appendix \ref{proof: IRLS}.
    \end{proof}

    Next, we discuss the implementation of the FISTA algorithm to estimate the parameter of the local hot-spots $\mytheta_h$.
    Suppose in the $k$-th iteration, we already get the iterative estimator $\mytheta_m^{(k)}$ by Proposition \ref{prop: alg IRLS}.
    Given that, we can updating $\mytheta_h$ from $\mytheta_h^{(k-1)}$ to $\mytheta_h^{(k)}$ by the minimizing the following optimization problem with respect to $\mytheta_h$:
    \begin{equation}
    \label{equ: F fix theta m}
        F(\mytheta_m^{(k)}, \mytheta_h) =
        \underbrace{
        \sum_{i=1}^{n}
        \left[
          - y_i
          \left( \myx_i^\top \mytheta_m^{(k)} +  \myz_i^\top \mytheta_h \right)
          +
          n_i
          e^{\left( \myx_i^\top \mytheta_m^{(k)} +  \myz_i^\top \mytheta_h \right)  }
        \right]
        }_{g (\mytheta_h)}
        +
        \underbrace{
        \lambda \left\| \mytheta_h \right\|_1
        }_{h(\mytheta_h)}.
    \end{equation}
    Essentially speaking, the above optimization problem is an optimization problem with a $\ell_1$ regularization term.
    There is a large body of literature available on the analysis of it, where both researchers in optimization research and statistics devote to it.

    In optimization research, one well-known method is the Iterative Shrinkage Threshold Algorithm (ISTA) proposed by \cite{ISTA}, where they approximate $g(\mytheta_h)$ by its second-order Taylor expansion.
    For the Hessian matrix, to speed up the algorithm, they approximate it by a diagonal matrix, whose diagonal entry is the maximal eigenvalue of the Hessian matrix.
    After the quadratic approximation of the objective function, a proximal mapping is used, where the soft-thresholding operator can be easily applied.
    \cite{ISTA} is a representative of this type of method, and the early foundational work of proximal gradient descent can be found in \cite{ISTAhistoryBregman1966}.
    %\cite{ISTAhistoryBregman1966}, \cite{ISTAhistoryHestenes1969}, and \cite{ISTAhistoryRockafellar1976}.
    As the techniques mature, they became widely used in several different fields.
    As a result, they have been referred to by a diverse set of names, including proximal algorithm, proximal point, and so on.
    In the survey of \cite{ISTAsurvy}, it can be seen that, many other widely-known methods -- including, Majorization-Minimization (MM) %\cite[see][]{ISTAhistoryLange2000, ISTAhistoryHunter2005}, 
    and Alternating Direction Method of Multipliers (ADMM) %\cite[see][]{ISTAhistoryBoyd2011} 
    -- also fall into the proximal framework.
    Following \cite{ISTA}, \cite{FISTA} proposes another algorithm called FISTA, which can be regarded as an accelerated version of \cite{ISTA}.
    The main difference between ISTA and FISTA is that FISTA uses more historical information, i.e., FISTA takes advantage of the gradient of the previous two estimators, while ISTA only uses the gradient of the previous one estimator.
    Accordingly, the convergence rate of FISTA is $O(1/k^2)$ iterations, which is faster than $O(1/k)$ achieved by ISTA \cite[see][]{zhao2020homotopic}.

    In statistics, researchers also do lots of works.
    One of the well-known methods is proposed by \cite{glmnet} with an R package called \textit{glmnet}.
    The main tool used is the coordinate descent algorithm.
    The coordinate descent method has been proposed several times for the optimization with $\ell_1$ regularization term, e.g. see discussions in \cite{glmnet} and its convergence rate is $O(1/k)$ \cite[see][]{zhao2020homotopic}.
    %convergence analysis by \cite{CDhistoryTseng2001}, which proves that the convergence rate of \cite{glmnet} is .
    %Early research work on coordinate descent include the early discovery by \cite{CDhistoryHildreth1957} and \cite{CDhistoryWarga1963}, 

    After comparing the existing algorithm to solve \eqref{equ: F fix theta m}, we select FISTA in our paper due to its faster convergence.
    The detailed implementation of the FISTA algorithm can be found in the following proposition.

    \begin{prop}
    \label{prop: alg FISTA}
        In Algorithm \ref{alg: pseudo code}, given $(\mytheta_m^{(k)}, \mytheta_h^{(k-1)})$, one can update $\mytheta_h$ from $\mytheta_h^{(k-1)}$ to $\mytheta_h^{(k)}$ by the FISTA algorithm.
        The updating scheme of the FISTA algorithm is summarized as follows. 
        For a fixed $k>0$, and any $s = 0, 1, \ldots$, one has
        \begin{eqnarray*}
          \mytheta^{(k)[s]}_h
          & = &
          S
          \left(
            \myalpha^{[s]}
            -
            \frac{\partial}{\partial \mytheta_m} g(\mytheta_m^{(k)}, \myalpha^{[s]})
            / L,
            \lambda/L
          \right) \\
          t_{s+1}
          & = &
          \left(1+\sqrt{1 +4 t_s^2} \right) / 2 \\
          \myalpha^{[s+1]}
          & = &
          \mytheta_h^{(k)[s]} + \frac{t_s - 1}{t_{s+1}}
          \left( \mytheta_h^{(k)[s]} - \mytheta_h^{(k)[s-1]} \right),
        \end{eqnarray*}
        where $\mytheta_h^{(k)[s]}$ is the iterative estimator of $\mytheta_h$ after $s$ inner loops under the $k$-th outer loops.
        The function $S(\cdot, \cdot)$ is a thresholding function, i.e., $S(x, a) = ( x - a) \mathbbm 1\{ x \geq a \} + (x+a) \mathbbm 1 \{x \leq - a \}$ for any $x \in \mathbb R$ and $a > 0$.
        The vector $\myalpha^{[s]}$ is an auxiliary vector and is initialized as $\mytheta_h^{(k-1)}$.
        The hyper parameter $t_s$ and $L$ can be initialized as $t_1 = 1$ and set as the maximal eigenvalue of the matrix $\frac{\partial^2}{\partial \mytheta_h \partial \mytheta_h^\top} g(\mytheta_h)$, respectively.
    \end{prop}

    By combining Proposition \ref{prop: alg IRLS} and Proposition \ref{prop: alg FISTA}, we update $\mytheta_m$ and $\mytheta_h$ iteratively until convergence.
    We summarize the details of our algorithm for parameter estimation as follows:

    \begin{algorithm}[htbp]
	\caption{Estimation of  $\mytheta_m, \mytheta_h$ by IRLS \& FISTA algorithm }
    \label{alg}
        \LinesNumbered
        \KwIn{
        $
          \myy, \myn \in \mathbb R^{n_1 n_2 n_3},
          \myB_m \in \mathbb R^{n_1 n_2 n_3 \times p_1 p_2 p_3},
          \myB_h \in \mathbb R^{n_1 n_2 n_3 \times q_1 q_2 q_3},
          \lambda>0,
          t_1
        $. }
        \KwOut{The estimators of $\mytheta_m, \mytheta_h$, noted as $\widehat\mytheta_m, \widehat\mytheta_h$.}
        %------ initialization --------------- %
        \bfseries{Initialization}
        $t_1 = 1, \mytheta_m^{(0)}, \mytheta_h^{(0)}$. \\
        %------ outer iteration --------------- %
        \textit{ $\blacktriangleright$ Outer-Iteration: $\blacktriangleleft$ }
        \For{ $k = 0, 1,2,3, \ldots, K$ }{
            $
              \myW
              =
              diag
              \left(
                e^{\myx_1^\top\mytheta_m^{(k)} + \myz_1^\top\mytheta_h^{(k)} },
                \ldots,
                e^{ \myx_{n_1 n_2 n_3}^\top\mytheta_m^{(k)} + \myz_{n_1 n_2 n_3}^\top\mytheta_h^{(k)} }
              \right)
            $\\
            $
              \mygamma^{(k)}
              =
              \left(
                e^{\myx_1^\top\mytheta_m^{(k)} + \myz_1^\top\mytheta_h^{(k)} },
                \ldots,
                e^{\myx_{n_1 n_2 n_3}^\top \mytheta_m^{(k)} + \myz_{n_1 n_2 n_3}^\top\mytheta_h^{(k)} }
              \right)^\top
            $\\
            $
              \myeta^{(k)}
              =
              \myX \mytheta_m^{(k)} + \myW^{-1} \left( \widetilde\myy - \mygamma^{(k)} \right)
            $\\
            $
              \mytheta_m^{(k+1)}
              =
              \left( \myX^\top \myW \myX  \right)^{-1} \myX^\top \myW \myeta^{(k)}
            $ \\
            %------ middle iteration --------------- %
            $
              \mybeta^{(k)[0]} = \mytheta_h^{(k)}
            $\\
            $
              \myalpha^{[1]} = \mytheta_h^{(k)}
            $\\
            $
              t_1 = 1
            $\\
            \textit{ $\blacktriangleright$ Inner-Iteration: $\blacktriangleleft$ }
            \For{s = 1,2, \ldots, S}{
            $
              \mytheta^{(k)[s]}_h
              =
              S
              \left(
                \myalpha^{[s]}
                -
                \frac{\partial}{\partial \mytheta_m} g(\mytheta_m^{(k+1)}, \myalpha^{[s]})/ L,
                \lambda/L
              \right)
            $ \\
            $
              t_{s+1} =  \left(1+\sqrt{1 +4 t_s^2} \right) / 2
            $ \\
            $
              \myalpha^{[s+1]}
              =
              \mytheta_h^{(k)[s]} + \frac{t_s - 1}{t_{s+1}}
              \left( \mytheta_h^{(k)[s]} - \mytheta_h^{(k)[s-1]} \right)
            $\\
            }
            $
              \mytheta_h^{(k+1)} = \mytheta_h^{(k)[S]}
            $
        }
        $
          \widehat\mytheta_m = \mytheta_m^{(K)}
        $    \\
        $
          \widehat\mytheta_h = \mytheta_h^{(K)}
        $
    \end{algorithm}

%---------------------------------------------------%
%                  hot-spots                         %
%---------------------------------------------------%
\subsection{Hot-spots Detection and Localization}
\label{sec: hot-spots detection and localization}

   In this subsection, we discuss the detection and localization of the hot-spots.
   For the ease of presentation, we first discuss the hot-spots detection, i.e., detect when a hot-spot occurs in Subsection \ref{sec: hot-spots detection}.
   Then, in Subsection \ref{sec: hot-spots localization}, we consider the localization of the hot-spots, i.e., determine which states and which type of diseases are involved for the detected hot-spots.

\subsubsection{Hot-spots Detection}
\label{sec: hot-spots detection}

    To detect when the hot-spots occur, we propose to monitor the time series of the Pearson residuals of count data,
    $
      (\mathbf{y}_t(\lambda) - \widehat{\mymu}_t(\lambda))/\sqrt{\mymu_t(\lambda)},
    $ over time $t = 1, \ldots, n_3.$
    If it has a shift toward the direction of the estimated hot-spots $\myh_t$, then it is highly likely that there is a hot-spot.
    Mathematically speaking, we develop a control chart based on the following hypothesis test problem:
    \begin{equation}
    \label{eq:chaneg_hypothesis_testing}
      H_{0}: \;
      \myr_t(\lambda) = 0 \;\;\;
      \text{v.s.} \;\;\;
      H_{1}:\;
      \myr_t(\lambda)
      =
      \delta \widehat{\myh}_{t}(\lambda)
      \;\;\;
      (\delta>0),
    \end{equation}
    where
    $
        \mathbf{r}_t(\lambda)
        =
        \mathbf{y}_t(\lambda)
        -
        \widehat{\mymu}_t(\lambda)
    $
    is the residual  after removing the global trend mean under the penalty parameters $\lambda$ and the vector
    $
      \widehat{\mymu}_{t}(\lambda)
      =
      \text{vec}
      \left(
        \widehat{\mathcal U}_{::t}(\lambda)
      \right),
      \widehat{\myh}_{t}(\lambda)
      =
      \text{vec} \left(\widehat{\mathcal H}_{::t}(\lambda) \right)
    $
    are the estimated global trend men and local hot-spots in $t$-th year.
    Here, we add $(\lambda)$ to emphasize that,
    $
      \widehat{\mymu}_{t}(\lambda), \widehat{\myh}_{t}(\lambda)
    $
    are the global trend mean and local hot-spots estimation under penalty parameter $\lambda$ respectively.

    The motivation of the above hypothesis test is articulated as follows.
    When there are no hot-spots,  the residual $\myr_t(\lambda)$ is exactly the model noises.
    However, when hot-spots exist, the residual $\myr_t$ includes both hot-spots and noises.
    By including the hot-spots information of $\widehat{\myh}_{t}(\lambda)$ in the alternative hypothesis, we hope to provide a direction in the alternative hypothesis space, which allows one to construct a test with more power \cite[see][]{zou2009multivariate}.

    Next, we construct the likelihood ratio test in the above-mentioned hypotheses testing problem.
    By \cite{hawkins1993regression},  the test statistics monitoring upward shift is
    $$
      P_{t}^{+}(\lambda)
      =
      \widehat{\myh}_{t}^{+}(\lambda)^\top\; \mathbf{r}_{t}(\lambda)
      \bigg/
      \sqrt{\widehat{\myh}_{t}^{+}(\lambda)^\top \;\widehat{\myh}_{t}^{+}(\lambda) }
    $$
    where
    $
      \widehat{\myh}_{t}^{+}(\lambda)
    $
    only takes the positive part of
    $
      \widehat{\myh}_{t}(\lambda)
    $
    with other entries as zero, because our objective is to detect positive hot-spots.
    The superscript ``+'' emphasizes that we aim at detecting upward shift.
    In other words, we focus on the hot-spots that have increasing means, partly because increasing infectious rates are generally more harmful to the societies and communities.
    If one is also interested in detecting decreasing mean shifts, one could modify it by using a two-sided test.

    It remains to discuss how to choose $\lambda$ suitably in our test.
    We propose to follow \cite{zou2009multivariate} to calculate a series of $P_{t}^{+}(\lambda)$ under different combination of
    $
      \lambda
      \in
      \Gamma
      =
      \{
        \lambda^{(1)},\ldots, \lambda^{(n_{\lambda})}
      \}
    $
    and then select the $\lambda$ with the largest power.
    The final chosen test statistics, denoted as $\widetilde{P}_{t}^{+}(\lambda_{t}^*) $, can be computed by
    \begin{equation}
    \label{equ: most power}
      \widetilde{P}_{t}^{+}(\lambda_{t}^*)
      =
      \max_{\lambda \in \Gamma}
      \frac{P_{t}^{+}(\lambda)-E(P_{t}^{+}
      (\lambda))}{\sqrt{ \text{Var}(P_{t}^{+}(\lambda))}},
    \end{equation}
    where
    $
      E(P_{t}^{+}(\lambda)),
    $
    $
      \text{Var}(P_{t}^{+}(\lambda))
    $
    respectively are the mean and variance of
    $
      P_{t}^+(\lambda)
    $
    under $H_{0}$ (e.g., for phase-I in-control samples).
    Here $\lambda_{t}^* \in \Gamma$ is the penalty parameter maximizing the above equation.

    With the test statistic available, we detect when hot-spots occur based on the widely used Cumulative Sum (CUSUM) Control Chart \cite[see][]{page1954continuous, lorden1971procedures}.
    At each time $t,$ we  recursively compute the CUSUM statistics as
    \begin{equation}
    \label{equ: Wt}
      W_{t}^{+}
      =
      \max
      \{
      0,W_{t-1}^{+}+\widetilde{P}_{t}^{+}(\lambda_{t}^{*}) -  d^*
      \},
    \end{equation}
    with the initial value $W_{t=0}^{+}=0$, where $d^*$ is a constant and can be chosen according to the degree of the shift that we want to detect.
    Then we declare that a hot-spot might occurs whenever $W_{t}^{+} > L$  for some pre-specified control limit $L.$

    Note that the CUSUM statistics  $W_{t}^{+}$ leads to the optimal control chart to detecting a mean shift from $\mu_0$ to $\mu_1 = 2 d^* - \mu_0$ for normally distributed data \cite[see][]{lorden1971procedures}.
    When the data are not normally distributed, the optimality properties might not hold, but it can still be a reasonable control chart.
    Also, it is important to choose the control limit $L$  in the CUSUM control chart suitably, and the detailed discussion will be presented in Section \ref{sec: simulation} for our simulation studies and in Section \ref{sec: case study} for our case study.

\subsubsection{Hot-spots Localization}
\label{sec: hot-spots localization}
    In this subsection, we discuss how to localize the hot-spots if the CUSUM control chart in \eqref{equ: Wt} raises an alarm at year $t^*$.
    In other words, we want to determine where and which infectious rates may account for the hot-spots.
    To do so, we propose to utilize the matrix $\widehat{\mathcal H}_{::t^*}(\lambda_{t^*})$, which is the hot-spots estimation in $t^*$-th year.
    If the $(i,j, t^*)$-th entry in $\widehat{\mathcal H}_{::t^*}(\lambda_{t^*})$ is non-zero, then we declare that there is a  hot-spot for the $j$-th type of disease in the $i$-th state at the $t^*$-th year.

    The mathematical procedure to derive
    $
      \widehat{\mathcal H}_{::t^*}(\lambda_{t^*})
    $
    is as follows.
    First,
    $
      \widehat{\mathcal H}(\lambda_{t^*})
    $
    is the tensor format of
    $
      \widehat{\myh}(\lambda_{t^*})
      =
      \myB_h \widehat{\mytheta}_h(\lambda_{t^*}),
    $
    where
    $
      \widehat{\mytheta}_h(\lambda_{t^*})
    $
    is the minimizer in \eqref{equ: objective function1} with penalty parameter as $\lambda_{t^*}$.
    Second,
    $
      \widehat{\mathcal H}_{::t^*}(\lambda_{t^*})
    $
    is the $t^*$-th frontal slices of
    $
      \widehat{\mathcal H}(\lambda_{t^*}).
    $

    It is possible that this approach might lead to a relatively high false positive rate (FPR), since some non-zero entries might not be statistically significant.
    Two possible ways to improve our approach are (1) to conduct the significant test, or (2) to set up a pre-specified threshold and only keep the positive entries that are larger than the threshold.
    In our paper, we use the second approach, i.e., set up a pre-specified threshold for hot-spots localization.
    As for the selection of the threshold, there are several choices:
    \begin{itemize}
      \item Hard-thresholding: if the $(i,j, t^*)$-th entry in the refined estimation of hot-spots
            $
              \widehat{\dot{\mathcal H}}_{::t^*}(\lambda_{t^*})
              =
              \widehat{\mathcal H}_{::t^*}(\lambda_{t^*})
              \mathbbm 1\{ \widehat{\mathcal H}_{::t^*}(\lambda_{t^*}) > h_{\text{hard}} \}
            $
            is larger than zero, then we declare there is a hot-spot in the $j$-th type of disease in the $i$-th state at the $t^*$-th year.
      \item Soft-thresholding: if the $(i,j, t^*)$-th entry in the refined estimation of hot-spots
            $
              \widehat{\dot{\mathcal H}}_{::t^*}(\lambda_{t^*})
              =
              \max\{\widehat{\mathcal H}_{::t^*}(\lambda_{t^*}) - h_{\text{soft}}, 0\}
            $
            is larger than zero, then we declare there is a hot-spot in the $j$-th type of disease in the $i$-th state at the $t^*$-th year.
      \item Order-thresholding: if the $(i,j, t^*)$-th entry in the refined estimation of hot-spots
            $
              \widehat{\dot{\mathcal H}}_{::t^*}(\lambda_{t^*})
              =
              \widehat{\mathcal H}_{::t^*}(\lambda_{t^*})
            \mathbbm 1\{ \widehat{\mathcal H}_{::t^*}(\lambda_{t^*}) \geq h_{\text{order}, r} \}
            $
            is larger than zero, then we declare there is a hot-spot in the $j$-th type of disease in the $i$-th state at the $t^*$-th year.
            Here $h_{\text{order}, r}$ is the $r$-th largest order statistics of
            $
              \{
              \widehat{\mathcal H}_{i,j,t^*}(\lambda_{t^*})
              :
              \forall (i,j), \; s.t. \; \widehat{\mathcal H}_{i,j,t^*}(\lambda_{t^*}) >0
              \}.
            $
    \end{itemize}
    
    In this paper, for simplicity, we use the order-thresholding as an illustration, since it has been used in other contents, such as outlier detection \cite[see][]{nagaraja1982some}, testing problem \cite[see][]{kim2010order}, etc..

%---------------------------------------------------%
%                simulation                         %
%---------------------------------------------------%
\section{Monte Carlo Simulations}
\label{sec: simulation}
    In this section, we report the numerical simulation results of our proposed method as well as its comparison with several benchmark methods in the literature.
    To better present our results, we divide this section into several subsections.
    Subsection \ref{sec: data generation mechanism} includes the data generation mechanism for our simulation studies, and Subsection \ref{sec: compared benchmarks} presents the benchmark methods for the comparison purpose.
    The performance of hot-spots detection and localization are reported in Subsection \ref{sec: simulation result}.

\subsection{Data Generation Mechanism}
\label{sec: data generation mechanism}
   This subsection describes the data generation mechanism we use in the numerical studies.
   In our simulation, the $(i,j,t)$-th entry of tensor $\mathcal Y \in \mathbb R^{n_1 \times n_2 \times n_3}$ is generated as follows:
   \begin{equation}
   \label{equ: sim generative model}
   \left\{
   \begin{array}{lcl}
       \mathcal Y_{i,j,t}
       & \sim &
       \text{Poisson}( \mathcal N_{i,j,t} \mathcal R_{i,j,t} )\\
       \mathcal R_{i,j,t}
       &   =  &
       0.2 + \delta
       \mathbbm 1 \left\{ t  > \tau \right\}
       \mathbbm 1 \left\{(i,j)  \in S \right\}  \\
   \end{array}
   \right..
   \end{equation}

   In the first line of the above generative model,
   $
     \mathcal Y_{i,j,t}, \mathcal N_{i, j, t}, \mathcal R_{i, j, t}
   $
   is the $(i,j,t)$-th entry of tensor
   $
     \mathcal Y, \mathcal N, \mathcal R
     \in
     \mathbb R^{n_1 \times n_2 \times n_3},
   $
   whose physical meaning is the number of infected patients, the population size, and the infectious rate in the $i$-th state in year $t$ under the $j$-th type of disease, respectively.
   In this simulation, to match the dimension of or motivating dataset, we set
   $
     n_1 = 49, n_2 = 10, n_3 = 26.
   $
   
   In the second line of the above generative model in \eqref{equ: sim generative model}, the parameter $\delta$ measures the magnitude of the hot-spots.
   In our paper, we use $\delta = 0.05$ as the small hot-spots and $\delta = 0.2$ as the large hot-spots.
   The reasons we declare $\delta = 0.05$ ($\delta = 0.2$) as small (large) hot-spots is that, the average Kullback-Leibler divergence 
   \cite[see][]{kullback1951information}
   %\cite[see][]{kullback1951information, kullback1997information} 
   is 0.3474 (4.6388).
   And $\mathbbm 1 \{ x \in A \}$ is an indictor function, i.e., $\mathbbm 1 \{ x \in A \} = 1$ if $x \in A$ and $\mathbbm 1\{ x \in A \} = 0$ otherwise.
   The first indictor function $\mathbbm 1\{ t > \tau\}$ means that the hot-spots only happens after the $\tau$-th year.
   And in our paper, we set $\tau = 15$.
   The second indictor function $\mathbbm 1\left\{(i,j)  \in S \right\}$ means that the hot-spots only occurs in certain spatial location set $S$.
   In our paper, $S$ is randomly sampled without replacement in each simulation and the number of elements of $S$ only accounts for $10\%$ of the total number of elements of $\vartheta_{h,:,:,t}$.

  It remains to discuss how to model the population size $\mathcal N_{i,j,t}$. We propose to follow the common practice to model the curve of the population growth by the logistic model \citep[see][]{pinheiro2006mixed}:
   \begin{equation}
   \label{equ: sim -- population generation}
     \mathcal N_{i,j,t}
     =
     \frac{\phi_{i,j, 1}}{1+\exp[-( t -\phi_{i,j, 2} ) / \phi_{i,j, 3} ]} + \epsilon_{i,j,t}
   \end{equation}
   where $E(\epsilon_{i,j,t}) = 0$ and $Var(\epsilon_{i,j,t}) = \sigma^2$.
   Here $\phi_{i,j,1}$ indicates an asymptotic upper limit of population size in the $i$-th state, $\phi_{i,j,2}$ is  the middle point of the S-shaped curve in the $i$-th state, and $\phi_{i,j,3}$ is the scale adjustment of time periods in the $i$-th state. 
   
   In our numerical studies,
   we consider two scenarios for the population size $\mathcal N$:
   (1) the population size with increasing trend and
   (2) the population size with decreasing trend.
   In the above two scenarios, we assume one state share the uniform characteristic in the population over different type of diseases, i.e.,
   $
     \phi_{i,j_1,1} = \phi_{i,j_2,1},
     \phi_{i,j_1,2} = \phi_{i,j_2,2},
     \phi_{i,j_1,3} = \phi_{i,j_2,3}
   $
   and
   $
     \mathcal N_{i,j_1, t} = \mathcal N_{i,j_2,t}
   $
   for any $j_1 \neq j_2$.
   In our paper, we first fit \eqref{equ: sim -- population generation} to the observed population sizes in our motivating dataset by a nonlinear least-squares method (we treat year 1993 as time 1, and the population sizes are in the units of 10,000).
   Using the mathematical software Matlab R2018b, the estimated parameters for the logistic model \eqref{equ: sim -- population generation} are summarized in Appendix \ref{appendix: logistic model for population}.
   Figure \ref{fig: sim -- population fitting} plots the actual observed population sizes and the estimated growth curve in New Mexico during 1993-2018.
   From the plot, one sees that the two curves are close to each other, implying that the logistic model is reasonable in our application.

   In our simulation, we will consider the following two extreme scenarios of the population sizes, depending on whether the population size in all states are increasing or decreasing.
   This might be unrealistic, but can indicate the performance in more practical scenarios, where the population can be the interested study group, such as specific age/gender/race sub-groups.
   \begin{itemize}
     \item There is an increasing trend in the population.
           The increasing population is generated from
           $
             \mathcal N_{i,j,t}
             =
             \frac{\hat\phi_{i,j, 1}}{1+\exp[-( t -\hat\phi_{i,j, 2} ) / \hat\phi_{i,j, 3} ]}
             +
             \epsilon_{i,j,t},
           $
           where $\{(\hat\phi_{i,j, 1}, \hat\phi_{i,j, 2}, \hat\phi_{i,j, 3})\}$ is estimated by fitting the real-world population into \eqref{equ: sim -- population generation}.
     \item There is a decreasing trend in the population.
           The decreasing population is generated from
           $
             \mathcal N_{i,j,t}
             =
             \frac{\hat\phi_{i,j, 1} / a_{i,j}}{1+\exp[( t -\hat\phi_{i,j, 2} ) / \hat\phi_{i,j, 3} ]} +
             1
             +
             \epsilon_{i,j,t},
           $
           where $\{(\hat\phi_{i,j, 1}, \hat\phi_{i,j, 2}, \hat\phi_{i,j, 3})\}$ is estimated by fitting the real-world population into \eqref{equ: sim -- population generation}.
           Here the parameter $\{a_{i,j}\}$ is  necessary to make sure that the simulated initial population size is the same as the observed value.
   \end{itemize}

   \begin{figure}
     \centering
     \begin{tabular}{cc}
       \includegraphics[width=0.45\textwidth]{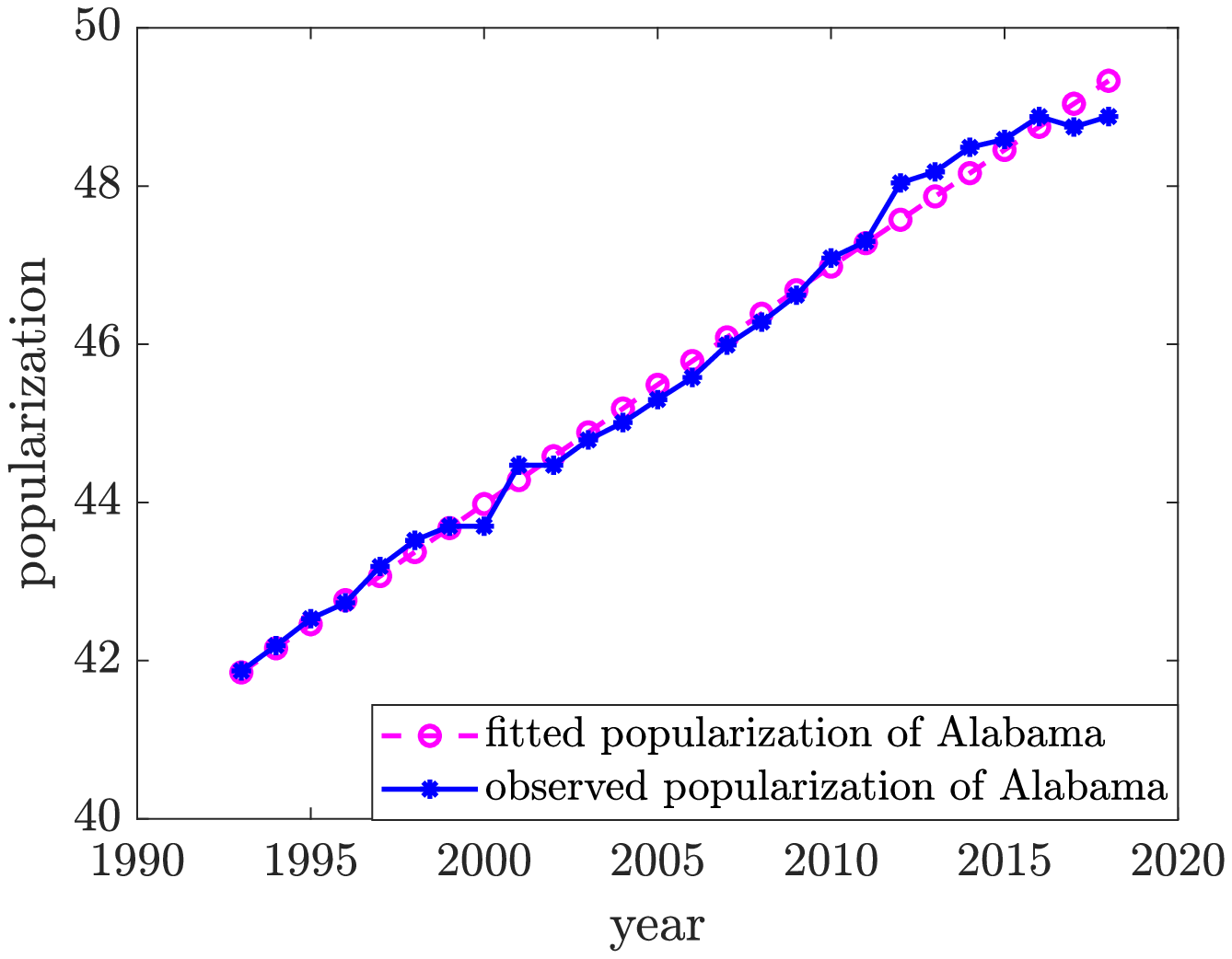} &
       \includegraphics[width=0.45\textwidth]{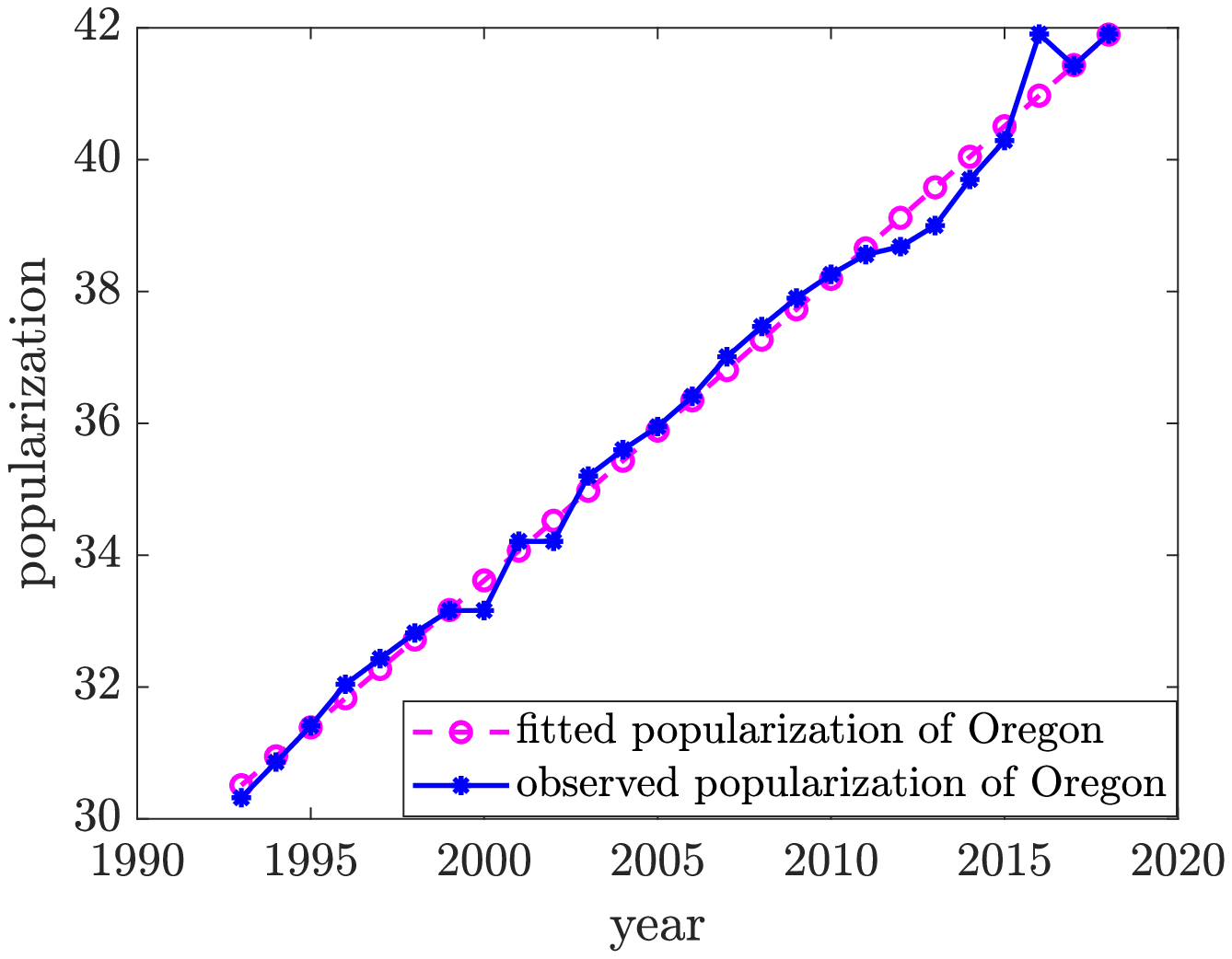} \\
       (a) Alabama & (b) Oregon\\
     \end{tabular}
     \caption{Population and estimate in Alabama and Georgia during 1993-2018
     \label{fig: sim -- population fitting}}
   \end{figure}

%% -----------------------------------------------------------------------------
%% -----------------------------------------------------------------------------
%% -----------------------------------------------------------------------------
\subsection{Benchmark Methods}
\label{sec: compared benchmarks}
% -----------------------------------------------------------------------------
% -----------------------------------------------------------------------------
% -----------------------------------------------------------------------------

    In this subsection, we present the description and implement of five benchmark or baseline methods that will be used to compare with our method.

    The first benchmark method is the scan statistics method in \cite{neill2006bayesian}, which is a Bayesian extension of Kulldorff's scan statistic and we denote it as ``NMC-scan-stat''.
    The reason for us to choose NMC-scan-stat is that it has large power to detect clusters and the fast runtime \cite{neill2006bayesian}.
    In our paper, we use the a R function called \textit{scan\underline{ }bayes\underline{ }negbin()} from the package \textit{scanstatistics}.
    Because \textit{scan\underline{ }bayes\underline{ }negbin()} can only handle one type of disease one time, we apply \textit{scan\underline{ }bayes\underline{ }negbin()} to ten diseases separately and set the probability of an outbreak as $0.02/10$.
    Because scan-statistics-based method does not give the clear calculation of \textit{average run length under the in-control status} ($\text{ARL}_0$) and $\text{ARL}_1$, so we can only use the probability of an outbreak as $0.02/10$ to define the control limit to achieve similar $\text{ARL}_0$ with other benchmarks.

    The second benchmark method is the smooth sparse decomposition method proposed by \cite{SSD}, and we denote it as ``YPS-SSD''.
    The YPS-SSD method assumes the raw data (or the raw data after transformation) follows normal distribution and decompose it into the functional mean, sparse anomalies, and random noises.
    Then YPS-SSD detect the hot-spots by monitoring  the model residual through Shewhart control chart.
    Although the main idea in YPS-SSD looks similar to our proposed PoSSTenD, there are two significant difference.
    The first difference is that, YPS-SSD focus the normal distributed tensor, while our PoSSTenD method focus on the Poisson distributed tensor.
    In other words, YPS-SSD ignore the effect of the population size, and aims at detecting the hot-spots in the number of infected events.
    However, for our PoSSTenD method, we take the population size into consideration and target on detecting the hot-spots in infectious rates.
    The second difference is that, YPS-SSD utilizes Shewhart control chart to monitor the hot-spots, which is not sensitive to the small hot-spots.
    Yet, our proposed PoSSTenD method uses CUSUM control chart, which has more capability to detect small hot-spots.
    For the selection of basis in YPS-SSD and our PoSSTenD method, we use the same basis for a fair comparison, i.e.,
    $
     \myB_{m,1} \in \mathbb R^{n_1 \times p_1},
     \myB_{m,2} \in \mathbb R^{n_2 \times p_2},
     \myB_{m,3} \in \mathbb R^{n_3 \times p_3}
   $
   are B-spline basis with order 3 and over equally spaced knots
   $\{1, 8, 15,\ldots, 50\}$,
   $\{1, 9.1667, 17.3333, \ldots, 50\}$,
   $\{1, 9.1667, 17.3333, \ldots, 50\}$,
   respectively \cite[see][]{BsplineToolbox} and $\myB_{h,1}, \myB_{h,2}, \myB_{h,3}$ are all set as identity matrix.

    The third benchmark method is the Lasso-based method proposed by  \cite{zou2009multivariate} and we denote it as ``ZQ-Lasso''.
    The main idea of ZQ-Lasso is to integrate the multivariate Exponentially Weighted Moving Average (EWMA) charting scheme.
    Under the assumption that the hot-spots are sparse, the Lasso model is applied to the EWMA statistics.
    If the Mahalanobis distance between the expected response (the Lasso estimator) and observed values is larger than a pre-specified control limit, temporal hot-spots are detected, with non-zero entries of the Lasso estimator are declared as spatial hot-spots.
    For the control limits and the penalty parameters of the Lasso-based method, we use the same criterion as our proposed PoSSTenD method.

    The fourth benchmark is the dimension-reduction method proposed by \cite{PCA}, and we denote it as ``DBS-PCA''.
    The DBS-PCA method uses PCA to extract a set of uncorrelated new features that are linear combinations of original variables.
    Note that DBS-PCA fails to localize the spatial hot-spots, and it can only detect the temporal change-point when the PCA-projected Mahalanobis distance is larger than a pre-specified control limit.
    For this control limit, we set it by using the same criterion as our proposed PoSSTenD method.
    In both our simulations and case study, we select three principle components, since they can explain more than $90\%$ cumulative percentage of variance (CPV).

    Finally, the fifth benchmark is the traditional $\text{T}^2$ control chart \cite[see][]{T2} method and we denote it as ``$\text{T}^2$''.
    The control limit of $\text{T}^2$ is set by using the same criterion as our proposed PoSSTenD method.
    Since the $\text{T}^2$ control chart method is a well-defined method, we skip the detailed description, and more details  can be found in \cite{T2}.

%% -----------------------------------------------------------------------------
%% -----------------------------------------------------------------------------
%% -----------------------------------------------------------------------------
\subsection{Simulation Results}
\label{sec: simulation result}
%% -----------------------------------------------------------------------------
%% -----------------------------------------------------------------------------
%% -----------------------------------------------------------------------------
    In this subsection, we compare our proposed PoSSTenD method with five benchmark methods with the focus on the performance of  hot-spots detection in Section \ref{sec: sim -- comparision of hot-spots detection} and localization \ref{sec: sim -- comparision of hot-spots localization}.
    The five benchmark methods are NMC-scan-stat method proposed by \cite{neill2006bayesian}, YPS-SSD proposed by \cite{SSD}, ZQ-Lasso proposed by \cite{zou2009multivariate}, DBS-PCA proposed by \cite{PCA}, and the traditional Hotelling $\text{T}^2$ control chart reviewed in \cite{T2}.
    All simulation results below are based on $1000$ Monte Carlo replications.

\subsubsection{Comparison of Hot-spots Detection}
\label{sec: sim -- comparision of hot-spots detection}

    In this section, we compare the performance of hot-spots detection, i.e., the detection delay after the occurrence of hot-spots.
    The criterion we use is $\text{ARL}_1$.
    Because $\text{ARL}_1$ measures the delay after the change occurs, the smaller the $\text{ARL}_1$, the better detection performance.
    The comparison results are visualized in Figure \ref{fig: sim -- ARL1 plot} (the numerical numbers to generate this plot are available in Appendix \ref{appendix: ARL1 data}).
    
    For our proposed PoSSTenD method (marked with blue), it has a small $\text{ARL}_1$ no larger than $4.2780$ under both two scenarios.
    This indicates that the PoSSTenD method can provide a rapid alarm after the hot-spots occur, even if there are increasing or decreasing global trends.
    This good performance is mainly due to its ability to remove the global trend and capture the small hot-spots by the CUSUM control chart.
    
    For the NMC-scan-stat method, it is hard to estimate the exact $\text{ARL}_1$ because it focuses on the hot-spots localization, not sequential change point detection.
    So we will not report its $\text{ARL}_1$.
    
    For the YPS-SSD method (marked with red), it can successfully detect hot-spots when there is an increasing trend in population.
    For example, its $\text{ARL}_1$ is around $2$ when $\delta = 0.05$ and there is an increasing trend in population.
    However, when the hot-spots are small and there is a decreasing trend in population, then YPS-SSD has very limited power to detect the hot-spots.
    For example, its $\text{ARL}_1$ is around $11$ when $\delta = 0.05$ and there is a decreasing trend in population.
    If $\text{ARL}_1 \approx 11$, it is bad, because it means the YPS-SSD method is unable to detect any hot-spots among a total of 26 years with hot-spots occurrence after the 15-th year.
    This is because YPS-SSD targets on the detection of hot-spots among the number of infectious people, rather than the infectious rates.
    So when the hot-spots are small and there is a decreasing trend in population, the numbers of infectious people have a very stable trend, which makes it difficult for YPS-SSD to detect the hot-spots.
    
    For ZQ-Lasso (marked with pink), it has a relatively larger $\text{ARL}_1$.
    For example, its $\text{ARL}_1$ is around 9 when there is an increasing trend in the population.
    And its $\text{ARL}_1$ is above 9 when $delta \leq 0.15$ and there is a decreasing trend in the population.
    This is not surprising because ZQ-Lasso is unable to separate the global trend and local hot-spots.
    
    For the DBS-PCA method (marked with green), its $\text{ARL}_1$ is also large compared with other methods.
    Specifically, its $\text{ARL}_1$ is above 8 for both two scenarios.
    And the $\text{T}^2$ method fails to detect the hot-spots in all scenarios within the entire $n_3 = 26$ (simulated) years.
    The reason for the unsatisfying results of DBS-PCA and $\text{T}^2$ is that they are designed based on the global mean change, which cannot take into account the non-stationary global mean trend and the sparsity of the hot-spots.
    
    In conclusion, our proposed PoSSTenD method has short detection delay, no matter there is increasing/decreasing population trend.
    Yet, for the other benchmarks, their detection performance is not robust to the population trend.
    Especially when there is a decreasing population trend, they all have relatively long detection delays. 
    
    \begin{figure}[htbp]
      \centering
      \begin{tabular}{cc}
        \includegraphics[width=0.43\textwidth]{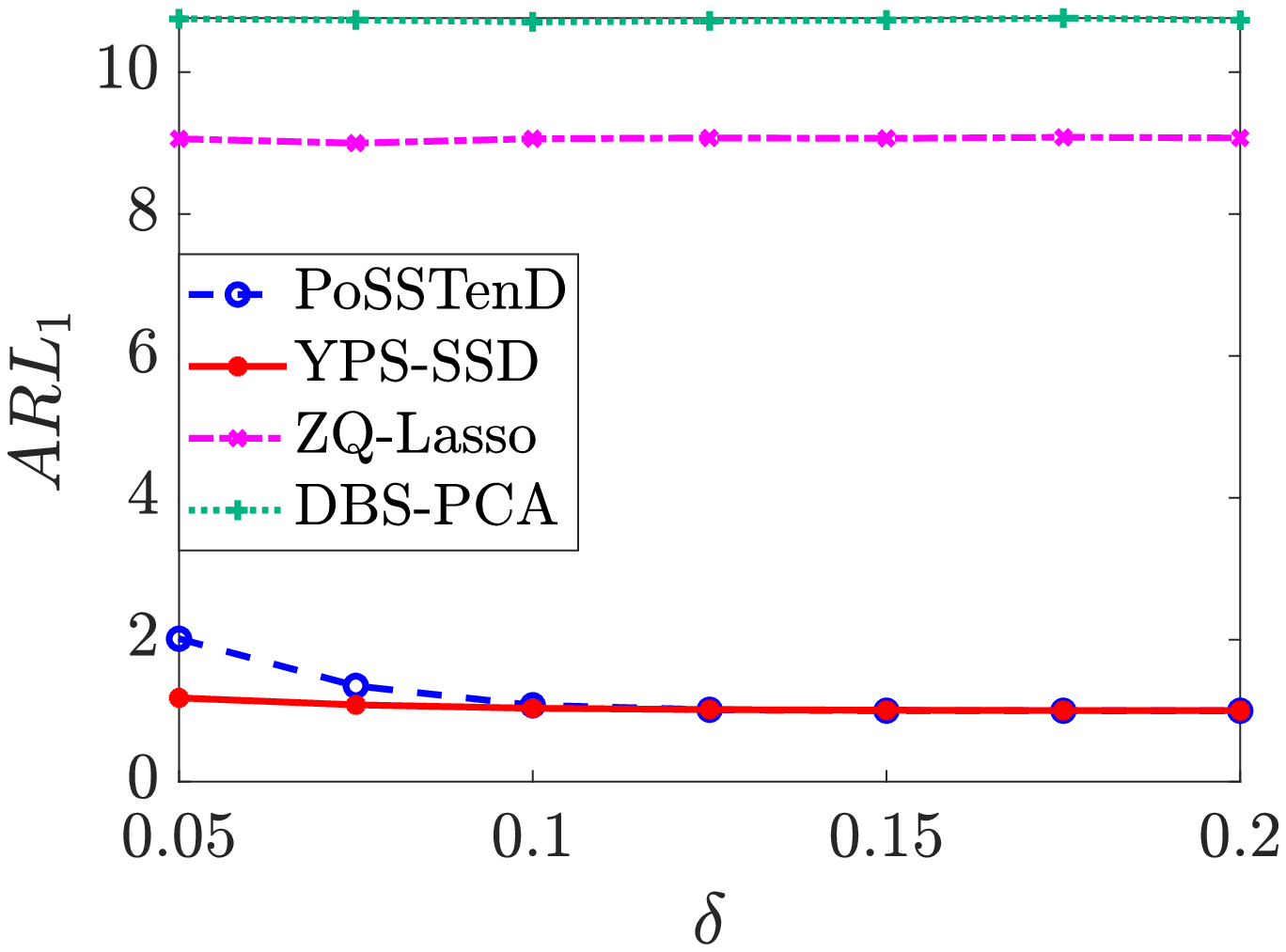} &
        \includegraphics[width=0.43\textwidth]{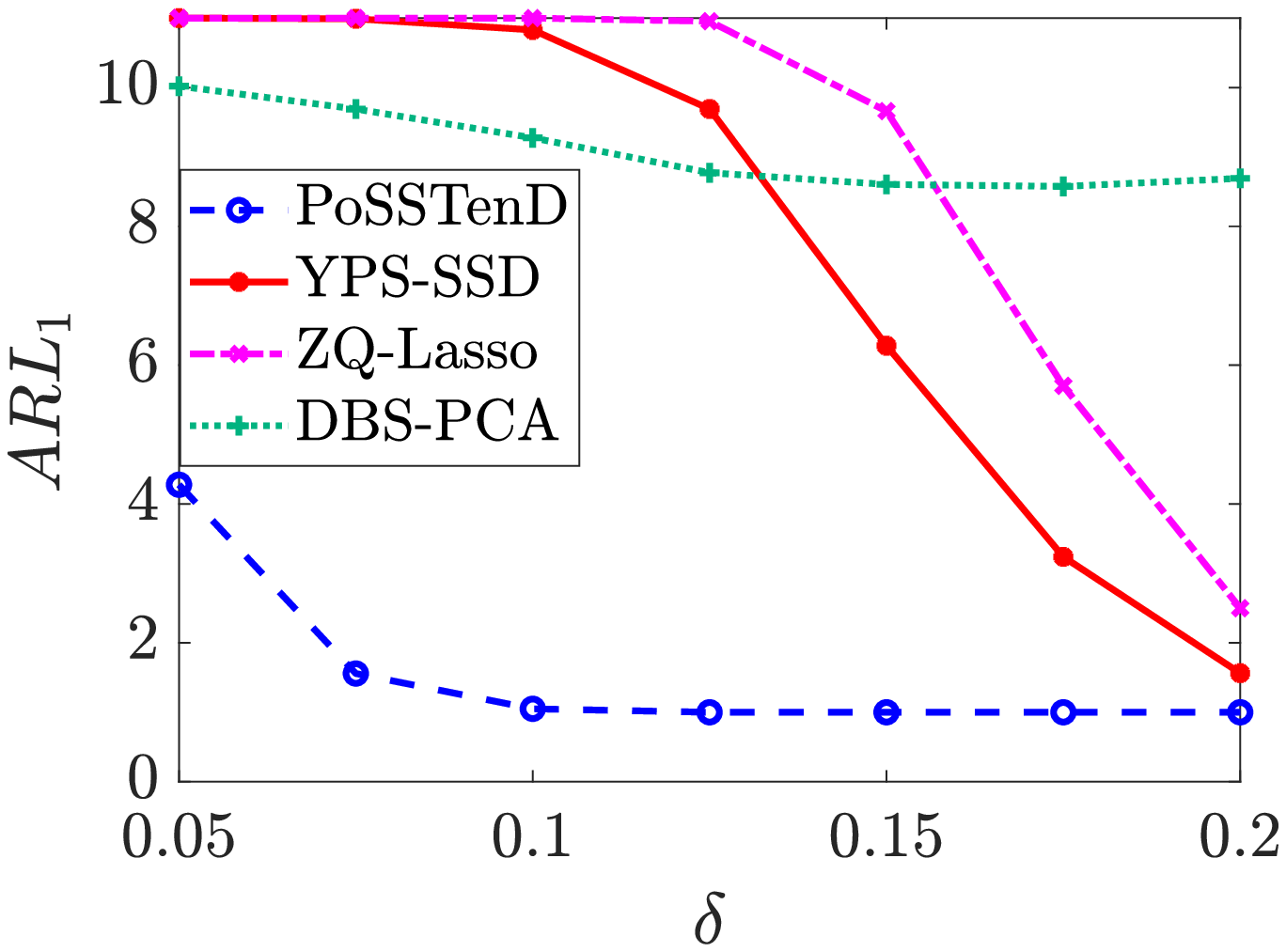} \\
        (a) population with an increasing trend &
        (b) population with a decreasing trend
      \end{tabular}
      \caption{The $\text{ARL}_1$ plot of our proposed PoSSTenD method, and the three benchmark methods, i.e., YPS-SSD method, ZQ-Lasso method and DBS-PCA method.
      \label{fig: sim -- ARL1 plot}}
    \end{figure}

\subsubsection{Comparison of Hot-spots Localization}

\label{sec: sim -- comparision of hot-spots localization}

    In this section, we compare the performance on hot-spots localization.
    The evaluation criterion are:
    (1) precision, defined as the proportion of detected hot-spots that are true hot-spots;
    (2) recall, defined as the proportion of the hot-spots that are correctly identified;
    (3) F-measure, a single criterion that combines the precision and recall by calculating their harmonic mean.
    In this paper, we classify a overall good method if it has a high precision, together with a high recall and a high F-measure.
    The localization performance of the compared six methods is available in Table \ref{table: simulation -- precison}.

    For our proposed PoSSTenD method, its localization performance is satisfactory no matter the population size is increasing or decreasing over time.
    For instance, when the population size is decreasing over time and $\delta = 0.2$, our method has $81.13\%$ precision, $54.52\%$ recall, and $64.41\%$ F-measure, which outperforms the benchmark methods.
    The only weakness of our proposed PoSSTenD method is that it has a relatively high false positive rate when $\delta$ is small.
    For example, its precision, recall, and F-measure are $46.10\%$, $15.55\%$ and $23.25\%$ when $\delta = 0.05$ and the population is increasing. 
    Though the above performance is not excellent, it is still satisfying given the challenging to detect sparse hot-spots among high-dimensional data, let along there are time-varifying global trend mean.
    
    For the NMC-scan-stat method, it has much lower precision than our proposed PoSSTenD method.
    For example, its precision, recall and F-measure are $2.50\%, 26.69\%$ and $4.51\%$ when $\delta = 0.2$ and the population is decreasing. 
    Its limited power is mainly caused by the property of scan statistics, which tends to detect clustered hot-spots.
    Accordingly, the NMC-scan-stat method detects fewer hot-spots and misses true hot-spots.
    It is worth noting that the precision/recall/F-measure might be overestimated for the NMC-scan-stat method: we record all the Monte Carlo runs, even if it is a false alarm because scan-stat fails to report $\text{ARL}_1$.
    
    For the YPS-SSD method, it has a high false positive rate in all scenarios, especially when the magnitude of hot-spots is small ($\delta = 0.05$).
    For example, its recall are $24.36\%$ when $\delta = 0.2$ and the population is decreasing.
    Besides, we notice that its precision, recall, and F-measure are all zero when $\delta = 0.05$ and there is a decreasing trend in population.
    This is because that the YPS-SSD method fails to detect when the hot-spots occur, and thus fails to further localize where the hot-spots occur.
    
    For the ZQ-Lasso method, it has high precision when the population is increasing, or $\delta = 0.2$.
    However, it also has high false positive rate.
    For example, its precision, recall, F-measure are $99.71\%$, $10.11\%$ and $18.36$ when $\delta = 0.2$ and the population is increasing. 
    We don't think its performance is overall good, since it has unsatisfactory performance in the sense of small recall.
    The physical meaning of the $99.71\%$-precision and $10.11\%$-recall is that, the ZQ-Lasso method detects almost ``everywhere'' as hot-spots, so it has around 100\% precision. 
    However, as shown in Section \ref{sec: data generation mechanism}, there are only 10\% of them are true hot-spots, so the recall is 10\%. 
    In practice, it is not preferred because it gives too much false alarm.
    The aforementioned unsatisfactory performance is due to two reasons.
    First, ZQ-Lasso lacks the ability to separate the global trend mean and local hot-spots, so ZQ-Lasso tends to localize all states as hot-spots which cause lots of false alarm.
    Second, ZQ-Lasso aims at hot-spots localization of the number of infectious people, instead of infectious rates.
    Thus, when there is an increasing trend of the population, ZQ-Lasso mistakenly regards it as the potential hot-spots in the number of infectious people.
    
    For the DBS-PCA and $\text{T}^2$ methods, their localization performance is not reported since they can only detect when hot-spots occur and are incapable to localize where hot-spots occur.
    
    Moreover, we visualize the hot-spots localization results from one Monte Carlo simulation in Figure \ref{fig: simulation map} under the scenario when $\delta = 0.2$ and population is decreasing over time.
    Different rows refer to different types of disease.
    And we select the first five disease as representative.
    In the first column, the blue states are the true hot-spots (true positive), whereas the white states are the normal states (true negative).
    In the second, third, and fourth columns, the red states are the detected hot-spots (true positive + false positive) by our proposed PoSSTenD method, NMC-scan-stat, and YPS-SSD method respectively.
    Different color represents how likely it is hot-spots: the darker red, the more likely it is.
    From Figure \ref{fig: simulation map}, we can see that the NMC-scan-stat method tends to detect clustered hot-spots, however, there is no clear pattern for the hot-spots detected by our proposed PoSSTenD method and YPS-SSD method.
    The difference between our proposed PoSSTenD method and the YPS-SSD method is that we have a lower false positive rate and a higher true positive rate.
    This is because that, our proposed PoSSTenD method is designed to detect and localize hot-spots in the infectious rate, which is aligned with the data generation mechanism introduced in Section \ref{sec: data generation mechanism}.
    However, the YPS-SSD method detects and localizes hot-spots in the number of infectious people, which leads to an inaccurate localization in our simulations.
    
    In conclusion, our proposed PoSSTenD method  has satisfying localization performance no matter the population is increasing or decreasing.
    For the other benchmarks, they either have relatively high false positive rate, or high false negative rate, or lack the ability to localize the hot-spots.
    Besides, they are also not robust to the population trend, especially when there is an increasing population trend. 
    
    \begin{table}[htbp]
    \caption{hot-spots detection and localization performance by five methods: our proposed method, NMC-scan-stat, YPS-SSD, ZQ-Lasso, DBS-PCA and $\text{T}^2$ under the population with neural growth and positive hot-spots
    \label{table: simulation -- precison}}
	\centering
    \begin{adjustbox}{max width=0.95\textwidth}
    \centering
    \begin{threeparttable}
	\begin{tabular}{c|cccc|cccc}
		\hline
		\multicolumn{1}{c|}{ \multirow{2}{*}{methods} } &
        \multicolumn{4}{c|}{population with increasing trend} &
        \multicolumn{4}{c }{population with decreasing trend}\\
		\cline{2-9}
        &\multicolumn{8}{c}{$\delta = 0.05$}\\
		\hline
		\multicolumn{1}{c|}{} &
        precision & recall & F-measure & $\text{ARL}_1$ &
        precision & recall & F-measure & $\text{ARL}_1$ \\
        \hline
		PoSSTenD
		& 0.4610&  \bf{0.1555} & \bf{0.2325} &2.0150
        & \bf{0.3393}&  \bf{0.1945} & \bf{0.2470}&\bf{4.2780}\\
		& (0.0775) & (0.0257) & (0.0381) &(1.4610)
        & (0.1638) & (0.0937) & (0.1187) &(3.6105)\\
		NMC-scan-stat
        &0.0113 &  0.1193& 0.0203 & -
        &0.0110 & 0.1134 & 0.0198 & - \\
		& (0.0144) & (0.1655) & (0.0256) &(-)
        & (0.0141) & (0.1575) & (0.0253) &(-)\\
        YPS-SSD
		&0.4184  & 0.1176& 0.1836&\bf{1.1806}
        &0.0000  & 0.0000& 0.0000&11.0000 \\
		& (0.0678) & (0.0190) & (0.0296) &(0.5138)
        & (0.0000) & (0.0000) & (0.0000) &(0.0000)\\
		ZQ-Lasso
		&\bf{0.9926}  & 0.1003& 0.1822&9.0609
        &0.0000  & 0.0000& 0.0000& 11.0000\\
		& (0.0168) & (0.0013) & (0.0023) &(1.0593)
        & (0.0000) & (0.0000) & (0.0000) &(0.0000)\\
		DBS-PCA
		& - & - & - &10.7515
        & - & - & - &10.0200 \\
		& - & - & - &(0.6330)
        & - & - & - &(2.0995) \\
		$\text{T}^2$
		& - & - & - &11.0000
        & - & - & - &11.0000 \\
		& - & - & - &(0.0000)
        & - & - & - &(0.0000) \\
		\hline
        &\multicolumn{8}{c}{$\delta = 0.2$}\\
        \hline
		PoSSTenD
		&\bf{ 0.7530}& \bf{0.3165}& \bf{ 0.4453}& \bf{1.0000}
        &\bf{0.8113}& \bf{0.5352 }& \bf{ 0.6441}& \bf{1.0000}\\
		& (0.0614) & (0.0269) & (0.0351) &(0.0000)
        & (0.0539) & (0.0428) & (0.0420) &(0.0000)\\
		NMC-scan-stat
        & 0.0211 & 0.2289 & 0.0381 & -
        & 0.0250 &0.2669  & 0.0451 & - \\
		& (0.0162) & (0.1985) & (0.0287) &(-)
        & (0.0176) & (0.2034) & (0.0312) &(-)\\
        YPS-SSD
		&0.5664  & 0.1617& 0.2515 &1.0010
        &0.9376  & 0.2436 & 0.3866& 1.5630\\
		& (0.0678) & (0.0194) & (0.0299) &(0.0316)
        & (0.2136) & (0.0563) & (0.0889) &( 2.1750)\\
		ZQ-Lasso
		&0.9971  & 0.1011& 0.1836 &9.0730
        &0.8990  & 0.0899& 0.1635&2.4990 \\
		& (0.0083) & (0.0016) & (0.0027) &(1.0406)
        & (0.3015) & (0.0301) & (0.0548) &(2.9734)\\
		DBS-PCA
		& - & - & - &10.7330
        & - & - & - &8.6920 \\
		& - & - & - &(0.7939)
        & - & - & - &(3.6631) \\
		$\text{T}^2$
		& - & - & - & 11.0000
        & - & - & - & 11.0000\\
		& - & - & - &(0.0000)
        & - & - & - &(0.0000) \\
		\hline
	\end{tabular}
    \begin{tablenotes}
      \footnotesize
        \item[1] The above results are based on 1000 simulations
    \end{tablenotes}
    \end{threeparttable}
    \end{adjustbox}
    \end{table}

    \begin{figure}[htbp]
      \centering
      \begin{tabular}{ccccc}
        disease 1&
        \includegraphics[width=0.18\textwidth]{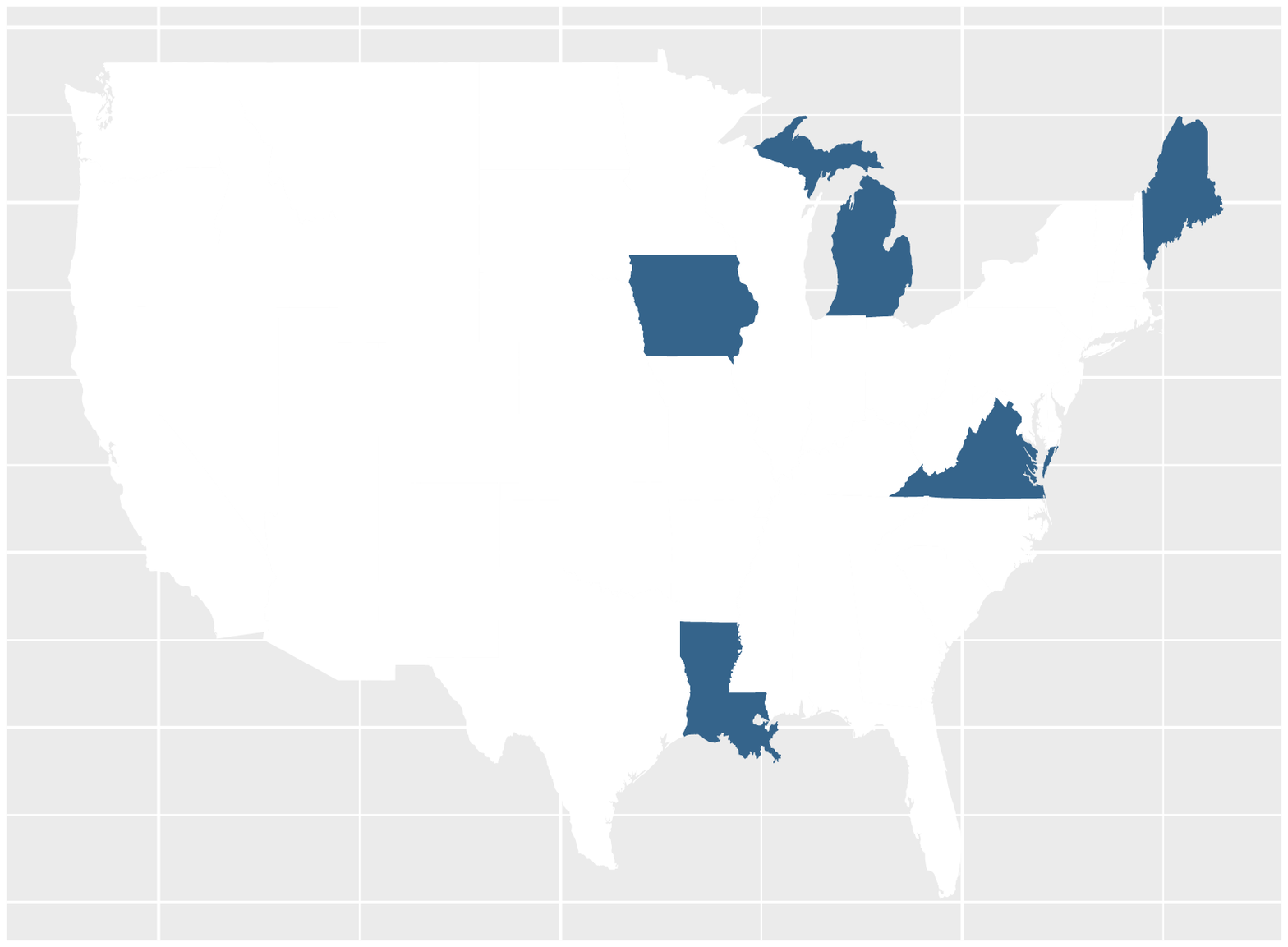} &
        \includegraphics[width=0.18\textwidth]{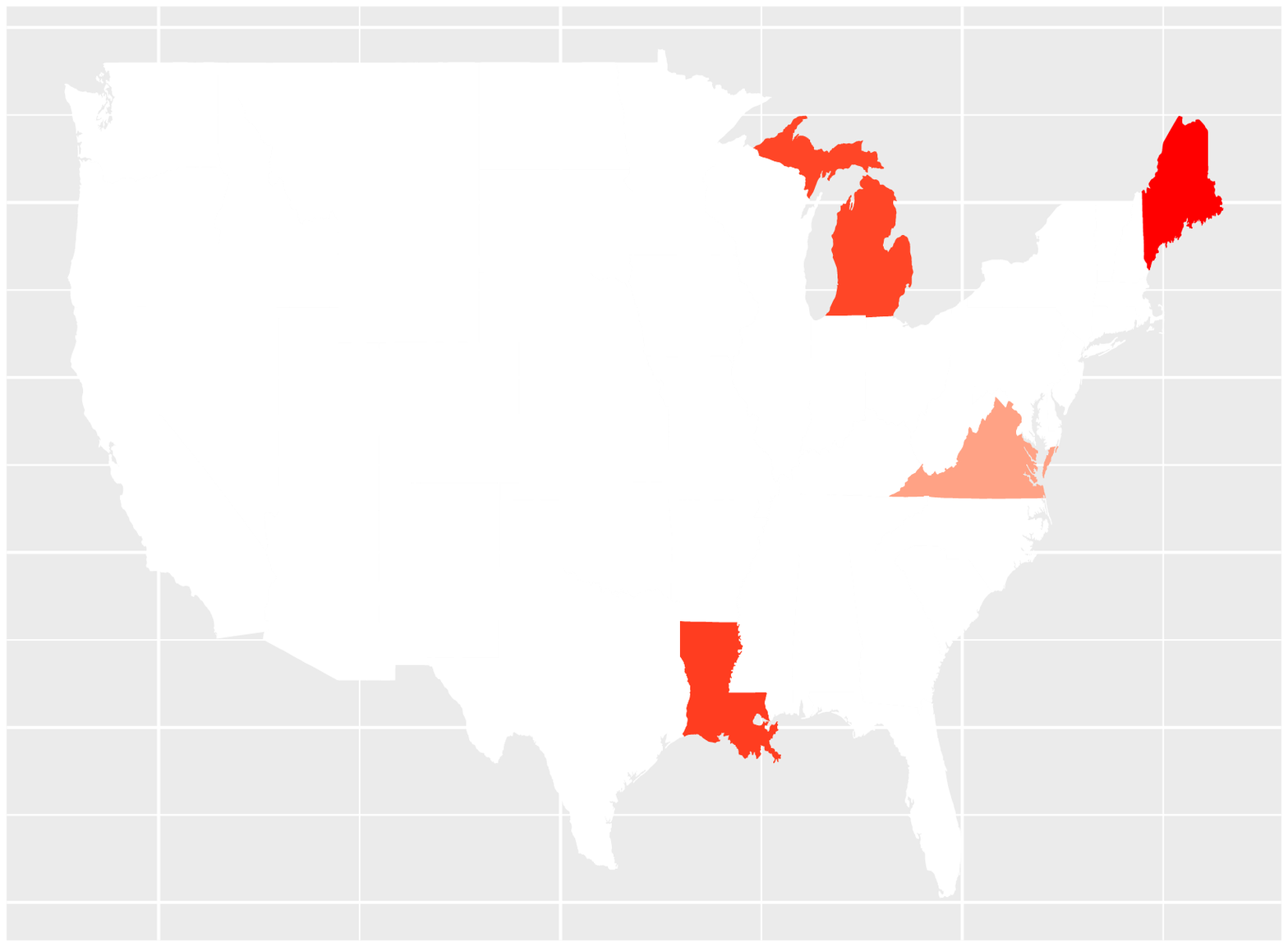}  &
        \includegraphics[width=0.18\textwidth]{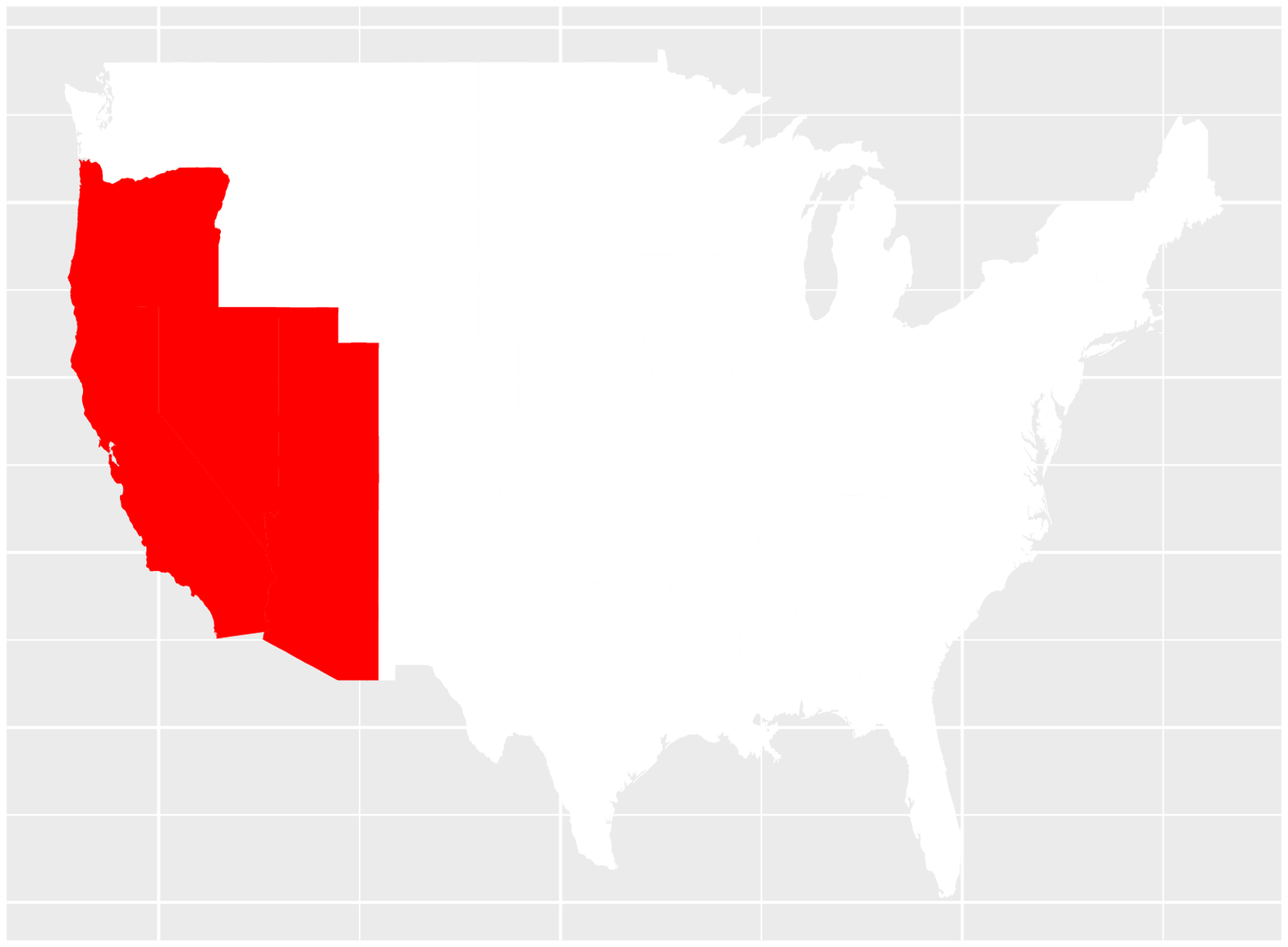} &
        \includegraphics[width=0.18\textwidth]{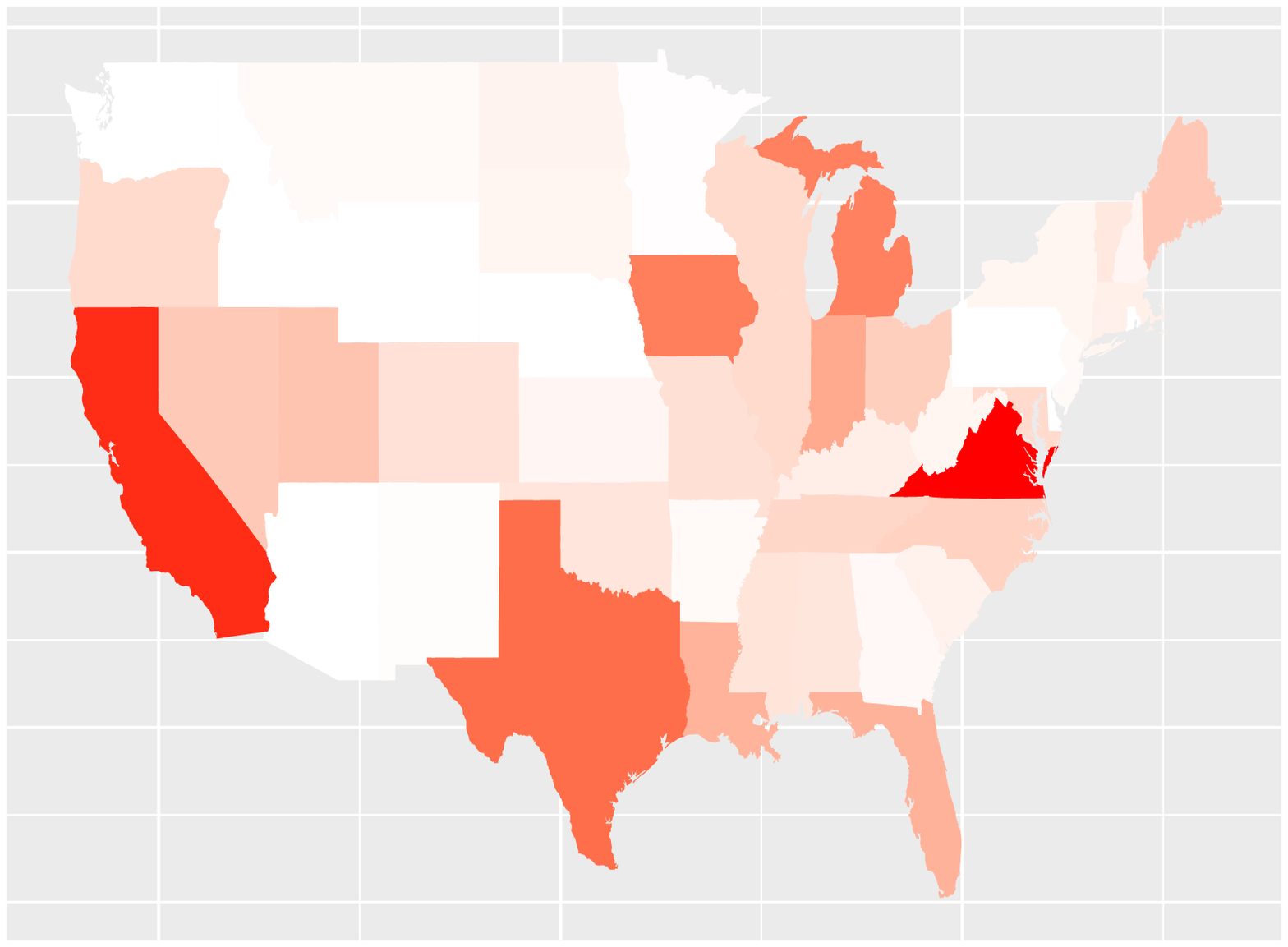}  \\
        &
        (a.1) true &
        (a.2) PoSSTenD  &
        (a.3) NMC-scan-stat &
        (a.4) YPS-SSD \\
        %%%%%%%%%%%%%%%%%%
        disease 2&
        \includegraphics[width=0.19\textwidth]{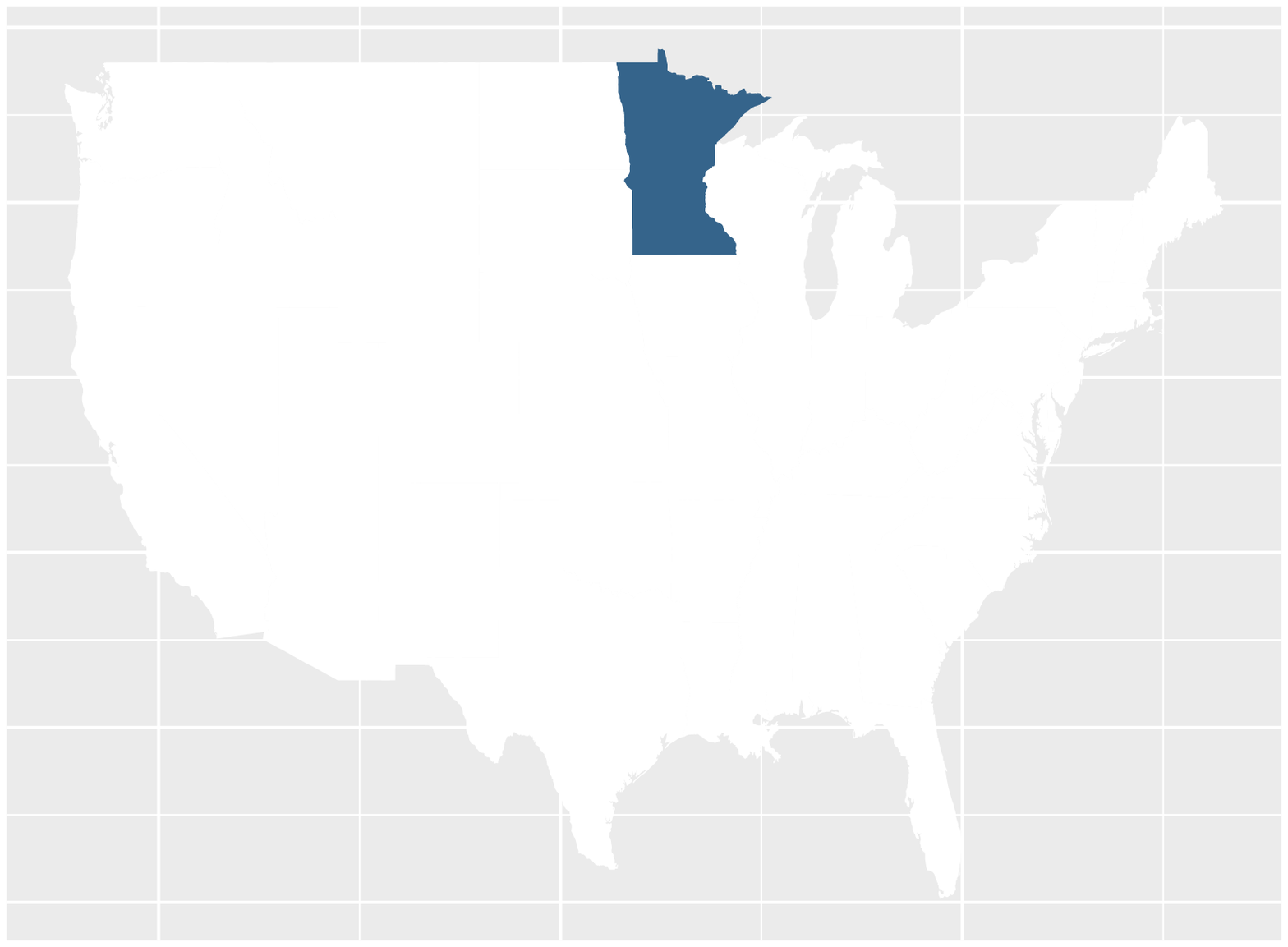} &
        \includegraphics[width=0.19\textwidth]{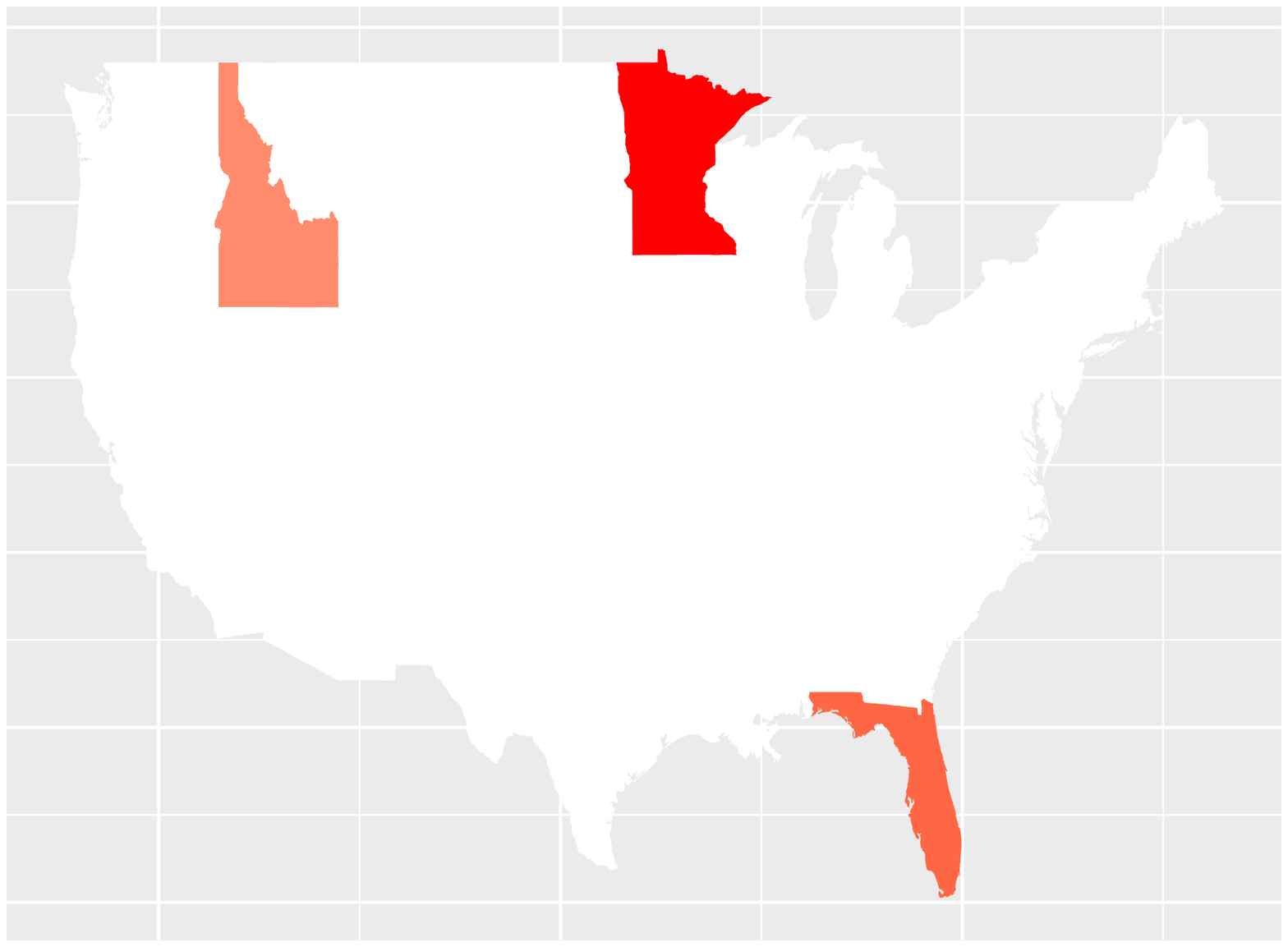}  &
        \includegraphics[width=0.19\textwidth]{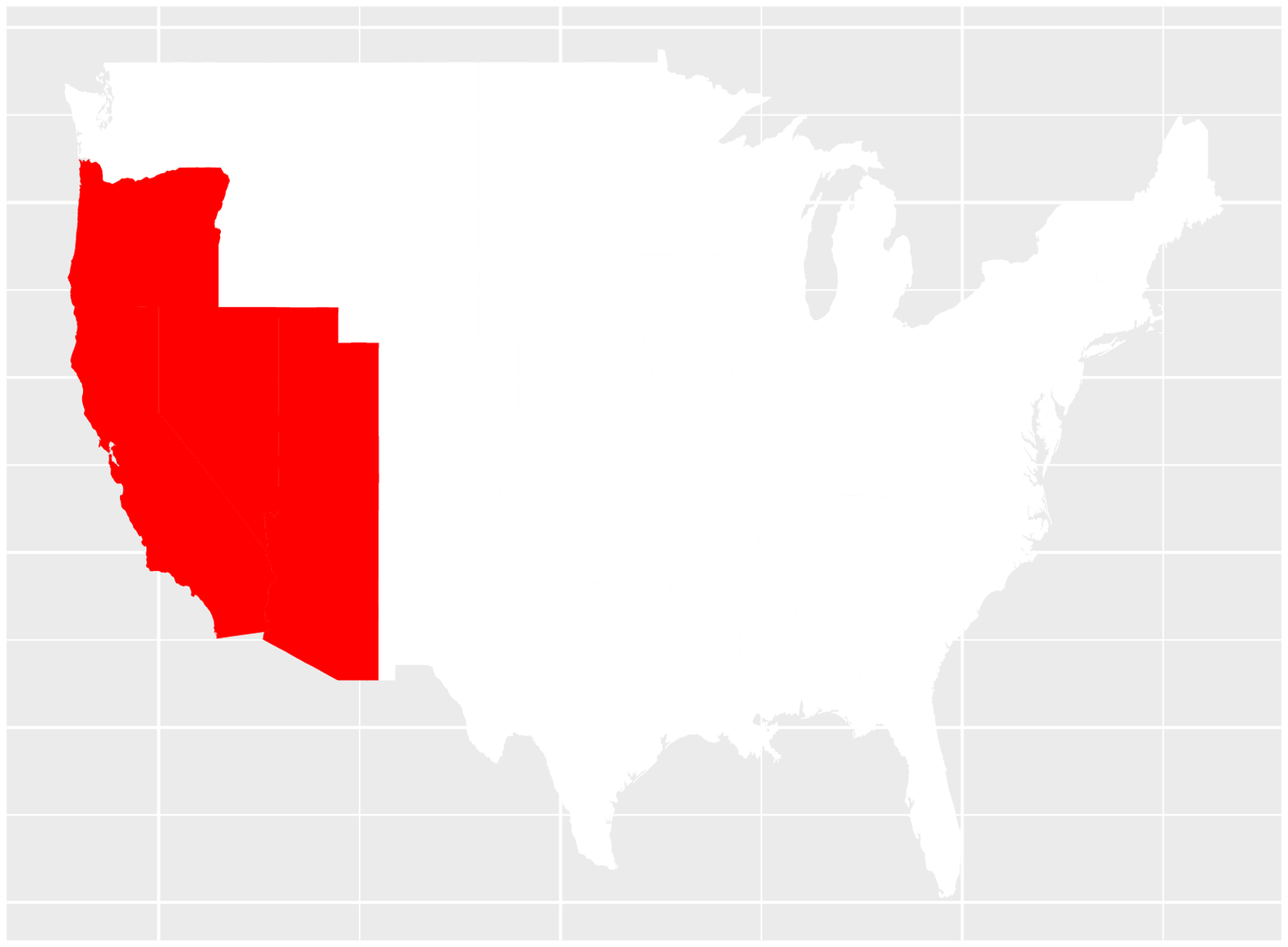}  &
        \includegraphics[width=0.19\textwidth]{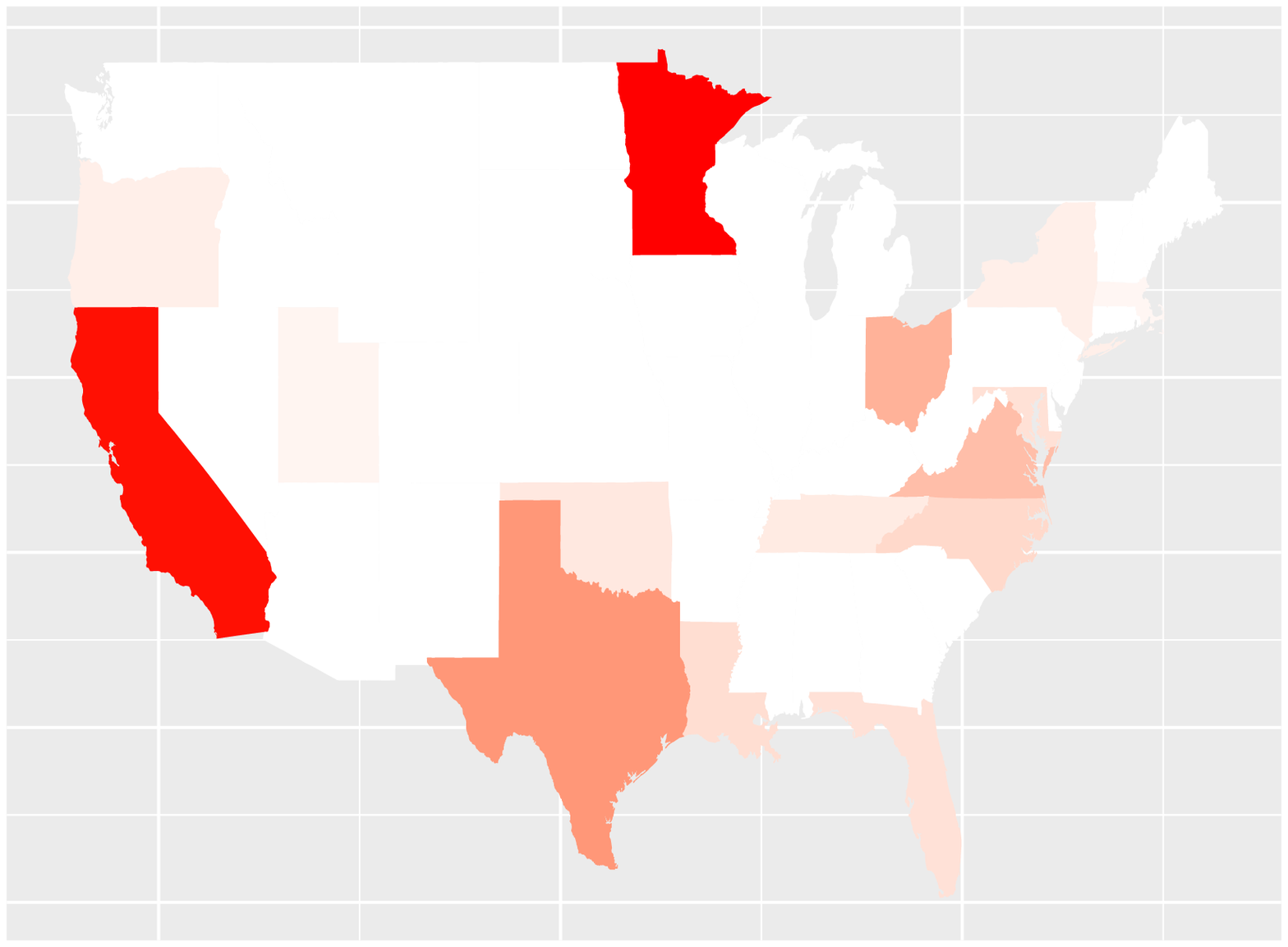} \\
        &
        (b.1) true &
        (b.2) PoSSTenD  &
        (b.3) NMC-scan-stat &
        (b.4) YPS-SSD \\
        %%%%%%%%%%%%%%%%%%
        disease 3&
        \includegraphics[width=0.19\textwidth]{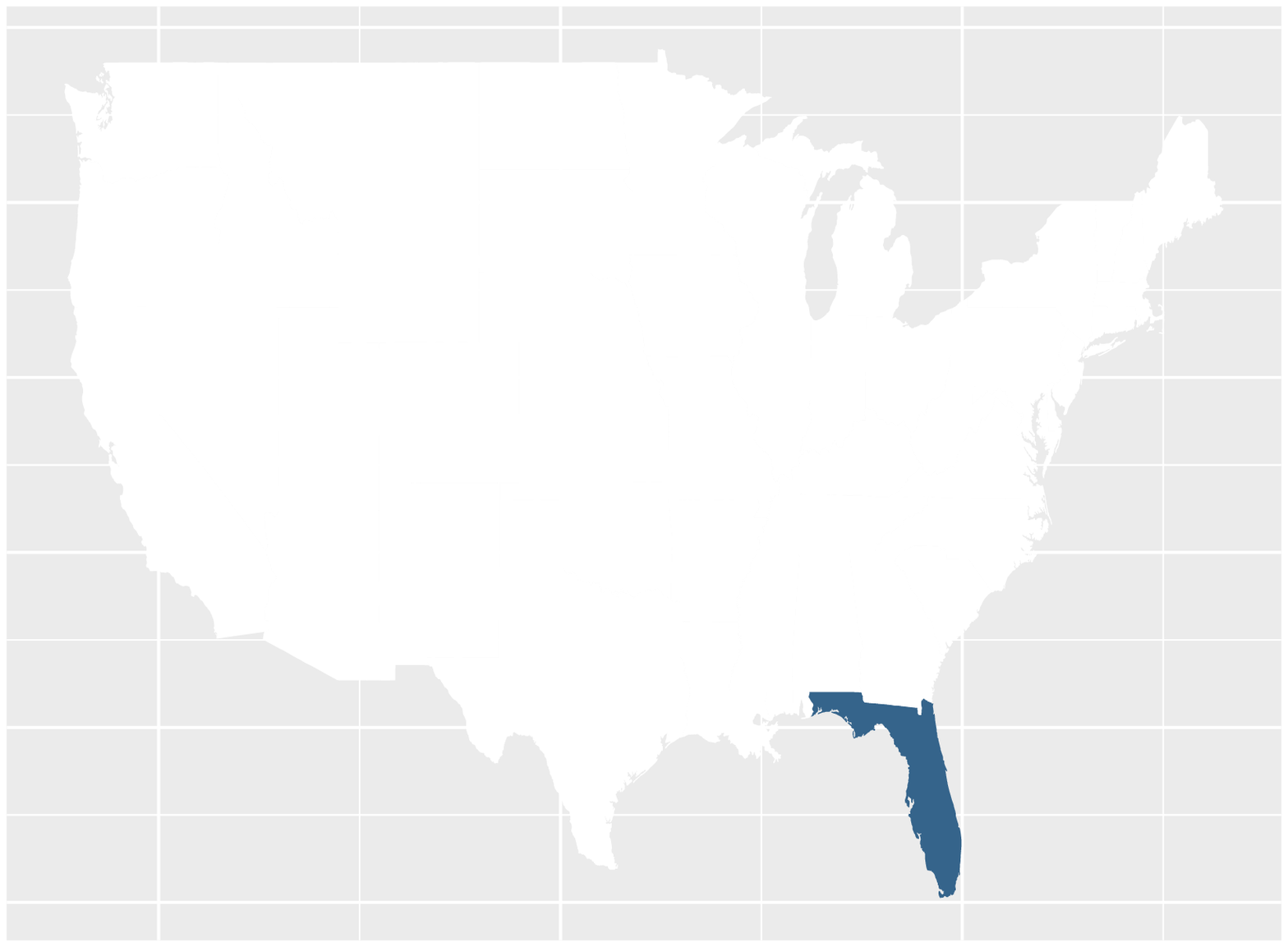} &
        \includegraphics[width=0.19\textwidth]{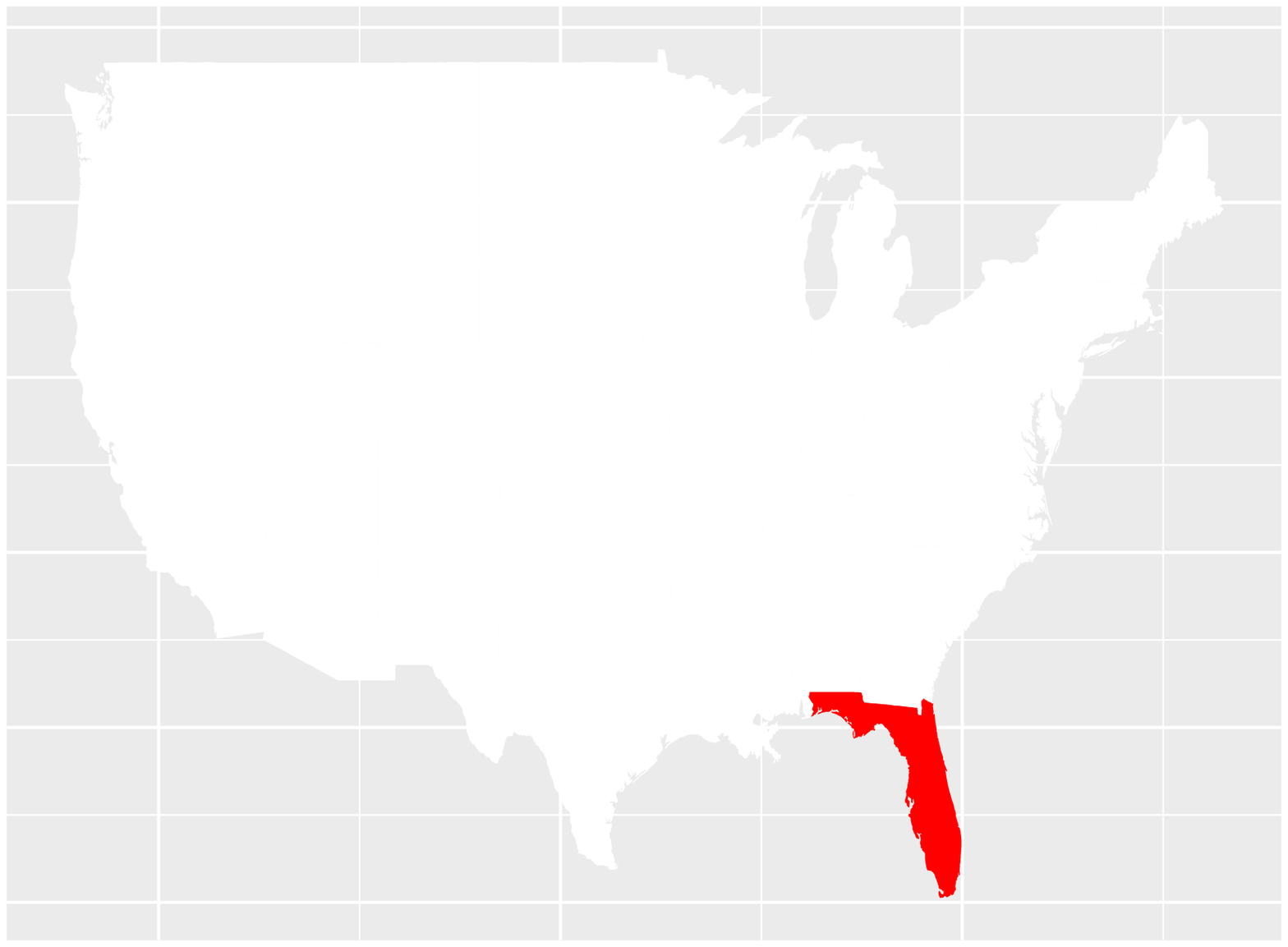}  &
        \includegraphics[width=0.19\textwidth]{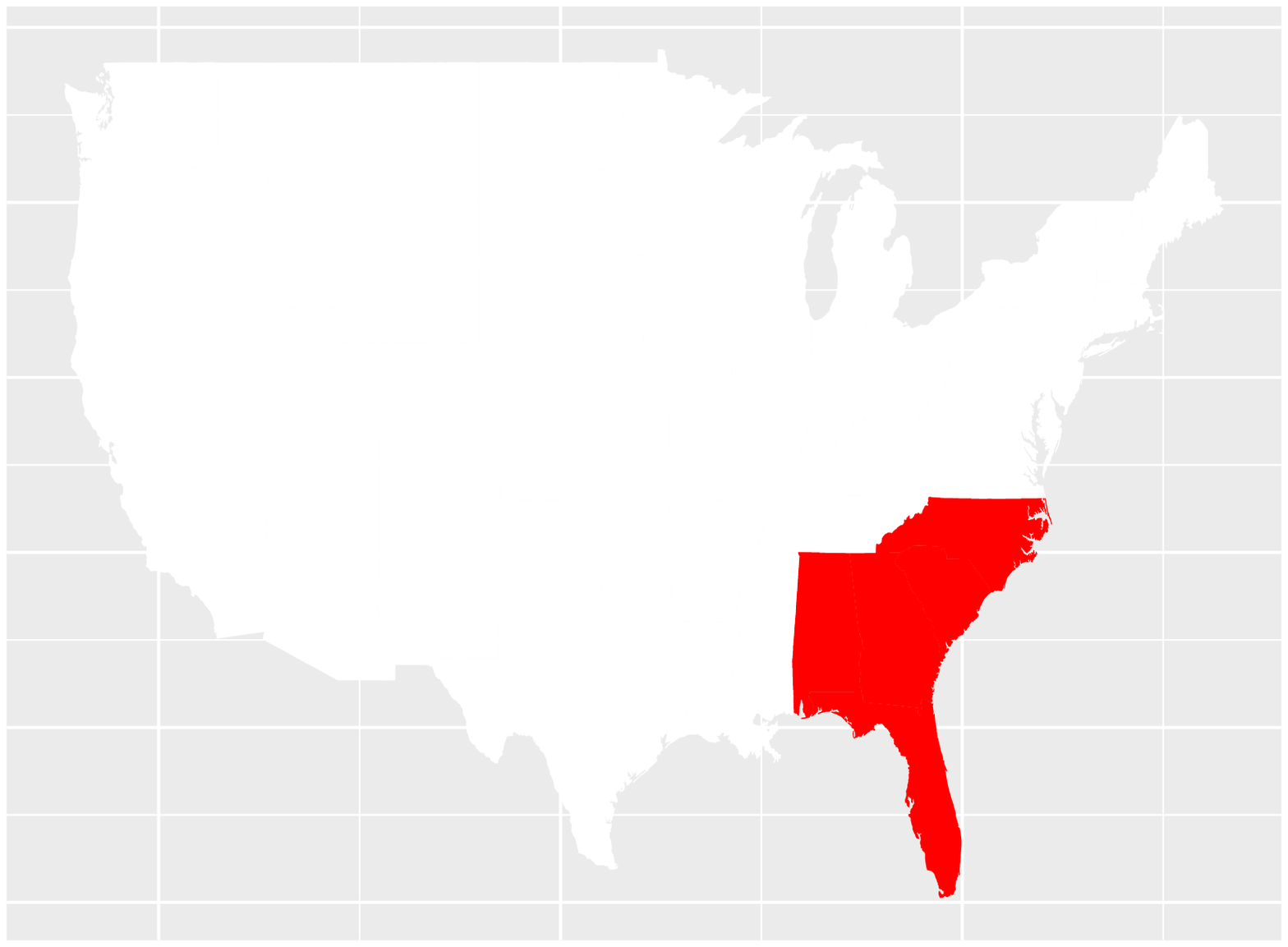} &
        \includegraphics[width=0.19\textwidth]{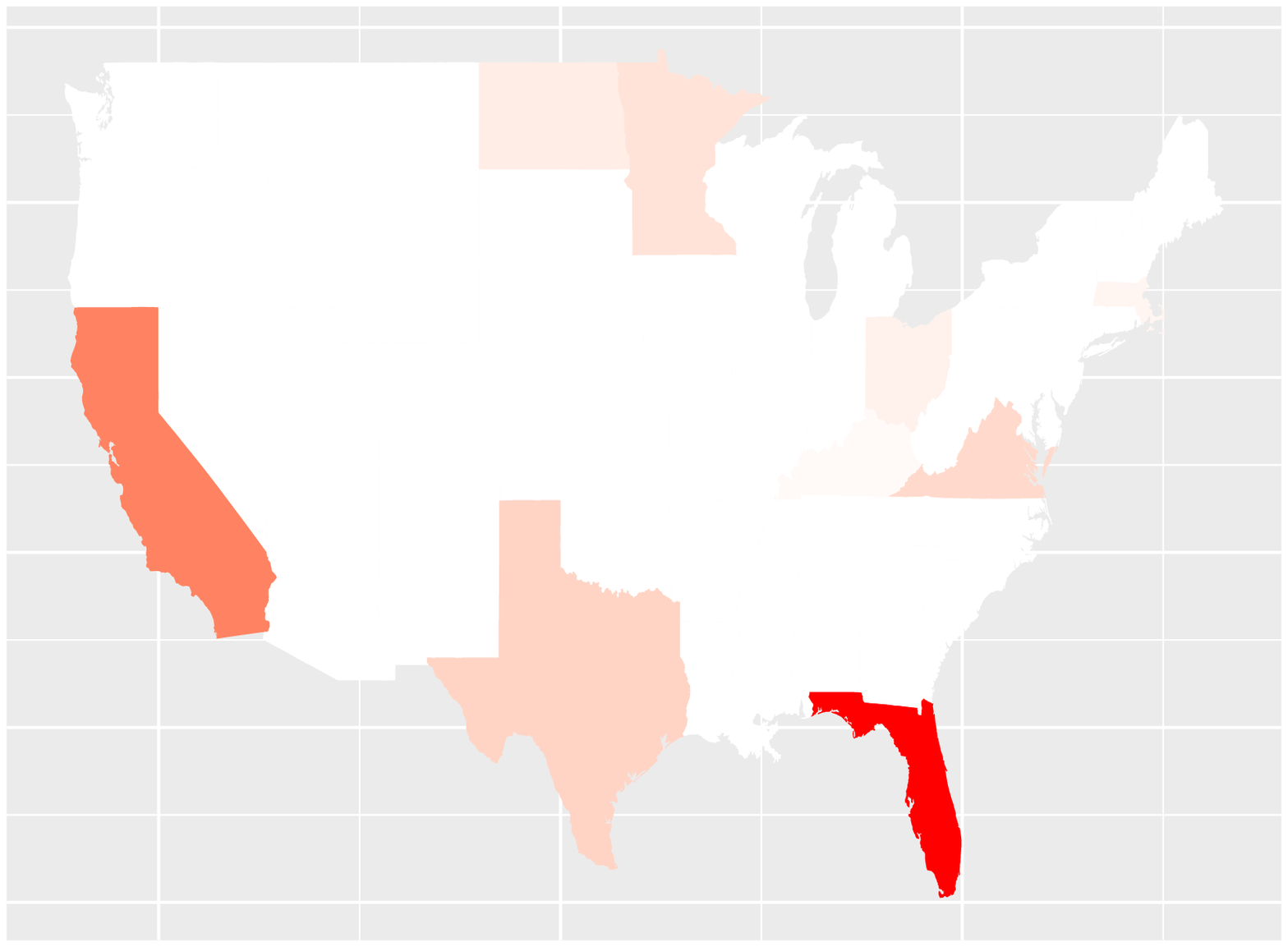} \\
        &
        (c.1) true &
        (c.2) PoSSTenD  &
        (c.3) NMC-scan-stat &
        (c.4) YPS-SSD \\
        %%%%%%%%%%%%%%%%%%
        disease 4&
        \includegraphics[width=0.19\textwidth]{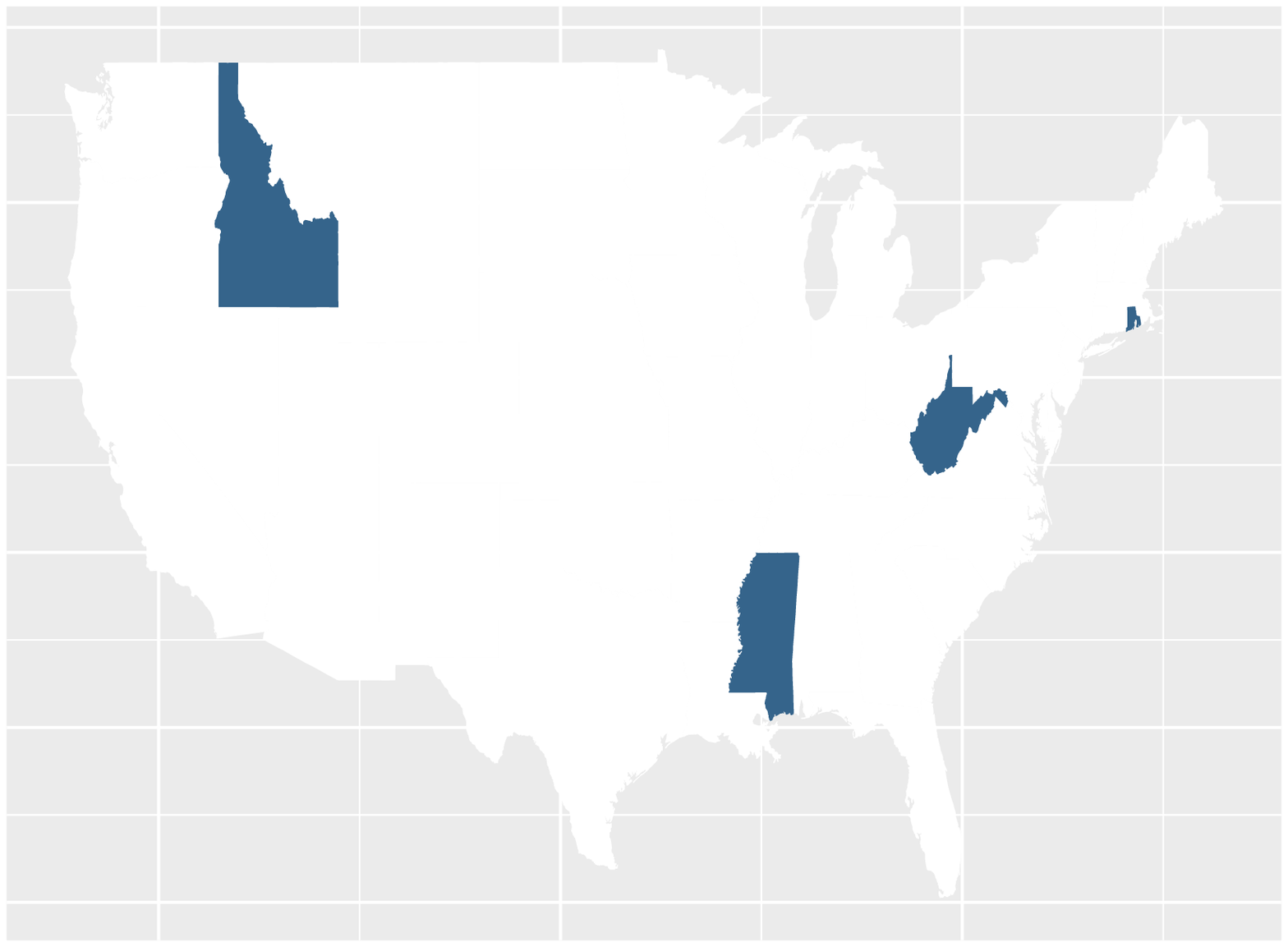} &
        \includegraphics[width=0.19\textwidth]{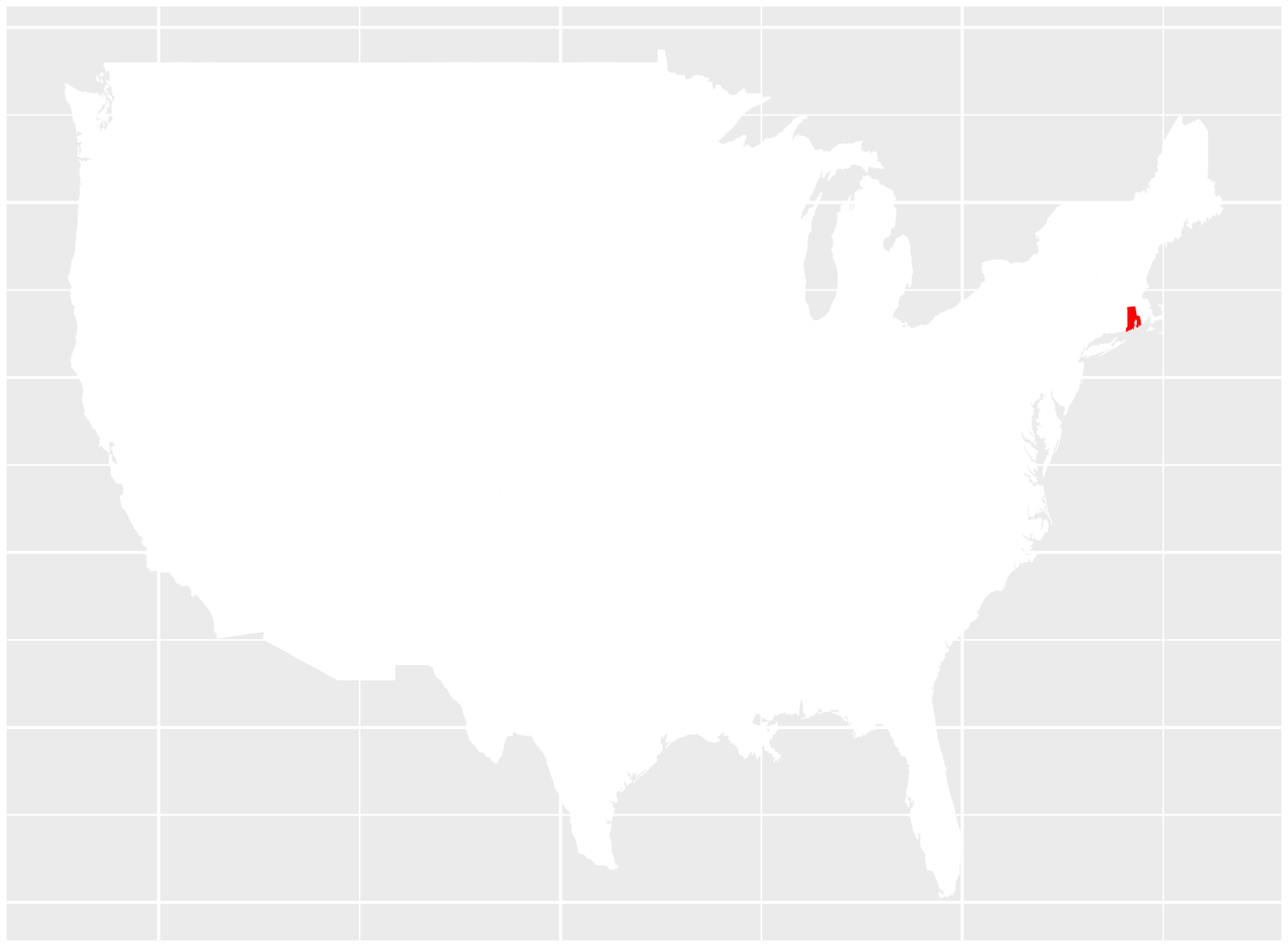}  &
        \includegraphics[width=0.19\textwidth]{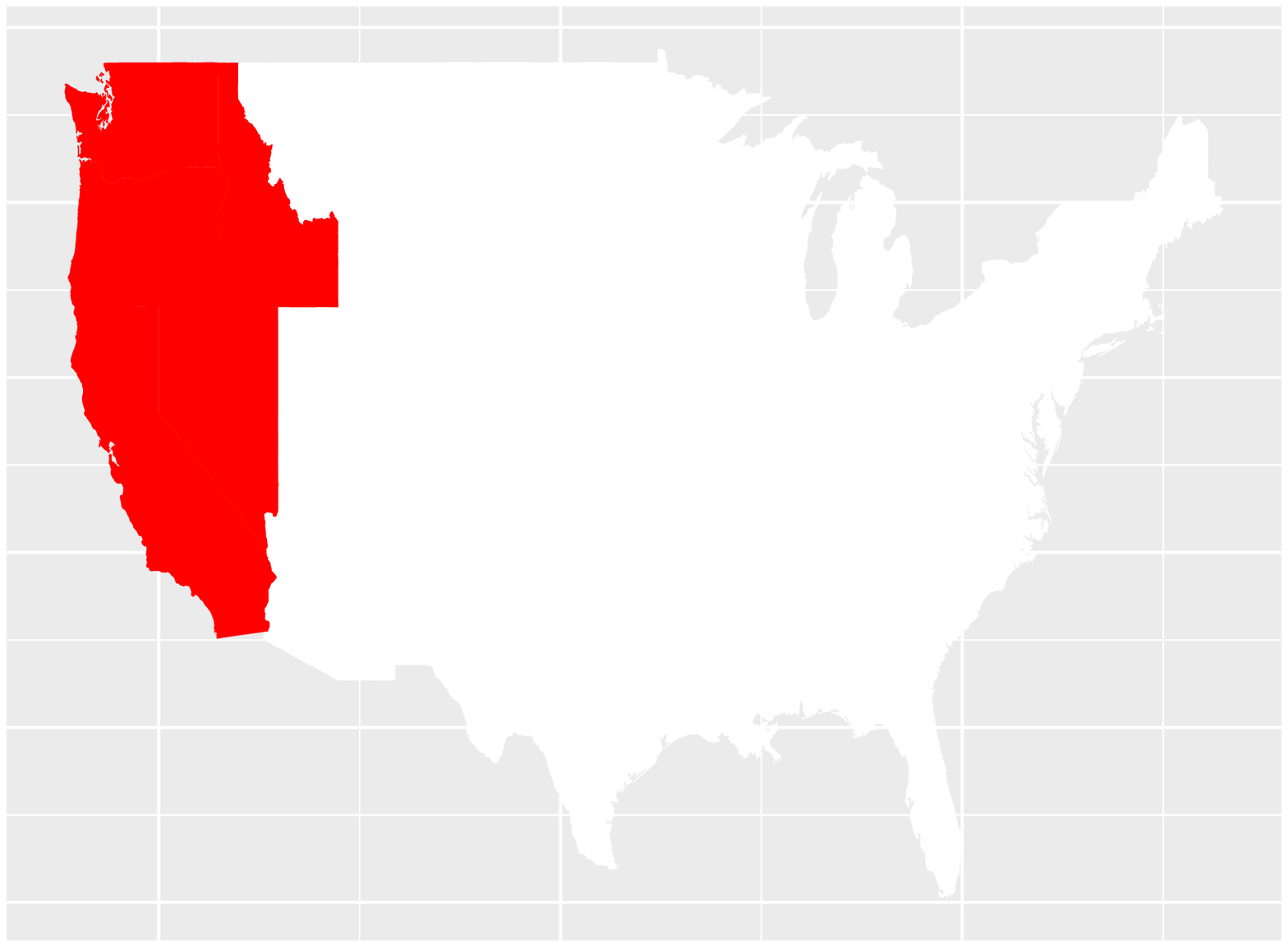} &
        \includegraphics[width=0.19\textwidth]{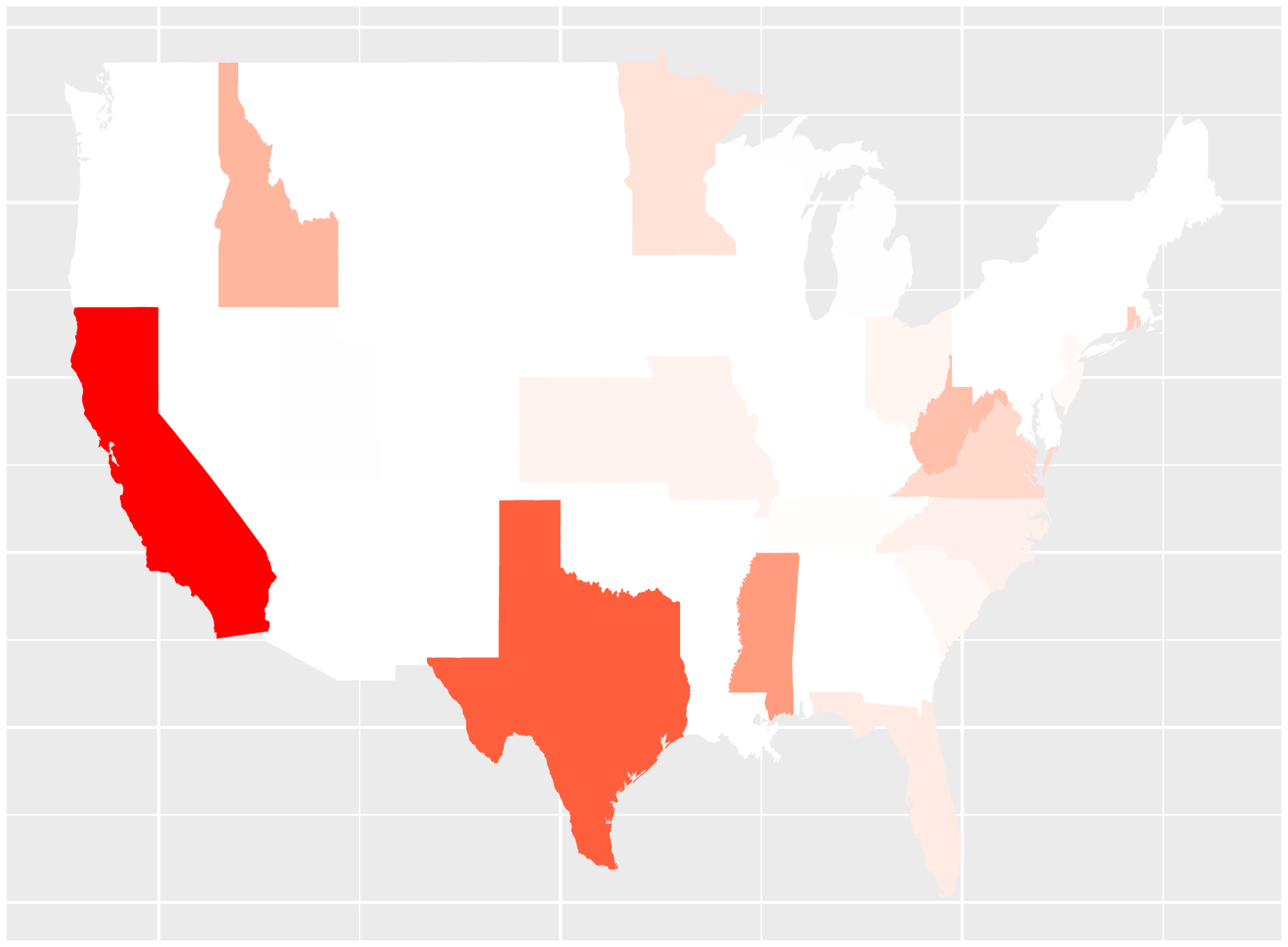} \\
        &
        (d.1) true &
        (d.2) PoSSTenD  &
        (d.3) NMC-scan-stat &
        (d.4) YPS-SSD \\
        %%%%%%%%%%%%%%%%%%
        disease 5&
        \includegraphics[width=0.19\textwidth]{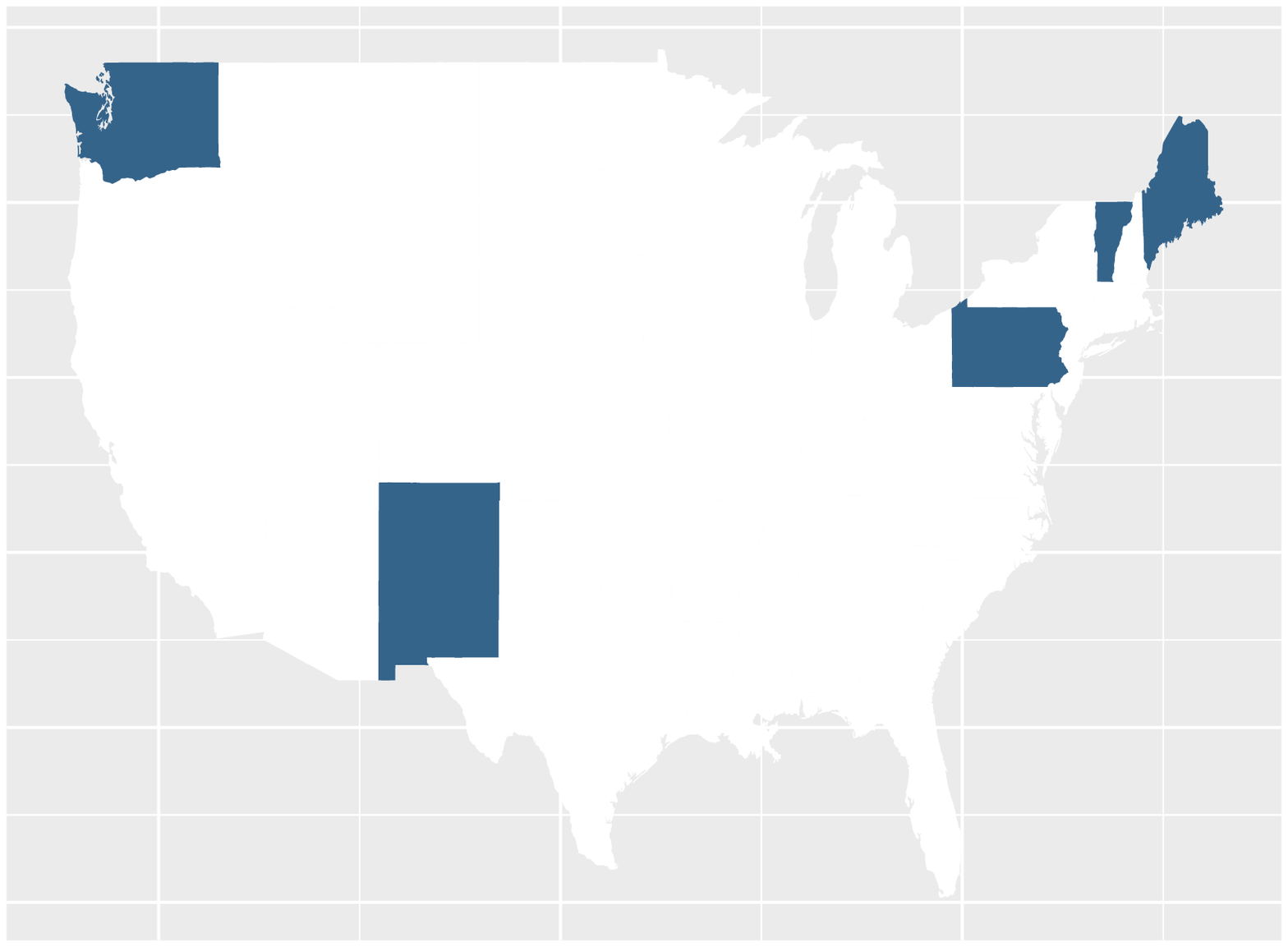} &
        \includegraphics[width=0.19\textwidth]{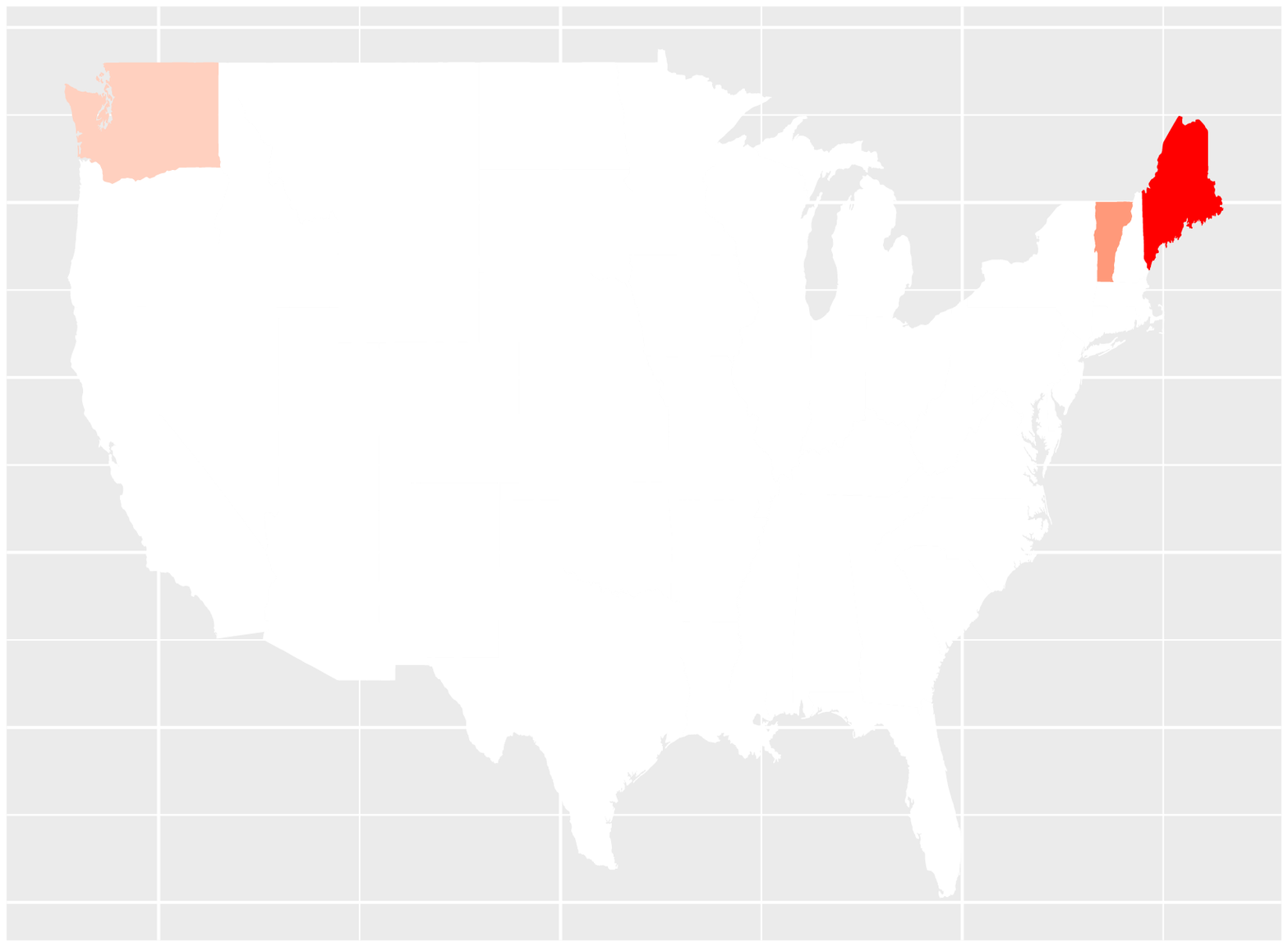}  &
        \includegraphics[width=0.19\textwidth]{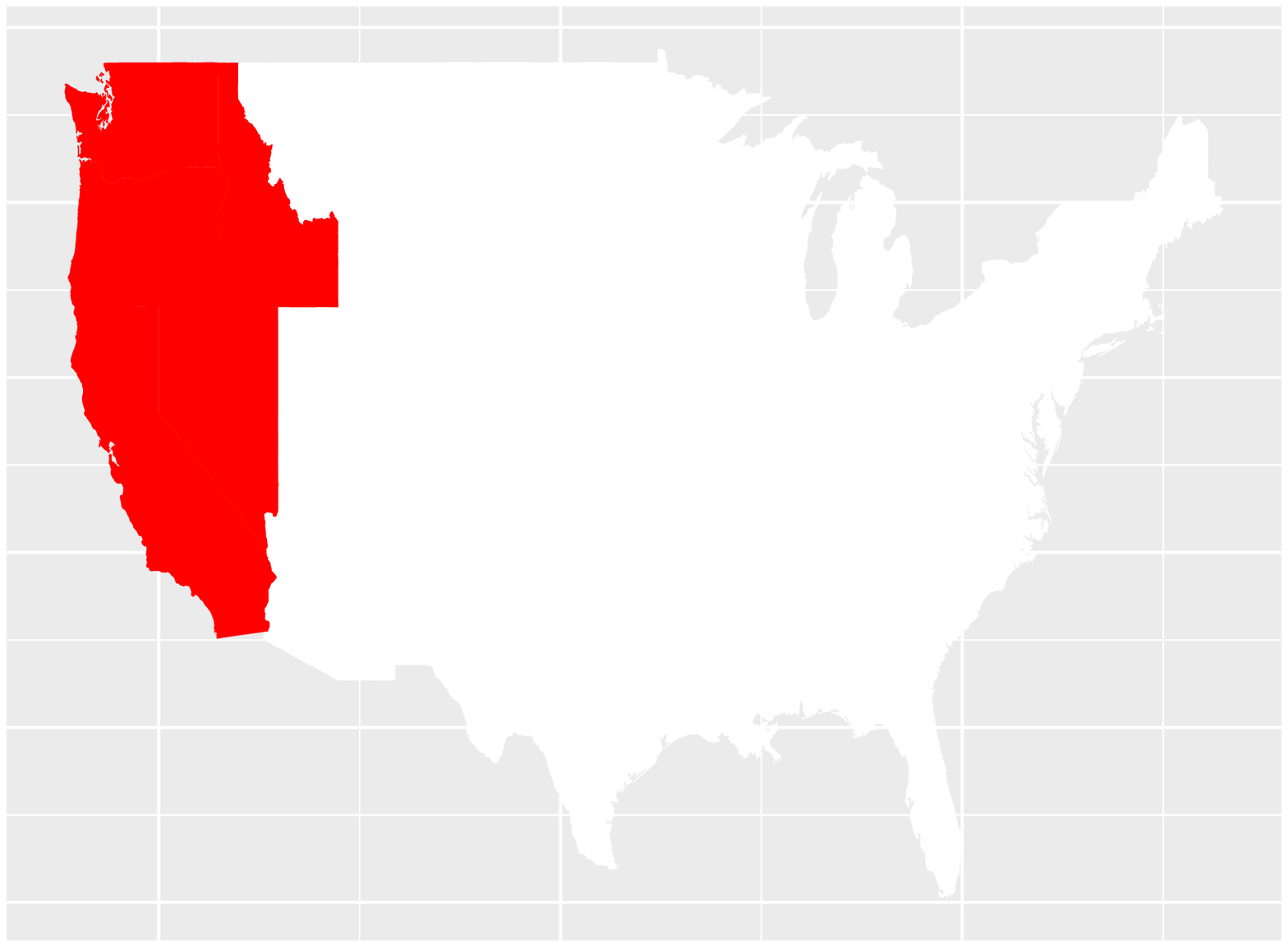}  &
        \includegraphics[width=0.19\textwidth]{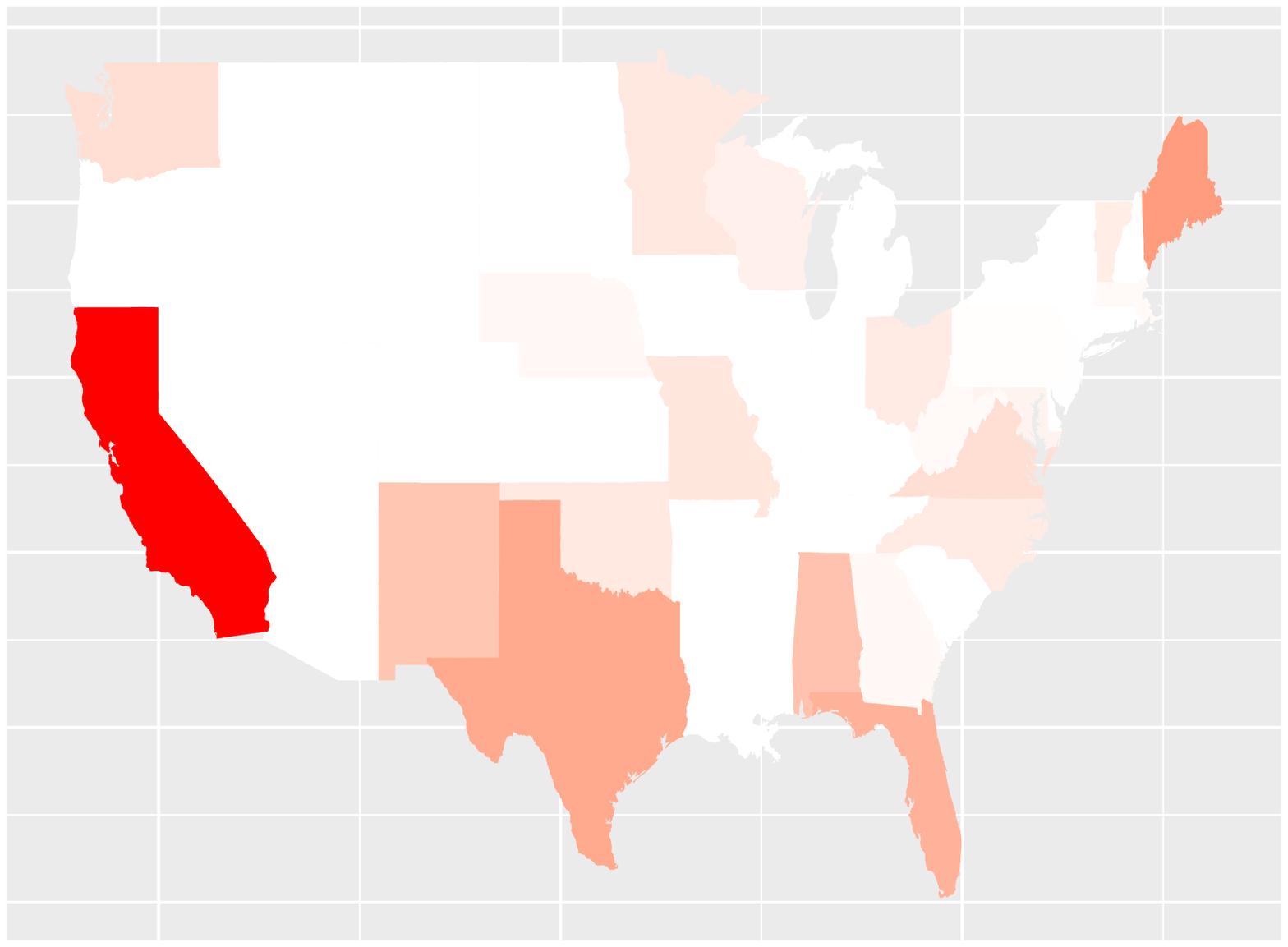} \\
        &
        (e.1) true &
        (e.2) PoSSTenD  &
        (e.3) NMC-scan-stat &
        (e.4) YPS-SSD \\
      \end{tabular}
      \caption{Comparison of true hot-spots and hot-spots detected by our PoSSTenD method and YPS-SSD method under the decreasing population size and large hot-spots $\delta = 0.2$
      \label{fig: simulation map}}
    \end{figure}

\subsubsection{Validation of Model Fitness}

    In this section, we demonstrate that our proposed PoSSTenD method leads to a reasonably well estimation for the global trend mean.
    The criterion we use is the Squared-Root of Mean Square Error (SMSE).
    Table \ref{table: sim -- fitness} shows the SMSE of the PoSSTenD, and YPS-SSD method.
    We do not report the SMSE of the other baseline methods (NMC-scan-stat, ZQ-Lasso, DBS-PCA, and $\text{T}^2$), becasue they cannot model the global trend mean.
    It is clear from Tables \ref{table: sim -- fitness} that, our proposed PoSSTenD method performs well in terms of the background fitness in both two scenarios.
    
    \begin{table}[htbp]
    \caption{Fitting error of our method and YPS-SSD (decreasing population size and positive hot-spots)
    \label{table: sim -- fitness}}
    \centering
    \begin{adjustbox}{max width=0.95\textwidth}
    \centering
    \begin{threeparttable}
      \begin{tabular}{c|cccccccc}
        \hline
        Fitting error &
        $\delta = 0.05$ &
        $\delta = 0.075$ &
        $\delta = 0.1$ &
        $\delta = 0.125$ &
        $\delta = 0.15$&
        $\delta = 0.175$ &
        $\delta = 0.2$    \\
        \hline
        &\multicolumn{7}{c}{population with increasing trend} \\
        \cline{2-8}
        PoSSTenD
            & $ 0.0259\times 10^5$
            & $0.0259 \times 10^5$
            & $0.0601\times 10^5 $
            & $0.0601\times 10^5 $
            & $0.0601\times 10^5 $
            & $ 0.0988\times 10^5 $
            & $ 0.0988 \times 10^5 $   \\
            & ($0.0150\times 10^3$)
            & ($0.0143 \times 10^3$)
            & ($0.0397\times 10^3 $)
            & ($0.0401\times 10^3 $)
            & ($0.0423\times 10^3 $)
            & ($0.0723\times 10^3$)
            & ($0.0702\times 10^3$)  \\
        YPS-SSD 
            & $ 1.3218 \times 10^5$
            & $1.3370\times 10^5$
            & $1.3520\times 10^5 $
            & $1.3698\times 10^5 $
            & $1.3899\times 10^5 $
            & $ 1.4098\times 10^5 $
            & $1.4282 \times 10^6 $   \\
            & ($1.1723 \times 10^3$)
            & ($1.6122 \times 10^3$)
            & ($2.0530\times 10^3 $)
            & ($2.6108\times 10^3 $)
            & ($3.0255\times 10^3 $)
            & ($3.7379\times 10^3$)
            & ($ 4.0804\times 10^3$)  \\
        \hline
        &\multicolumn{7}{c}{population with decreasing trend} \\
        \cline{2-8}
        PoSSTenD
            & $ 1.3202\times 10^4$
            & $1.3198 \times 10^4$
            & $1.3188 \times 10^4 $
            & $1.7187 \times 10^4 $
            & $1.7157\times 10^4 $
            & $ 2.1140 \times 10^4 $
            & $ 2.1091 \times 10^4 $   \\
            & ($65.6439$)
            & ($68.9465$)
            & ($65.4763 $)
            & ($0.0915 \times 10^3 $)
            & ($0.0934 \times 10^3 $)
            & ($0.1134 \times 10^3$)
            & ($ 0.1201\times 10^3$)  \\
        YPS-SSD & $ 4.7511\times 10^4$
            & $4.8117 \times 10^4$
            & $4.9009 \times 10^4 $
            & $5.0127\times 10^4 $
            & $5.1698\times 10^4 $
            & $ 5.3398 \times 10^4 $
            & $ 5.5376 \times 10^4 $   \\
            & ($745.1110$)
            & ($870.4066$)
            & ($ 963.7134 $)
            & ($1.1771 \times 10^3 $)
            & ($1.3612 \times 10^3 $)
            & ($1.5949 \times 10^3$)
            & ($1.8306 \times 10^3$)  \\
        \hline
      \end{tabular}
    \begin{tablenotes}
      \footnotesize
        \item[1] The above results are based on 1000 simulations
        \item[2] The generated population size is of order $10^{5}$.
    \end{tablenotes}
    \end{threeparttable}
    \end{adjustbox}
    \end{table}

%---------------------------------------------------%
%                case study                         %
%---------------------------------------------------%

\section{Case Study}
\label{sec: case study}

    In this section, we apply our proposed PoSSTenD method to the infectious disease dataset described in  Section \ref{sec: data description}.
    There are only two missing values, and we handle them by the mean imputation with a reasonable assumption that they are missing at random.
    If one encounters the missing-not-at-random case, one can use the methods in \cite{chen2019semipara,chen2018semiparainf,chen2021instrument, chen2018functional}. 
    After the missing data is addressed, we use PoSSTenD for hot-spots detection and localization and compare its performance with other benchmarks as in Section \ref{sec: compared benchmarks}.
    
    First, we compare the performance of the detection delay.
    For all the methods, we set the control limits so that the average run length to false alarm constraint $\text{ARL}_0 = 50$ via Monte Carlo simulation under the assumption that data from the first $15$ years are in control.
    For the setting of the parameters and the selection of basis, they are the same as that in Section \ref{sec: simulation}.
    For our proposed PoSSTenD method, we build a CUSUM control chart utilizing the test statistic in Section \ref{sec: hot-spots detection}, which is shown in Figure \ref{fig: case study control chart}.
    From this figure, we can see that the hot-spots are detected in 2017 by our proposed PoSSTenD method.
    For the benchmark methods for comparison, we also use
    YPS-SSD \cite[see][]{SSD},
    ZQ-Lasso \cite[see][]{zou2009multivariate},
    DBS-PCA \cite[see][]{PCA} and
    $\text{T}^2$ \cite[see][]{T2}
    to our motivating dataset and summarize the performance of the detection of a hot-spots in Table \ref{table: case study temporal detection}.
    Note that the value in Table \ref{table: case study temporal detection} is the first year that raises alarm, i.e, $\min_{t=1986,\ldots,2014}\{t: W_{t}^{+}> L\}$.
    Our proposed PoSSTenD method, YPS-SSD, and ZQ-Lasso all raise alarms of hot-spots in the year 2017, while other benchmarks fail to detect any hot-spots (we do not represent the hot-spots year of NMC-scan-stat, as it does not report the hot-spots).
    While nobody knows the ground truth when hot-spots occur in this real-world dataset, our numerical simulation experiences suggest that year 2017 is likely a hot-spot.

    \begin{table}[htbp]
    \centering
    \begin{adjustbox}{max width=0.95\textwidth}
    \begin{tabular}{c|cccccc}
        \hline
        methods & PoSSTenD & NMC-scan-stat & YPS-SSD & ZQ-Lasso & DBS-PCA & $\text{T}^2$ \\
        \cline{2-5}
        \hline
        Year when an alarm is raised & 2017 & - & 2017 & 2017 & None & None\\
        \hline
    \end{tabular}
    \end{adjustbox}
    \caption{Detection of change-point year in infectious rate dataset.
             The label ``Year when an alarm is raised'' is first year that raises alarm, i.e, $\min_{t=2008,\ldots,2018}\{t: W_{t}^{+}> L\}$,  where $W_{t}^{+}$ is the CUSUM statistics defined in equation \eqref{equ: Wt}, and $L$ is control limits to achieve the average run length to false alarm constraint $\text{ARL}_0= 50$ via Monte Carlo simulation under the assumption that data from the first 15 years are in control.
    \label{table: case study temporal detection}
    }
    \end{table}

    \begin{figure}
      \centering
      \includegraphics[width=0.4 \textwidth]{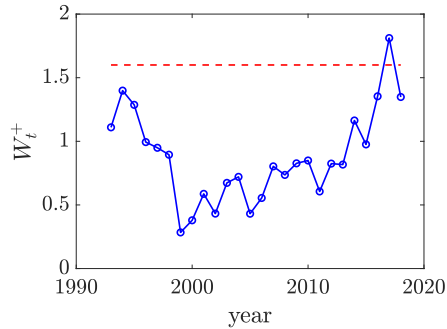}
      \caption{Control chart of our proposed method}
      \label{fig: case study control chart}
    \end{figure}

    Next, after the temporal detection of hot-spots, we need to further localize the hot-spots in the sense that we need to find out which state and which type of disease may lead to the occurrence of the temporal hot-spot in 2017.
    Because for the baseline methods, DBS-PCA and $\text{T}^2$  can only realize the detection of temporal changes, and ZQ-Lasso declares all states as hot-spots, we only show the localization of spatial hot-spots by our proposed PoSSTenD method, NMC-scan-stat and YPS-SSD method in Figure \ref{fig: case study hot-spots map}.
    In Figure \ref{fig: case study hot-spots map}, one row is one type of disease and we select three types of diseases as representatives, i.e., mumps, syphilis, and pertussis.
    The first column is the raw data of the number of infected people in 2017, the second, third, and fourth columns are the hot-spots localized by our proposed PoSSTenD method, NMC-scan-stat, and YPS-SSD method, respectively.
    It can be seen from Figure \ref{fig: case study hot-spots map} that, our proposed PoSSTenD method realizes more sparse hot-spots localization.
    Given the simulation results in Section \ref{sec: simulation}, our proposed PoSSTenD method has very high precision, recall, and F measure, we can declare that these states are highly likely to be hot-spots.
    For NMC-scan-stat, it trends to localize clustered hot-spots, which has a very low recall.
    For YPS-SSD, since it has relative lower precision and recall than our method, there are more false-alarm than our method.
    This is reasonable for us to conclude that, our method has better performance in hot-spots localization than the selected benchmarks.

    \begin{figure}[t]
      \centering
      \begin{tabular}{ccccc}
      Mumps&
      \includegraphics[width=0.18\textwidth]{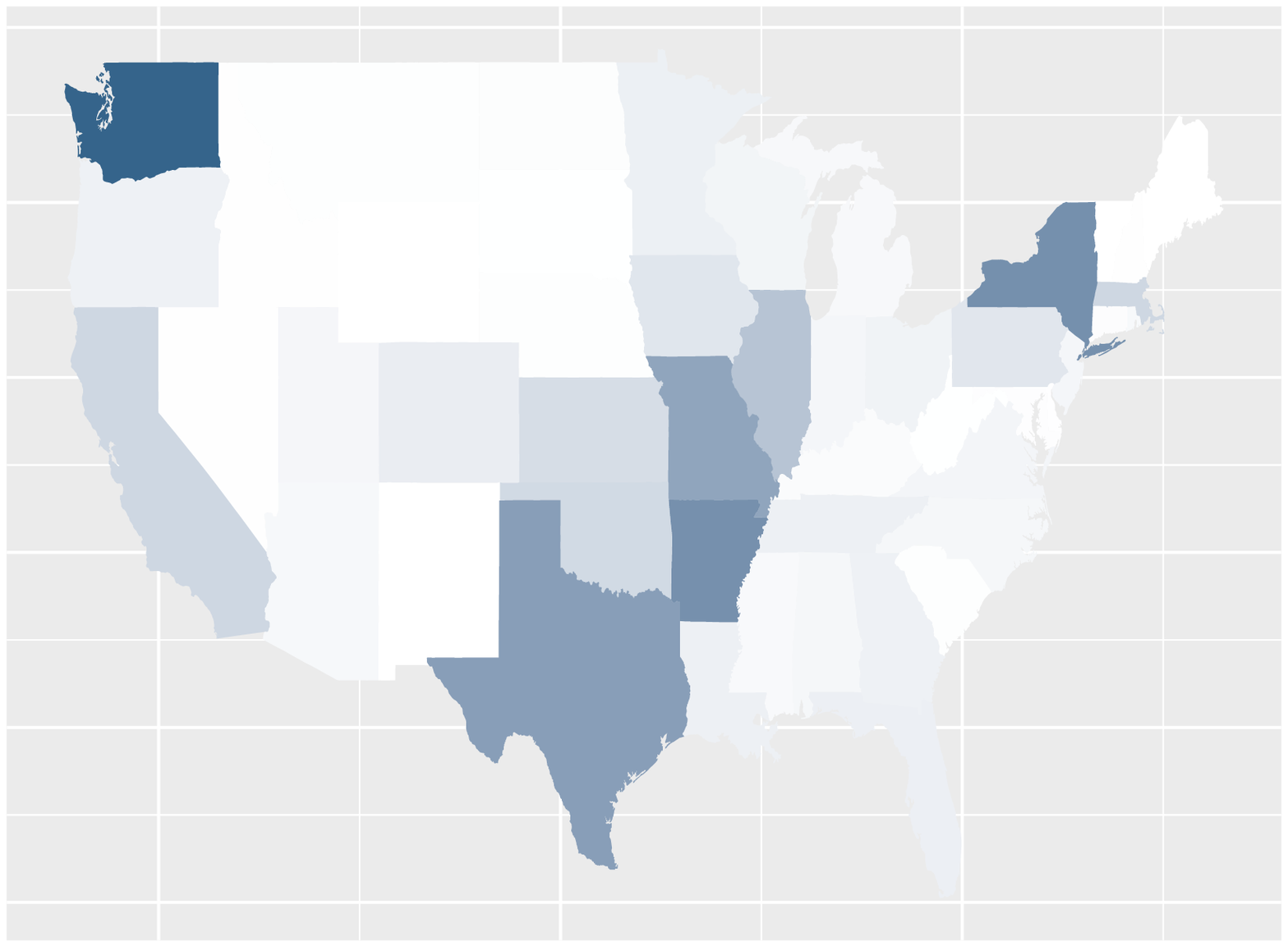} &
      \includegraphics[width=0.18\textwidth]{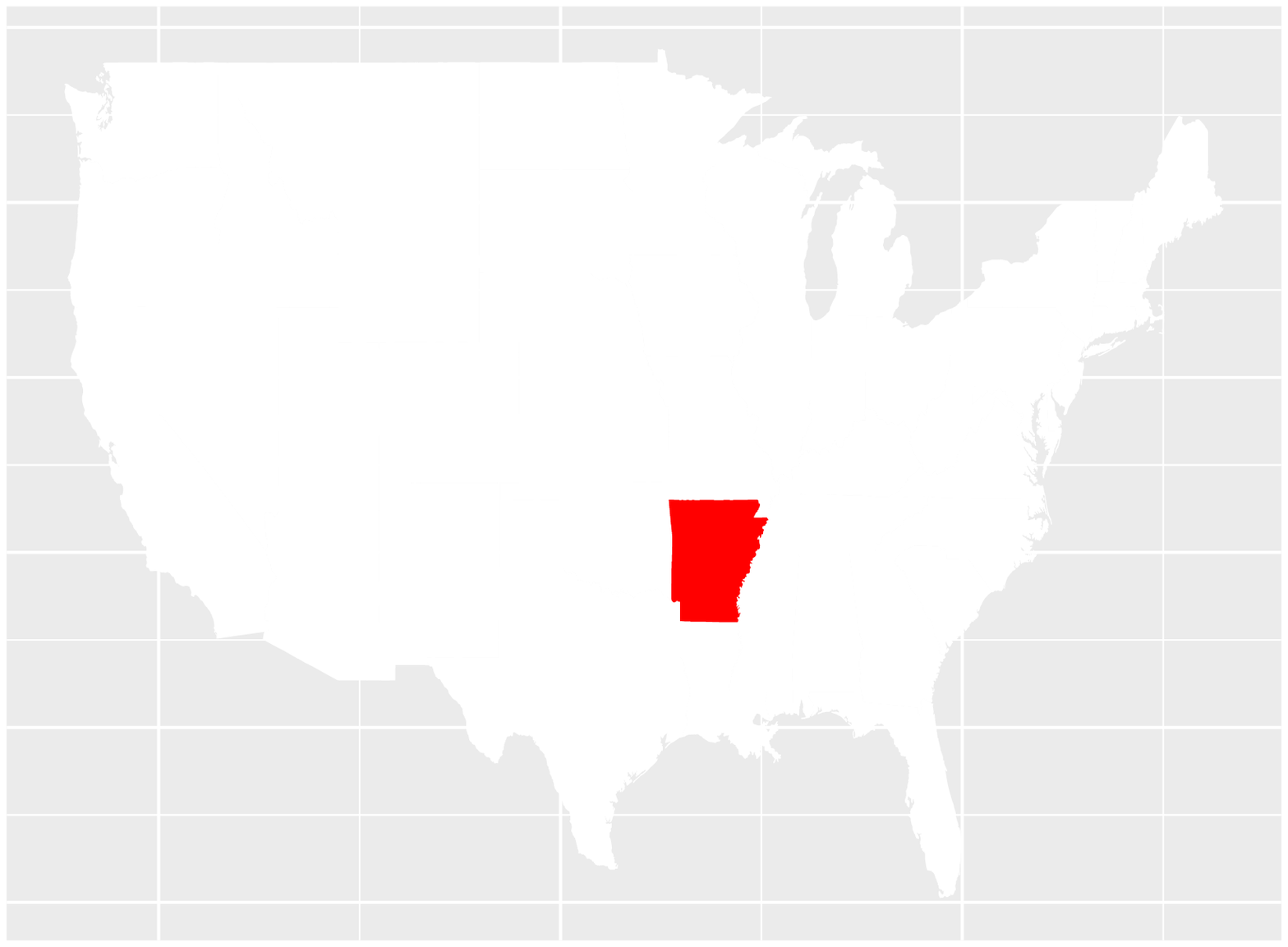} &
      \includegraphics[width=0.18\textwidth]{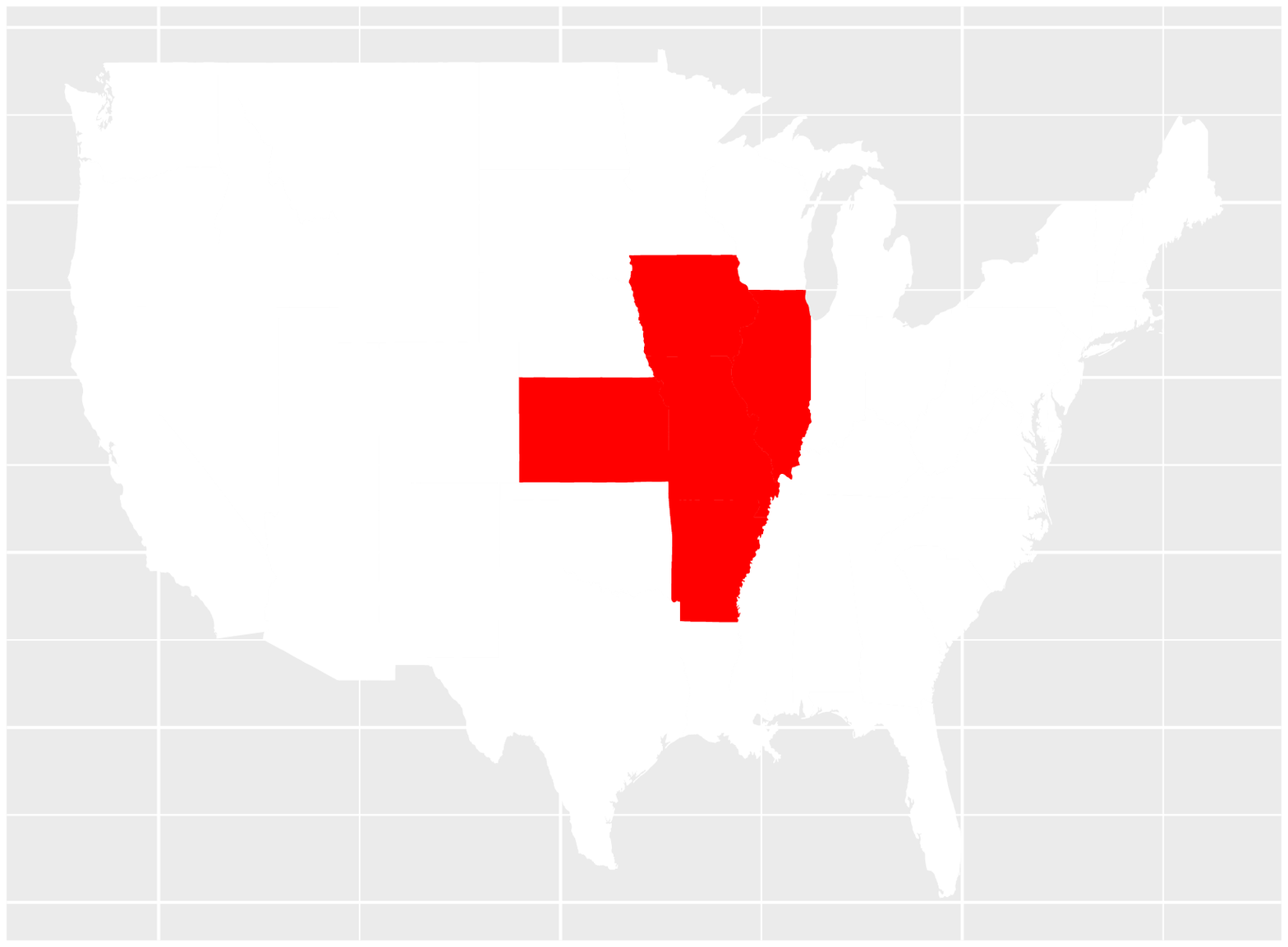} &
      \includegraphics[width=0.18\textwidth]{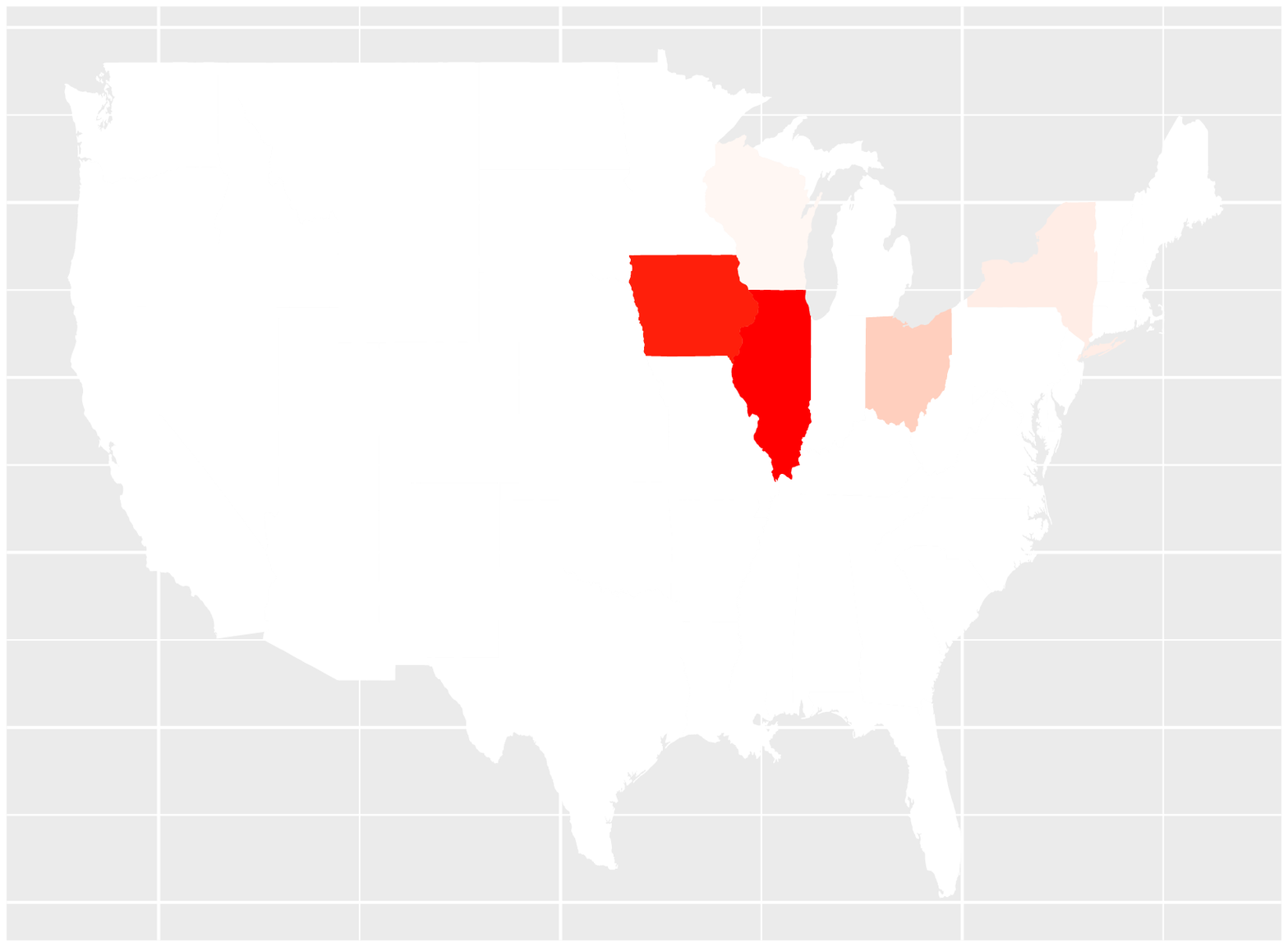} \\
      &(a.1)  2017 raw data & (a.2) PoSSTend & (a.3) NMC-scan-stat & (a.4) YPS-SSD \\
      Syphilis&
      \includegraphics[width=0.18\textwidth]{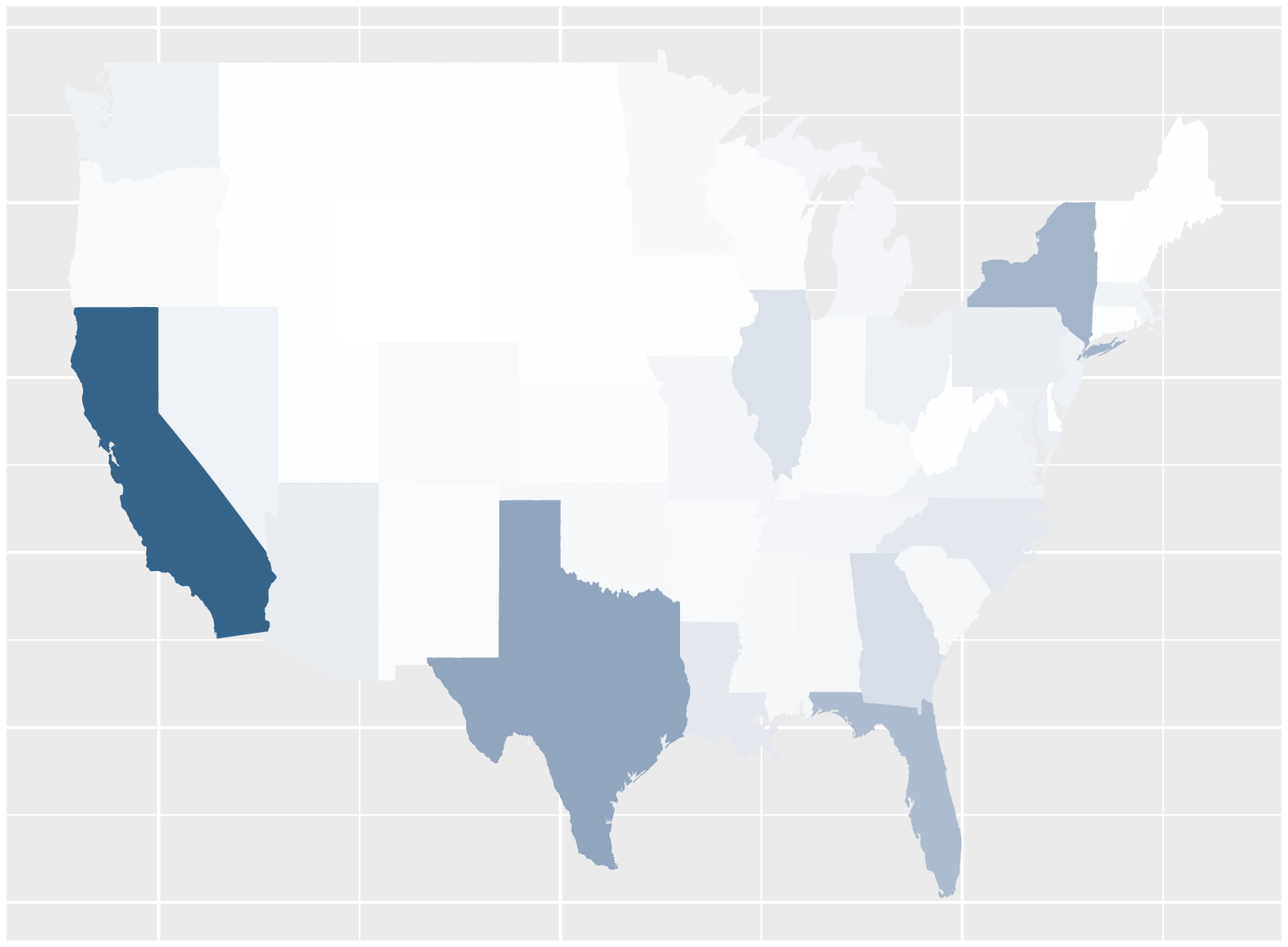} &
      \includegraphics[width=0.18\textwidth]{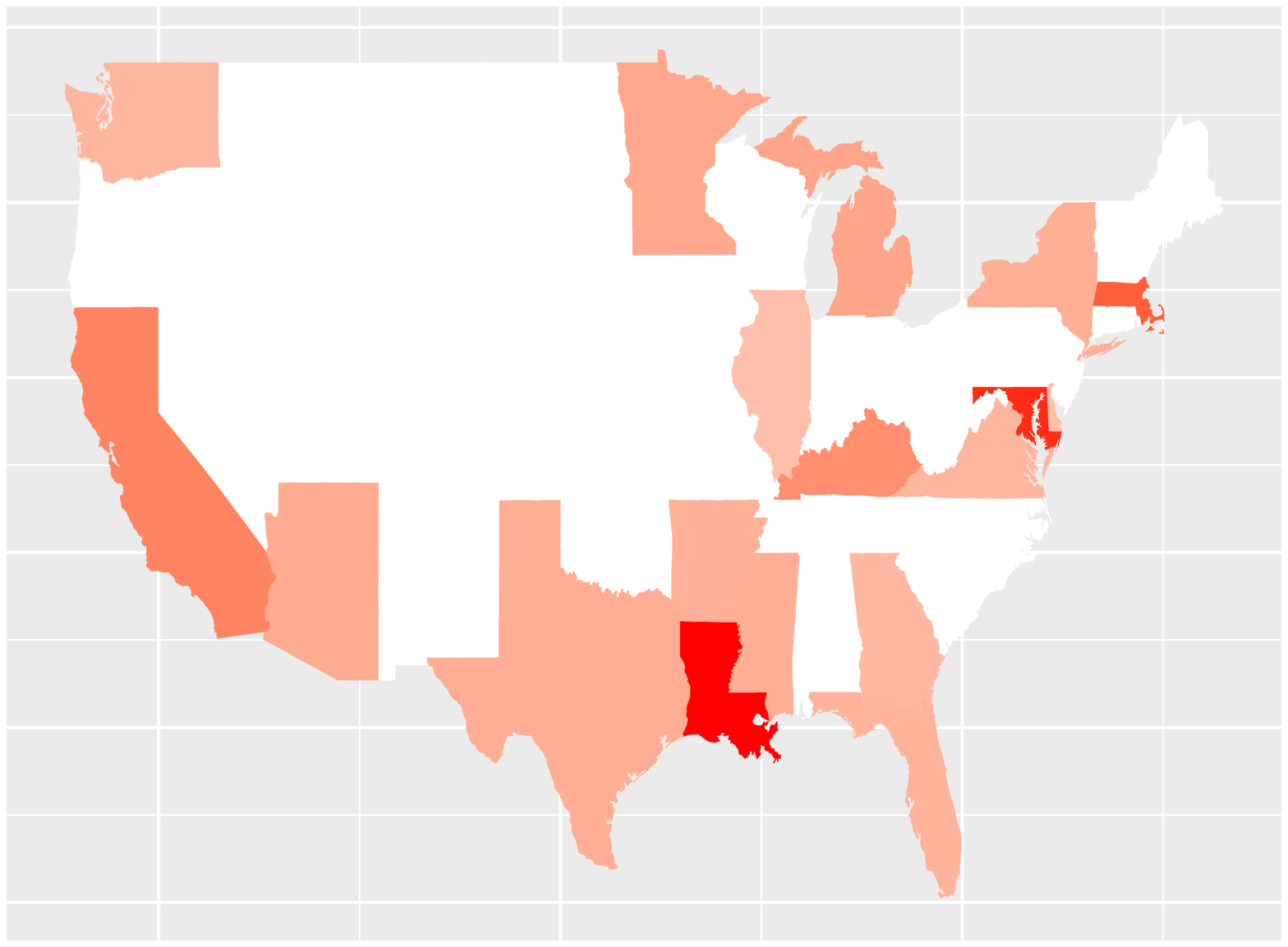} &
      \includegraphics[width=0.18\textwidth]{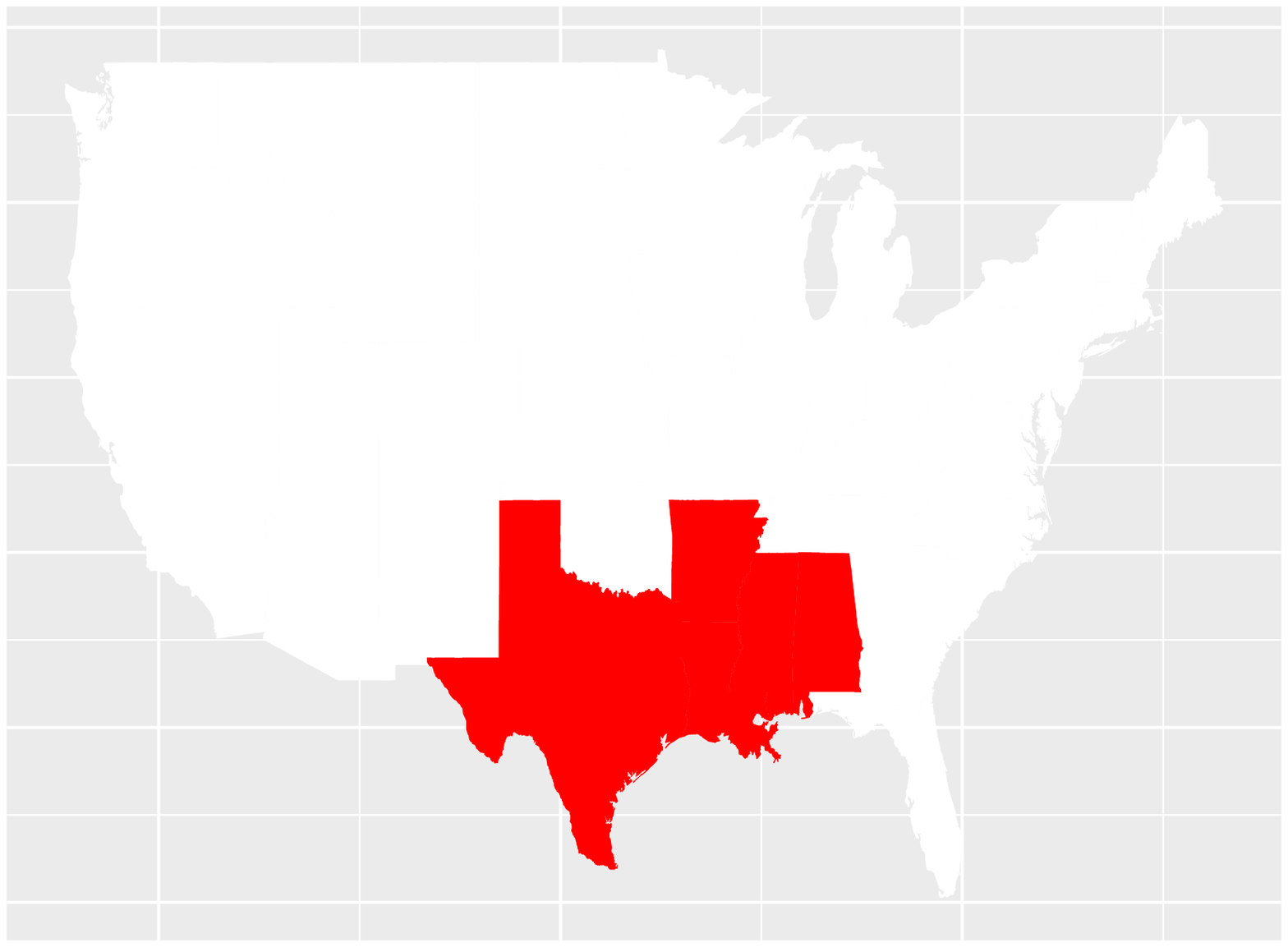} &
      \includegraphics[width=0.18\textwidth]{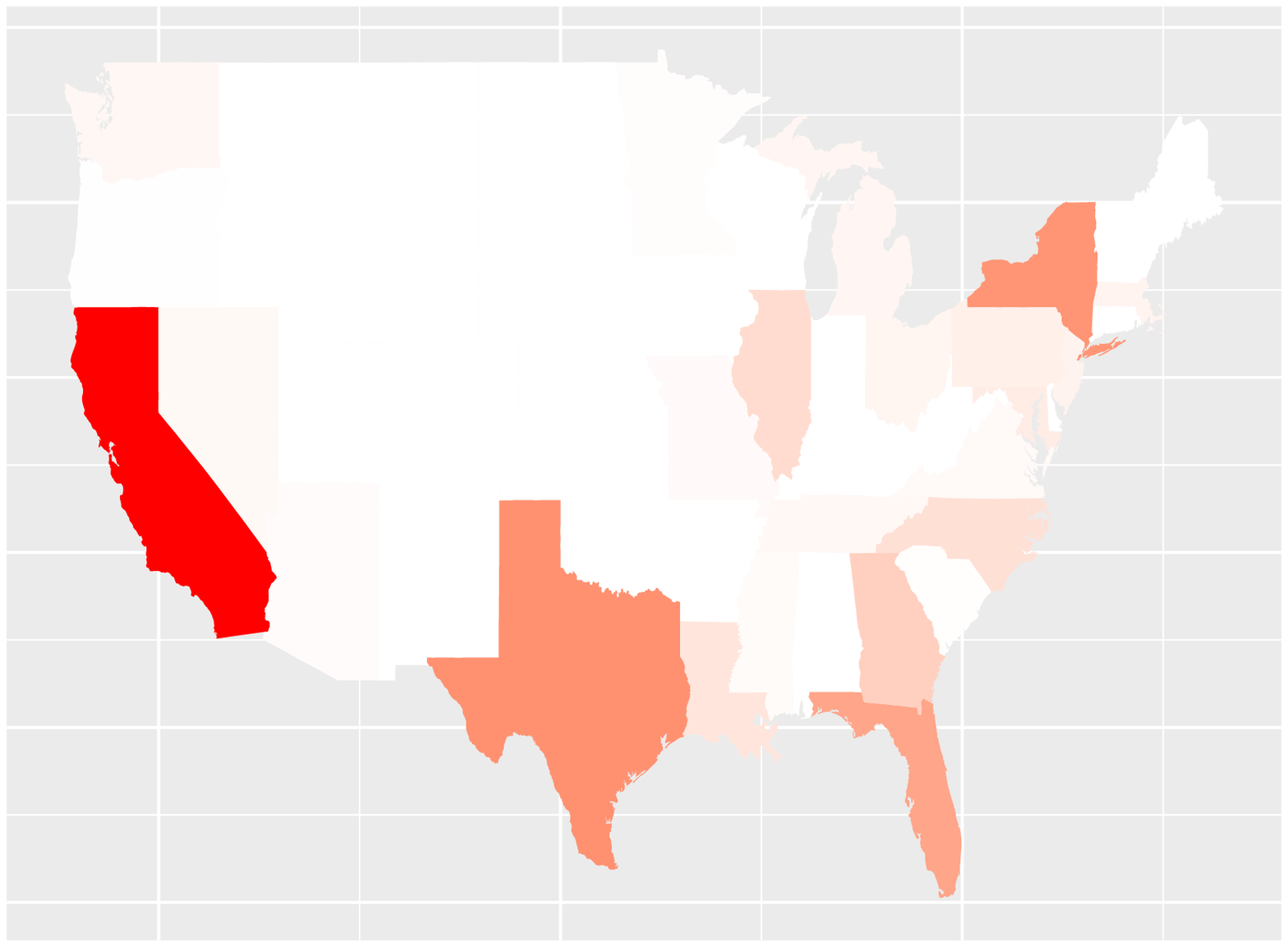} \\
      &(b.1)  2017 raw data & (b.2) PoSSTend & (b.3) NMC-scan-stat & (b.4) YPS-SSD \\
      Pertussis&
      \includegraphics[width=0.18\textwidth]{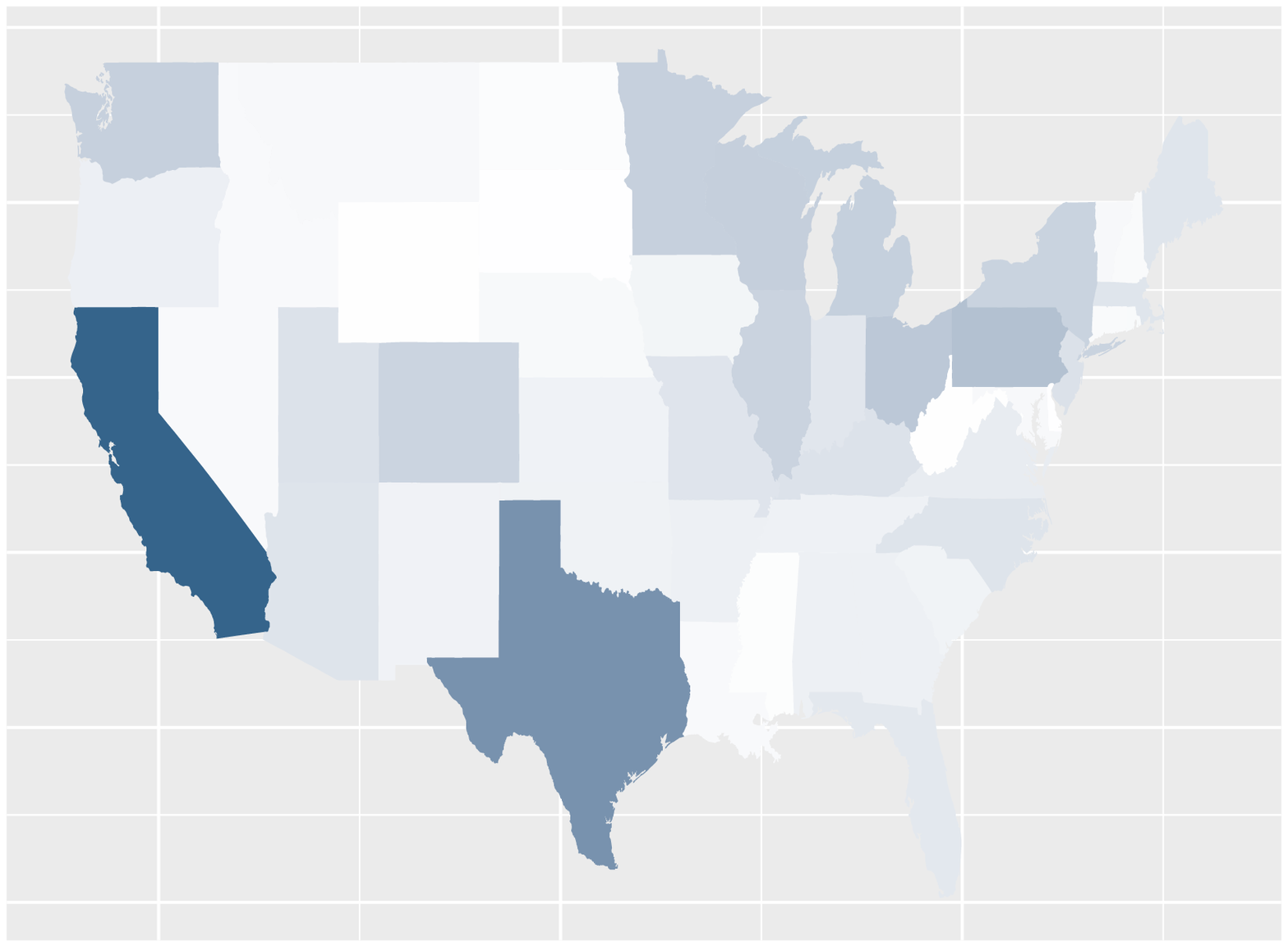} &
      \includegraphics[width=0.18\textwidth]{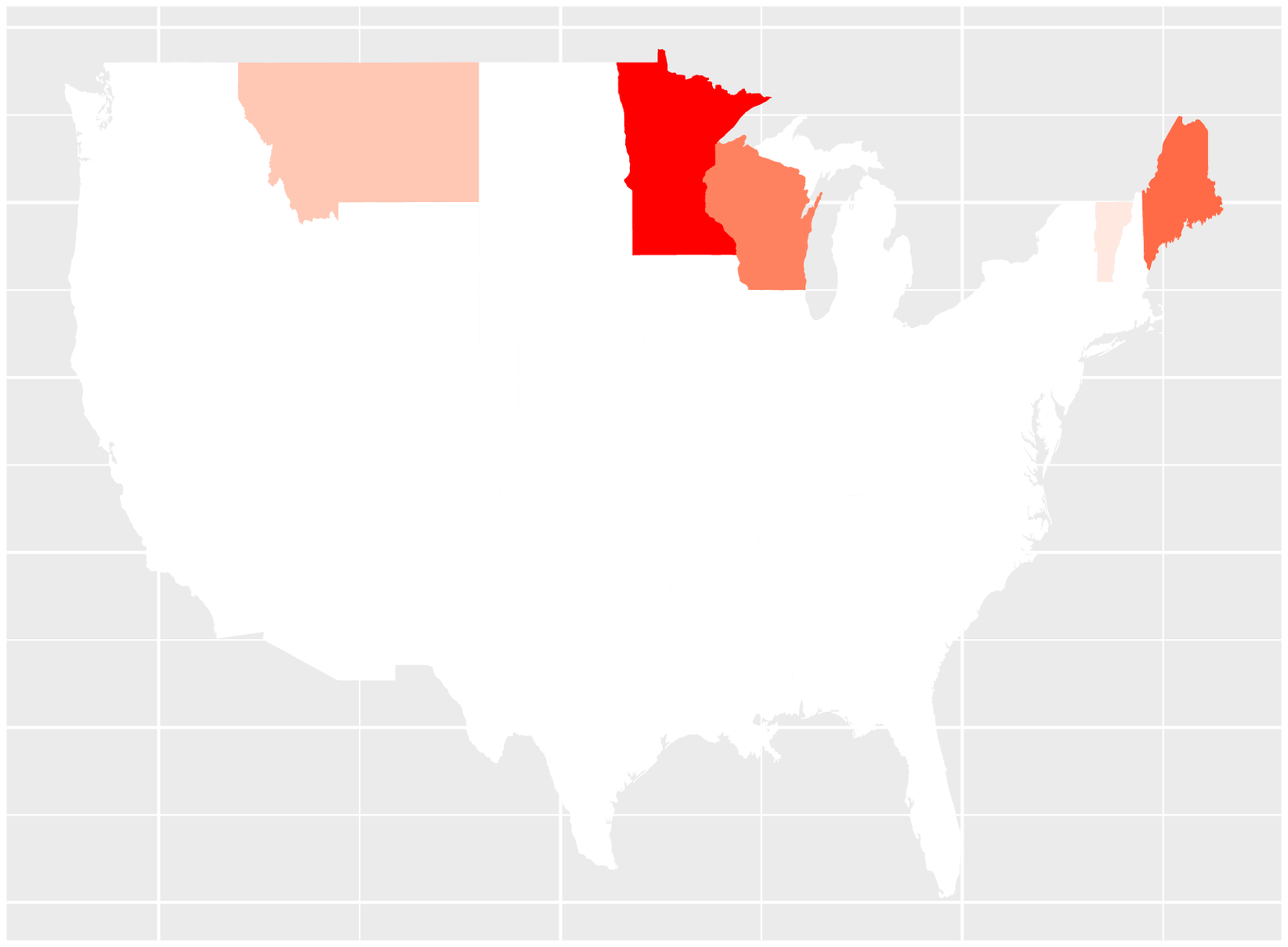} &
      \includegraphics[width=0.18\textwidth]{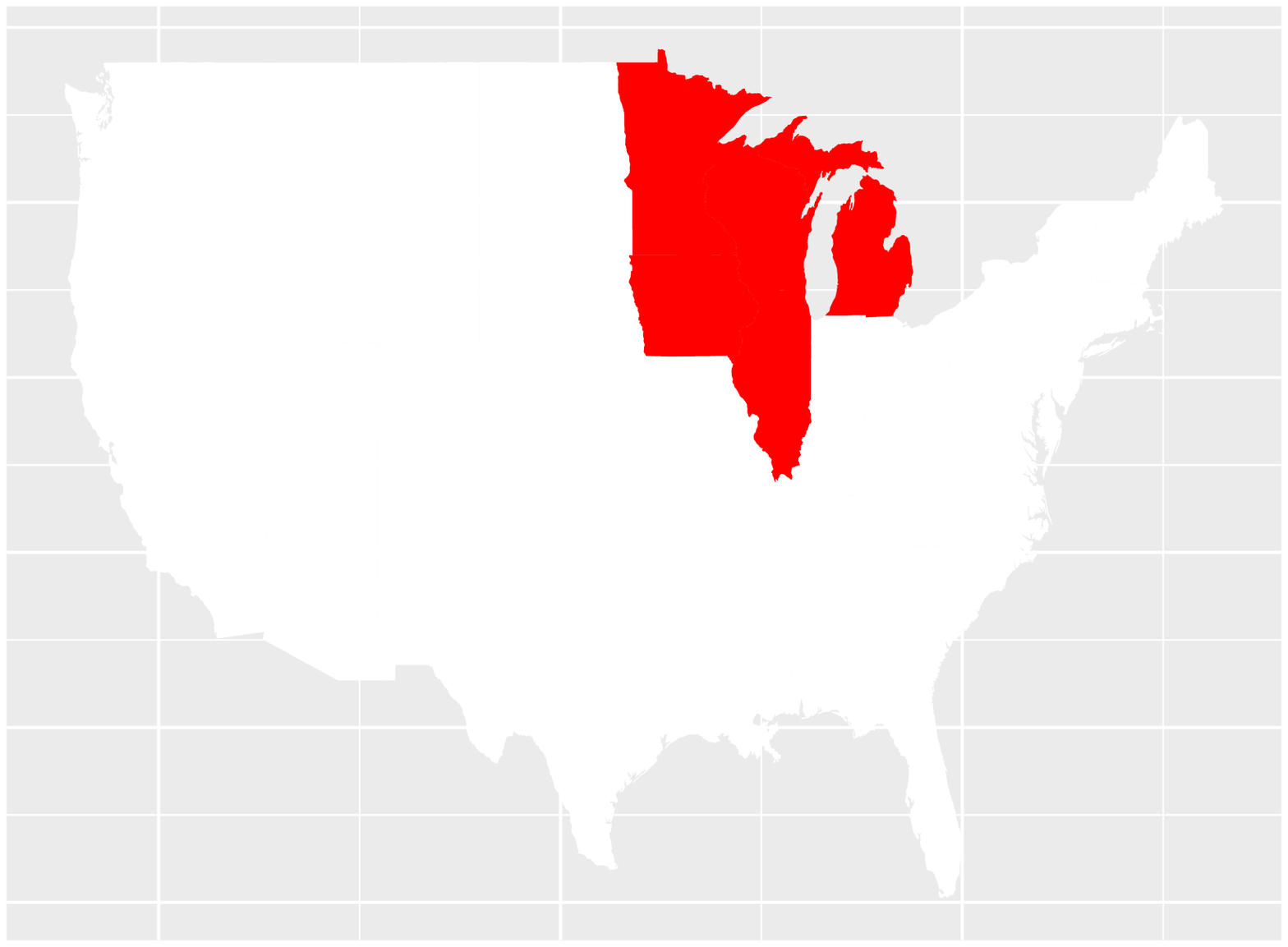} &
      \includegraphics[width=0.18\textwidth]{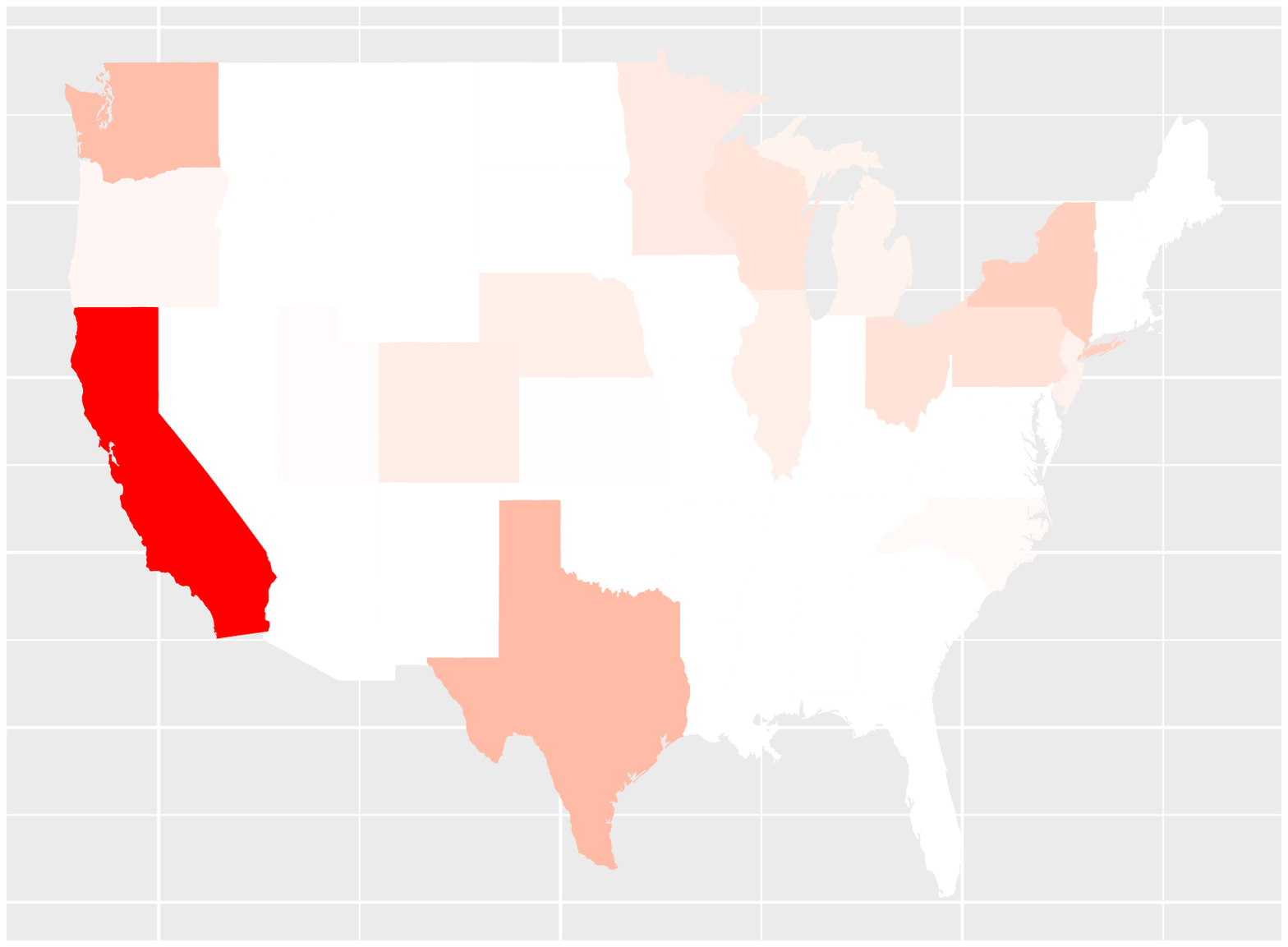} \\
      &(c.1) 2017 raw data  & (c.2) PoSSTend & (c.3) NMC-scan-stat & (c.4) YPS-SSD \\
      \end{tabular}
      \caption{
      hot-spots Detection result of mumps, syphilis and pertussis in 2017 by our proposed PoSSTenD method (second column), NMC-scan-stat (third column) and YPS-SSD method (last column).
      The first column is the raw data of the number of infected people of mumps, syphilis and pertussis in 2017. }
      \label{fig: case study hot-spots map}
    \end{figure}

%\vspace*{0.1 in}
\section*{Acknowledgments}
The authors are grateful to the Editor, the Associate Editor and two anonymous reviewers for their constructive comments that greatly improved the quality and presentation of this article.  
This project is partially supported by the Transdisciplinary Research Institute for Advancing Data Science (TRIAD), \url{http://triad.gatech.edu}, which is a part of the TRIPODS program at NSF and locates at Georgia Tech, enabled by the NSF grant CCF-1740776.
Huo is supported in part by NSF grant DMS-2015363.
Mei is supported in part by NSF grant DMS-2015405.

\newpage

\bibliographystyle{apalike}
\bibliography{referenceTensorPoisson}

\appendix
\section{Proof for Proposition \ref{prop: alg IRLS}}
\label{proof: IRLS}
\begin{proof}
  According to Newton-Raphson method, we update from $\mytheta_m^{(k+1)}$ to $\mytheta_m^{(k)}$ as follows:
  \begin{eqnarray}
  \label{equ: update theta m}
  \mytheta_m^{(k+1)}
  =
  \mytheta_m^{(k)}
  -
  \left[
    \frac{\partial^2}{\partial \mytheta_m \partial \mytheta_m^\top}
    F \left( \mytheta_m, \mytheta_h \right)
  \right]^{-1}
  \frac{\partial}{\partial \mytheta_m}
  F\left( \mytheta_m, \mytheta_h \right),
  \end{eqnarray}
  where we have
  \begin{eqnarray*}
  \frac{\partial}{\partial \mytheta_m}
  F\left( \mytheta_m, \mytheta_h \right)
  & = &
  \sum_{i=1}^{n}
  - y_i \myx_i
  +
  n_i e^{\left( \myx_i^\top \mytheta_m +  z_i^\top \mytheta_h \right)} \myx_i
  \triangleq
  -\myX^\top \myy + \myX^\top \mygamma \\
  %%%%%%%%%%%%%%%%%%%%
  \frac{\partial^2}{\partial \mytheta_m \partial \mytheta_m^\top}
  F\left( \mytheta_m, \mytheta_h \right)
  & = &
  \sum_{i=1}^{n}
   n_i e^{\left( \myx_i^\top \mytheta_m +  \myz_i^\top \mytheta_h \right)} \myx_i \myx_i^\top
  \triangleq
  \myX^\top \myW \myX \\
  %%%%%%%%%%%%%%%%%%%%%%
  \frac{\partial}{\partial \mytheta_h}
  F\left( \mytheta_m, \mytheta_h \right)
  & = &
  \sum_{i=1}^{n}
  - \frac{1}{N_i} y_i \myz_i + e^{\left( \myx_i^\top \mytheta_m +  \myz_i^\top \mytheta_h \right)} \myz_i
  \triangleq
  -\myZ^\top\myy + \myZ^\top \myeta,
  \end{eqnarray*}
  where the $i$-th entry of vector $\mygamma$ is
  $
    n_i e^{\left( \myx_i^\top \mytheta_m^{(k-1)} +  \myz_i^\top \mytheta_h^{(k-1)} \right)}
  $
  and matrix $\myW$ is a diagonal matrix whose $(i,i)$-th entry is
  $
    n_i e^{\left( \myx_i^\top \mytheta_m^{(k-1)} +  \myz_i^\top \mytheta_h^{(k-1)} \right)}.
  $
  So we can rewrite equation \eqref{equ: update theta m} as follows:
  \begin{eqnarray*}
    \mytheta_m^{(k+1)}
    & = &
    \mytheta_m^{(k)}
    -
    \left[
      \frac{\partial^2}{\partial \mytheta_m \partial \mytheta_m^\top}
      F\left( \mytheta_m, \mytheta_h \right)
    \right]^{-1}
    \frac{\partial}{\partial \mytheta_m}
    F\left( \mytheta_m, \mytheta_h \right) \\
    & = &
    \mytheta_m^{(k)}
    +
    \left( \myX^\top \myW \myX \right)^{-1} \myX^\top\left( \myy - \mygamma  \right) \\
    & = &
    \left( \myX^\top \myW \myX \right)^{-1}
    \left( \myX^\top \myW \myX \right)
    \mytheta_m^{(k)}
    +
    \left( \myX^\top \myW \myX \right)^{-1}
    \myX^\top
    \left( \myy - \mygamma  \right) \\
    & = &
    \left( \myX^\top \myW \myX \right)^{-1}
    \myX^\top \myW
    \left[
      \myX \mytheta_m^{(k)} + \myW^{-1} \left( \myy - \mygamma \right)
    \right] \\
    & \triangleq &
    \left( \myX^\top \myW \myX \right)^{-1}
    \myX^\top \myW
    \myeta,
  \end{eqnarray*}
  where $\myeta = \myX \mytheta_m^{(k)} + \myW^{-1} \left( \myy - \mygamma \right)$.
\end{proof}

\section{Estimation Results for the Logistic Model of Population Size in Section \ref{sec: simulation}}
\label{appendix: logistic model for population}
    The estimation of $\{\phi_{i,j,1}, \phi_{i,j,2}, \phi_{i,j,3} \}_{i = 1,\ldots, 49}$ is summarized in Table \ref{table: sim -- estimation of phi}.

    \begin{table}
      \centering
      \begin{adjustbox}{max width=0.95\textwidth}
      \centering
      \begin{tabular}{c|cccc|c|cccc}
        \hline
        &state&
        $\widehat\phi_{i,j,1}$ & $\widehat\phi_{i,j,2}$ & $\widehat\phi_{i,j,3}$ &
        &state&
        $\widehat\phi_{i,j,1}$ & $\widehat\phi_{i,j,2}$ & $\widehat\phi_{i,j,3}$ \\
        \hline
        $i=1$ &(Alabama& 80.6404   & -4.0135 &  66.0162& $i=26$ &Nebraska& 90.1254  & 168.2284 & 109.1338\\
        $i=2$ &(Arizona& 82.7195   & 2.6681 & 12.4089 & $i=27$ &Nevada&   35.6347 &   6.3429 &  11.1336  \\
        $i=3$ &Arkansas& 37.4245  & -15.4787 & 28.1058 & $i=28$&New Hampshire & 13.8539 &  -13.8127 &   10.6643   \\
        $i=4$ &California& 541.8391   & -8.4830 & 33.6786 & $i=29$&New Jersey & 92.5943 &  -21.7473 &   13.7113 \\
        $i=5$ &Colorado& 98.7504   & 17.0426 &  28.2669& $i=30$ &New Mexico& 24.5653 &  -12.2919 &   20.4325 \\
        $i=6$ &Connecticut& 38.0542  & -33.2342 & 20.1424 & $i=31$ &New York& 187.9958 &  -10.3970 &    3.8098  \\
        $i=7$ &Delaware& 18.8717   & 23.0528 & 40.3679 & $i=32$ &North Carolina& 209.7977  &  25.8509 &   34.7033  \\
        $i=8$ &District of Columbia& 84.8814  & 234.3312 & 85.0780 & $i=33$ &North Dakota& 61.2268 &  266.2989 &  120.0412  \\
        $i=9$ &Florida&  472.3076   & 32.7741 & 35.0153 & $i=34$ &Ohio&  120.8197 &  -61.9571 &   26.2627  \\
        $i=10$ &Georgia& 128.0025   &-1.0137  & 17.0737 & $i=35$&Oklahoma & 170.7639 &  136.1955  &  92.1323  \\
        $i=11$ &Idaho&   47.9991   & 46.7379 & 38.0065 & $i=36$ &Oregon& 80.6280  &  22.5818 &   43.4817  \\
        $i=12$ &Illinois&  130.1208  &-16.1579  &8.3610  & $i=37$ &Pennsylvania& 119.5912 & -156.0891 &    4.3750  \\
        $i=13$ &Indiana& 87.0352  & -26.3047 &42.7584  & $i=38$&Rhode Island & 10.6925 &  -17.1867  &   7.8247   \\
        $i=14$ &Iowa& 110.5464  & 168.3283 & 154.6580 & $i=39$ &South Carolina& 287.8320  & 116.8317 & 59.0815\\
        $i=15$ &Kansas& 35.6595  & -33.1733 & 38.3685 & $i=40$ &South Dakota& 63.0675  & 210.4712 &  100.9319  \\
        $i=16$ &Kentucky& 61.6822  & -21.2109 & 47.4401 & $i=41$ &Tennessee& 126.5228  &  18.9025  &  45.5198 \\
        $i=17$ &Louisiana&  138.2943  & 170.8833 & 212.9227 & $i=42$ &Texas&1296.6000  & 79.1413 & 42.5022\\
        $i=18$ &Maine& 13.9479 & -37.1852 & 19.3080 & $i=43$ &Utah& 87.5383 &   44.0237 &   32.6645 \\
        $i=19$ &Maryland& 107.4051  & 10.3277 & 56.6712 & $i=44$ &Vermont&  287.8320  & 116.8317 &59.0815 \\
        $i=20$ &Massachusetts& 172.8622   & 72.9025 & 113.4860 & $i=45$ &Virginia& 287.8320  & 116.8317 & 59.0815\\
        $i=21$ &Michigan& 100.0291   & -9.4310 & 3.8436 & $i=46$ &Washington& 72.4669   &  0.2911  &   6.1436   \\
        $i=22$ &Minnesota&  79.6624  & -10.0102 & 41.2455 & $i=47$ &West Virginia&  40.7546 & 388.1061 & 1747.3000\\
        $i=23$ &Mississippi& 30.6879  &-20.8025  & 12.1407 & $i=48$&Wisconsin &  63.7490 &  -29.2415 &   23.1078\\
        $i=24$ &Missouri& 68.6103 & -26.4275 & 23.9467 & $i=49$&Wyoming & 60.6138 &  173.6913  &  66.5942  \\
        $i=25$ &Montana& 65.8753  & 179.3000 & 92.5485 &  &  &  & \\
        \hline
      \end{tabular}
      \end{adjustbox}
      \caption{Estimation of $\{\phi_{i,j,1}, \phi_{i,j,2}, \phi_{i,j,3} \}_{i = 1,\ldots, 49}$ in Section \ref{sec: simulation}
      \label{table: sim -- estimation of phi}}
    \end{table}

\section{Table to Generate Figure \ref{fig: sim -- ARL1 plot}}
\label{appendix: ARL1 data}
    In this section, we present the table to generate Figure \ref{fig: sim -- ARL1 plot} in Table \ref{table: sim -- ARL1}.

    \begin{table}
    \caption{$\text{ARL}_1$ of our proposed PoSSTenD method, YPS-SSD method and ZQ-Lasso under population with decreasing trend
    \label{table: sim -- ARL1}}
    \centering
    \begin{adjustbox}{max width=0.95\textwidth}
    \centering
    \begin{threeparttable}
      \begin{tabular}{c|ccccccccc}
        \hline
        Fitting error    &
        $\delta = 0.05$  &
        $\delta = 0.075$ &
        $\delta = 0.1$   &
        $\delta = 0.125$ &
        $\delta = 0.15$  &
        $\delta = 0.175$ &
        $\delta = 0.2$    \\
        \hline
        &\multicolumn{7}{c}{population with increasing trend} \\
        \cline{2-8}
        PoSSTenD 
                 &2.0150  &1.3500& 1.0780 & 1.0140 & 1.0010 &1.0000& 1.0000\\
                 &(1.4610)&(0.6049)& (0.3000)&(0.1175) & (0.0316)&(0.0000)&(0.0000)\\
        YPS-SSD  &1.1806  & 1.0830& 1.0350 &1.0160 & 1.0080 & 1.0030&1.0010 \\
                 &(0.5138)&(0.3259)& (0.2093)&(0.1542) & (0.1093)&(0.0547)&(0.0316)\\
        ZQ-Lasso &9.0609  & 9.0000&  9.0620&9.0730 & 9.0670 & 9.0820& 9.0730\\
                 &(1.0593)&(1.1096)& (1.0580)&(1.0162) & (1.0704)&(1.0865)&(1.0406)\\
        DBS-PCA  &10.7515  &10.7360 & 10.7090 &10.7220 & 10.7340 & 10.7620& 10.7330\\
                 &(0.6330)&(0.7778)& (0.8666)&(0.7832) & (0.7442)&(0.8295)&(0.7939)\\
        \hline
        &\multicolumn{7}{c}{population with decreasing trend} \\
        \cline{2-8}
        PoSSTenD 
                 &4.2780  & 1.5580 & 1.0500 & 1.0000& 1.0000 & 1.0000& 1.0000\\
                 &(3.6105)&(1.1908)& (0.2357)&(0.0000) & (0.0000)&(0.0000)&(0.0000)\\
        YPS-SSD  &11.0000  & 10.9900 & 10.8300 & 9.6910& 6.2810 & 3.2390&1.5630 \\
                 &(0.0000)&(0.3162)& (1.2620)&(3.2690) & (4.8147)&(3.9708)&(2.1750)\\
        ZQ-Lasso &11.0000  & 11.0000 & 11.0000 & 10.9550& 9.6600 & 5.7020& 2.4990\\
                 &(0.0000)&(0.0000)& (0.0000)&(0.6383) & (3.3238)&(4.7103)&(2.9734)\\
        DBS-PCA  &10.0200  & 9.6900 &  9.2750 & 8.7720& 8.6040 & 8.5750& 8.6920\\
                 &(2.0995)&(2.6342)& (3.0684)&(3.4392) & (3.6128)&(3.6477)&(3.6631)\\
        \hline
      \end{tabular}
    \begin{tablenotes}
      \footnotesize
        \item[1] The above results are based on 1000 simulations
    \end{tablenotes}
    \end{threeparttable}
    \end{adjustbox}
    \end{table}

\bibliographystyle{abbrv}
\bibliography{referenceTensorPoisson}
%\bibstyle{tfs}
%\bibdata{referenceTensorPoisson}

\end{document}